\documentclass[nature,notitlepage,onecolumn,superscriptaddress,longbibliography]{revtex4-1}
\usepackage{amsmath,amsfonts}
\usepackage{todonotes}
\usepackage{graphicx}
\definecolor{darkgreen}{rgb}{0,0.5,0}
\definecolor{shadecolor}{rgb}{1,0.8,0.3}
\definecolor{shad2}{gray}{.4}
\definecolor{lightgray}{RGB}{230, 230, 230}
\definecolor{darkred}  {rgb}{0.5,0,0}
\definecolor{darkblue} {rgb}{0,0,0.5}
\usepackage{hyperref}
\hypersetup{colorlinks, breaklinks, linkcolor=darkred, urlcolor=darkblue, citecolor=darkblue}

\newtheorem{theorem}{Theorem}

\newtheorem{corollary}{Corollary}

\newtheorem{lemma}{Lemma}

\newtheorem{proposition}{Proposition}

\newenvironment{proof}[1][Proof]{\noindent\textbf{#1.} }{\ \rule{0.5em}{0.5em}}

\def\sq{\operatorname{sq}}

\def\>{\rangle}
\def\<{\langle}
\def\({\left(}
\def\){\right)}

\newcommand{\ket}[1]{\left|{#1}\right\rangle}
\newcommand{\bra}[1]{\left\langle{#1}\right|}

\newcommand{\mc}[1]{\mathcal{#1}}

\begin{document}
\title{Linear programs for entanglement and key distribution in the quantum internet}

\author{Stefan B{\"a}uml}
\email{stefan.baeuml@icfo.eu}
\affiliation{QuTech, Delft University of Technology, Lorentzweg 1, 2628 CJ Delft, Netherlands}
\affiliation{ICFO-Institut de Ciencies Fotoniques, The Barcelona Institute of Science and Technology, Av. Carl Friedrich Gauss 3, 08860 Castelldefels (Barcelona), Spain}
\affiliation{NTT Basic Research Laboratories, NTT Corporation, 3-1 Morinosato-Wakamiya, Atsugi, Kanagawa 243-0198, Japan}
\affiliation{NTT Research Center for Theoretical Quantum Physics, NTT Corporation, 3-1 Morinosato-Wakamiya, Atsugi 243-0198, Japan}	
\author{Koji Azuma}
%\email{koji.azuma.ez@hco.ntt.co.jp}
\affiliation{NTT Basic Research Laboratories, NTT Corporation, 3-1 Morinosato-Wakamiya, Atsugi, Kanagawa 243-0198, Japan}
\affiliation{NTT Research Center for Theoretical Quantum Physics, NTT Corporation, 3-1 Morinosato-Wakamiya, Atsugi 243-0198, Japan}	
\author{Go Kato}
%\email{go.kato.gm@hco.ntt.co.jp}
\affiliation{NTT Communication Science Laboratories, NTT Corporation, 3-1 Morinosato-Wakamiya, Atsugi 243-0198, Japan}
\affiliation{NTT Research Center for Theoretical Quantum Physics, NTT Corporation, 3-1 Morinosato-Wakamiya, Atsugi 243-0198, Japan}	
\author{David Elkouss}
%\email{d.elkousscoronas@tudelft.nl}
\affiliation{QuTech, Delft University of Technology, Lorentzweg 1, 2628 CJ Delft, Netherlands}

\date{\today}
\begin{abstract}
Quantum networks will allow to implement communication tasks beyond the reach of their classical counterparts. A pressing and necessary issue for the design of quantum network protocols is the quantification of the rates at which these tasks can be performed. Here, we propose a simple recipe that yields efficiently computable lower and upper bounds on the maximum achievable rates. For this we make use of the max-flow min-cut theorem and its generalization to multi-commodity flows to obtain linear programs. We exemplify our recipe deriving the linear programs for bipartite settings, settings where multiple pairs of users obtain entanglement in parallel as well as multipartite settings, covering almost all known situations.  We also make use of a generalization of the concept of paths between user pairs in a network to Steiner trees spanning a group of users wishing to establish Greenberger-Horne-Zeilinger states.
\end{abstract}

\maketitle
\tableofcontents

\section{Introduction}
Quantum entanglement allows for the implementation of communication tasks not possible by classical means. The most prominent examples are quantum key distribution and quantum teleportation between two parties \cite{EK91,BBM92,bennett1993teleporting}, but there is a host of other tasks also involving more then two parties \cite{wehner2018quantum}. An example of a protocol using multipartite entanglement is quantum conference key agreement \cite{augusiak2009multipartite}, where multiple parties who trust each other need to establish a common key. Another example is 
quantum secret sharing \cite{hillery1999quantum}, where multiple parties who do not trust each other wish to encrypt a message in such a way that it can only be decrypted if all parties cooperate. Multipartite entanglement can also be used for the synchronization of a network of clocks \cite{komar2014quantum} and plays an important role in quantum computing \cite{raussendorf2003measurement}. Quantum networks allow for the distribution of entanglement as a resource for such tasks among parties, which could, in principle, be spread out across different continents in an efficient manner. Whereas small-scale quantum networks can be designed in such a way that they perform optimally in distributing a particular resource to a particular set of users, a future quantum version of the internet will most likely grow to have a complex structure and involve a number of user pairs, or groups, requiring entangled resources for different tasks in parallel. 

Recently, in light of the experimental promise of short-term quantum network deployment, the community has begun to devote attention to communication problems for networks of noisy quantum channels and their general structures. Arguably, the most important one is the computation of the maximum rates at which the different tasks can be performed. Given that, even in the case of point-to-point links, entanglement makes the characterization of capacities notably more complicated than its classical counterpart, with phenomena such as superactivation \cite{smith2008quantum}, it was unclear how much it would be possible to borrow from the theory of classical networks. Besides, the usage of a quantum channel is much more expensive than that of its classical counterpart. This motivates the introduction of different capacities which account for resources in different ways. The results of \cite{pirandola2016capacities,pirandola2019capacities} introduced the quantum problem and successfully established upper and lower bounds on a capacity of a quantum network which quantifies the maximum size of bipartite maximally entangled states (for quantum teleportation) or private states (for quantum key distribution) per network use, as a generalization of the fundamental/established notion of classical network capacity \cite{el2011network}. These upper and lower bounds coincide when the network is composed only of a very relevant class of quantum channels, called distillable channels. The results of \cite{AML16,AK16,rigovacca2018versatile} derive analogous bounds, alternatively defining the capacity of a quantum network per total number of channel uses (related with a cost) or per time, rather than per network use, for generality. In any case, rather surprisingly, the series of fundamental works \cite{pirandola2016capacities,pirandola2019capacities,AML16,AK16,rigovacca2018versatile} have shown that these capacity of quantum networks for bipartite communication behave similar to that of classical networks. The distribution of bi- and multipartite entanglement in quantum networks has been in considered in a number of other works, including \cite{van2013designing,van2013path,epping2016quantum,epping2016large,epping2016quantum,wallnofer20162d,hahn2019quantum,chakraborty2019distributed}. These works differ from \cite{pirandola2016capacities,pirandola2019capacities,AML16,AK16,rigovacca2018versatile} in that they are not concerned with networks of general noisy channels.

Namely, given a network of quantum noisy channels and bounds on their capacities satisfying certain properties, one can conceptually construct a classical version of the quantum network where each quantum channel is replaced by a perfect classical channel with a capacity given by the bound on the quantum channel capacity. Then, by considering cuts between two nodes in the induced `classical' network, it is possible to obtain upper and lower bounds on a capacity of the network for distributing private keys or entanglement between two clients. The same techniques have found application for many user pairs \cite{pirandola2016capacities,pirandola2019bounds,bauml2017fundamental} and for the distribution of multipartite entanglement among multiple users \cite{bauml2017fundamental,yamasaki2017graph}. While the early work has laid down extremely useful techniques to characterize quantum network capacities, it has either not focused on their computation \cite{AML16,bauml2017fundamental} or left open the computability of several of the scenarios considered \cite{pirandola2019capacities}. However, this is rather important in practice, in the sense that the quantum network will be required to serve entanglement resources quickly according to the requests of clients, and, in so doing, efficient estimation of the quantum network capacities is a necessary basis for choosing a proper subnetwork to accomplish that. The goal of this paper is to provide a simple recipe to find such efficiently computable bounds for quantum network capacities. 

In this paper, using the approach taken in \cite{AML16,AK16,bauml2017fundamental,rigovacca2018versatile}, i.e., defining a network capacity as a rate per the total number of channel uses or per time, we introduce or generalize the capacities for private or quantum communication in the following scenarios: bipartite communication, concurrent communication between multiple user pairs with the objective of (1) maximizing the sum of rates achieved by the user pairs or (2) maximizing the worst-case rate that can be achieved by any pair, as well as multipartite state sharing where the goal is either to distribute Greenberger-Horne-Zeilinger (GHZ) or multipartite private states \cite{augusiak2009multipartite} for a group of network users. We then provide linear-program lower and upper bounds on the all these capacities. The size of the linear programs scales polynomially in the parameters of the network, making it computable in polynomial time by interior point algorithms \cite{ye1991n3l}. A central tool deriving upper bounds in the case of multiple user pairs are approximate min-cut max-flow theorems for multi-commodity flows \cite{aumann1998log,gunluk2007new,garg1996approximate}. Up to a factor of the logarithmic order of the number of user pairs, these results link quantities that occur in the known upper bounds \cite{bauml2017fundamental}, such as the minimum cut ratio (i.e. the smallest ratio of the capacity of a cut and the demand across the cut) and the minimum capacity multicut (i.e. the smallest capacity set of edges whose removal disconnects all user pairs), both of which are NP-hard problems to calculate in general graphs \cite{garg1997primal,aumann1998log}, to multi-commodity flow maximizations that can be computed by linear programs (LPs). A challenge we address in this work is to find protocols that can achieve the upper bounds. In the bipartite case, protocols involving distillation of Bell pairs across all edges of a network, and entanglement swapping along paths have been used to provide lower bounds on the network capacities \cite{pirandola2016capacities,pirandola2019capacities,AK16}. Using such simple routing methods, it was shown in \cite{pirandola2016capacities,pirandola2019capacities,AK16} that the bipartite upper bounds can be achieved for networks consisting of a wide class of channels, known as distillable channels \cite{PLOB17}, which include erasure channels, dephasing channels, bosonic quantum amplifier channels and lossy optical channels. Here, we extend the bipartite protocol presented in \cite{AK16} to the case of many user pairs and to the distribution of GHZ states among a set of users. We do so by considering edge-disjoint Steiner trees spanning the set of users.

\begin{figure}
\centering
\includegraphics[width=1\textwidth]{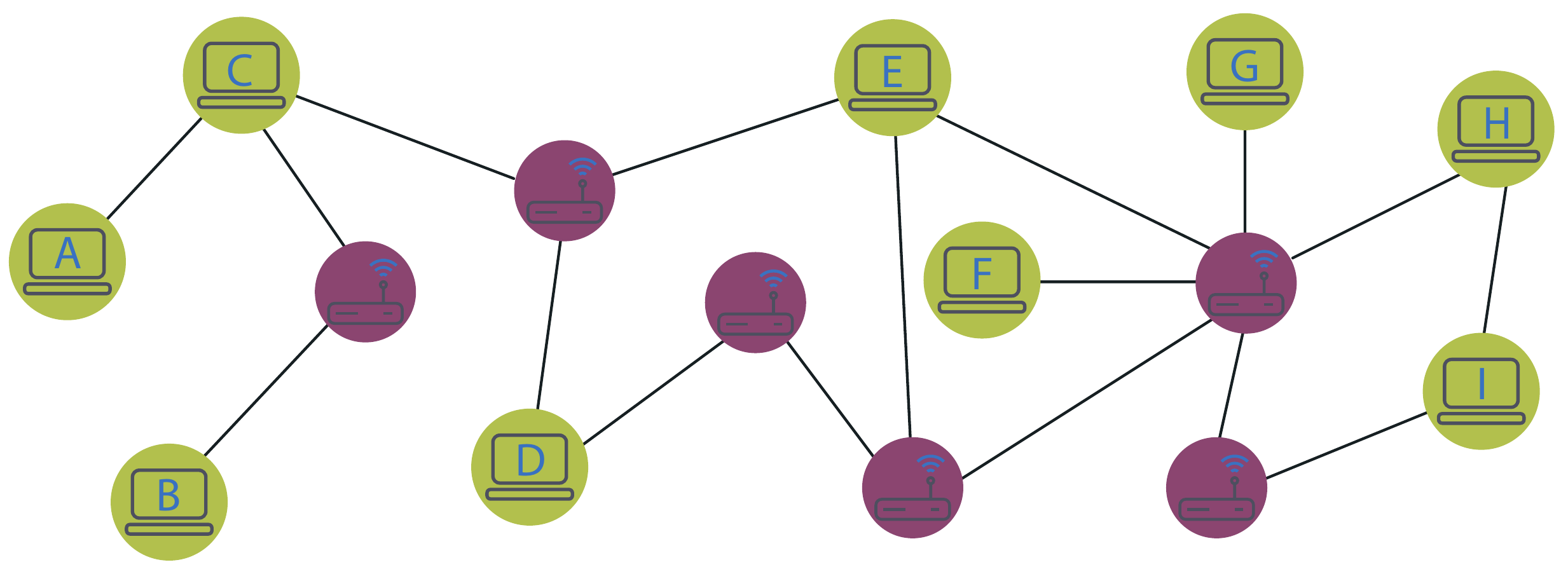}     
\caption{Example of a quantum network: It consists of quantum channels, repeater stations and end users $A,B,\ldots,I$. In such a network there are many possible communication tasks. Some examples are the distribution of private states, Bell states and Greenberger-Horne-Zeilinger (GHZ) states. The first two are bipartite tasks. We study the implementation of these tasks between a single pair of users, for instance $A$ and $I$, and between multi-pairs of users in parallel, for instance $A$ and $I$, $C$ and $D$ and $F$ and $H$. The last task, the distribution of GHZ states, is a multipartite user scenario, for instance $A,B,E,G,I$ could distill a five-partite GHZ state.  }\label{fig:net}
\end{figure}

\section{Results}
\label{sec:summary}

Our base setup is a network of nodes connected by noisy quantum channels (see Fig.~\ref{fig:net}). The nodes act either as end users or as repeater stations and have the ability to store and process quantum information locally. In addition, all nodes are connected by classical lines of communication, which can be used freely. 

We are interested in the possibilities and limitations of quantum networks for different communication tasks and usage scenarios. Fortunately, most tasks of interest can be rephrased as the distribution of an entangled target state among users of the quantum network \cite{wilde2013quantum}. Here, we consider the distribution of a bipartite entangled target state between a pair of users, of multiple bipartite entangled target states between multiple pairs of users in parallel as well as of a multipartite entangled target state among a group consisting of more than two users. The distribution of these states is known to be equivalent to the problems of quantum information transmission, private classical communication, quantum key distribution and quantum conference key agreement among others.

As we are interested in emergent, organically grown quantum networks, such as the classical internet, we do not make any assumptions on the structure of the network except that it can be described  by a finite directed graph. Let the quantum network be given by the directed graph $G=(V,E)$, where $V$ denotes the set of the finite vertices and $E$ the set of the finite directed edges, which represent quantum channels. Each directed edge $e\in E$ has tail $v\in V$ and head $w\in V$.  We also denote $e$ by $vw$. We can also assign a nonnegative edge capacity \footnote{In order to avoid confusion, we will use the terms `edge capacity' when referring to edges and `capacity', when referring to quantum or private capacities of a channel or the entire network.} $c(e)$ to every edge. Each edge $vw$ corresponds to a channel $\mc{N}^e=\mc{N}^{vw}$ with input in $v$ and output in $w$. We assume that each vertex has the capability to store and process quantum information locally and that all vertices are connected by public lines of classical communication, the use of both of which is considered to be a free resource. Let us assume there is a subset $U\subset V$ of the vertices, the users who wish to establish a target state containing the desired resource, whereas the remaining vertices serve as repeater stations. In the following section we will elaborate on the exact form of the target state.

We assume that initially there is no entanglement between any of the vertices. A target state can be distributed by means of an adaptive protocol, consisting of local operations and classical communication (LOCC) among the vertices in the network interleaved by channel uses \cite{pirandola2016capacities,pirandola2019capacities,AML16}.  In this work we are not concerned with the inner workings of the protocol but describe a protocol only by the total number of channel uses, and usage frequencies of each channel.  We describe a protocol as follows:  Given upper bounds $n_e$ on the average of the number of uses of each channel ${\cal N}^e$, we define a set of usage frequencies $\{p_e\}_{e\in E}$ of each channel ${\cal N}^e$ as $p_e := n_e/ n (\ge 0)$. Here $n$ can be regarded as time or $n$ with $\sum_{e \in E}p_e=1$ can be considered to be an upper bound on the average of total channel uses  (see \cite{AK16}). Further, we introduce an error parameter $\epsilon$ such that after the final round of LOCC a state $\epsilon$-close in trace distance to the target state is obtained. Depending on the user scenario, the target state can be a maximally entangled state, a tensor product of multiple maximally entangled states between multiple pairs of users or a GHZ state. By average we mean that parameters of a protocol are averaged over all possible LOCC outcomes. We call such a protocol an $(n,\epsilon, \{p_e\}_{e\in E})$ adaptive protocol. In the asymptotic limit where $n\to\infty$ it then holds ${n}_e\to\infty$ for edge $e$ with $p_e>0$ while $\{p_e\}_{e\in E}$ remains fixed \cite{AK16}.

Note that whereas quantum channels are directed, the direction does not play a role when we use them to distribute entanglement under the free use of (two-way) classical communication. For example, once a channel has been used to distribute a Bell state, which is invariant under permutations of nodes across the channel. This motivates the introduction of an undirected graph $G'=(V,E')$, where $E'$ is obtained from $E$ as follows:
For any edge $vw\in E$ with $wv \in E$, the directed edges $vw$ and $wv$ are replaced by single undirected edge $\{vw\}$ (or, equivalently $\{wv\}$) with $c'(\{vw\})=c(vw)+c(wv)$, while, for any edge $vw\in E$ with $wv \notin E$, the directed edge $vw$ is replaced by undirected edge $\{vw\}$ with $c'(\{vw\})=c(vw)$. For more details about our notations see Supplementary Note 1.

%we append zero capacity for the opposite edge $wv$, that is, $c(wv)=0$. Then, for either $vw\in E$, we replace $vw$ by an undirected edge $\{vw\}$ (or, equivalently $\{wv\}$) and define $c'(\{vw\})=c(vw)+c(wv)$.

Let us also note that whereas it is common from a quantum information theory point of view to allow for free LOCC operations, there are practical challenges to implement quantum memories with long storage times. By a slight abuse of our notation, however, it is possible to include such effects into our scenario, as well. Namely one could divide a vertex into a pre- and post storage vertex and add an additional noisy channel describing the noisy quantum memory (for instance, see \cite{AML16}).

\subsection{Bipartite user scenario}\label{sec:bip}
In this section we obtain linear-program upper and lower bounds on the entanglement and key generation capacities of a network for bipartite scenarios. While some of the discussion have been made implicitly in earlier results \cite{pirandola2016capacities,pirandola2019capacities,AK16}, it is worth giving an explicit formulation here, given its relevance. %While this can be regarded as implicit in earlier results \cite{pirandola2016capacities,pirandola2019capacities,AK16}, we believe that given its relevance it is worth giving an explicit formulation. 
It will also serve as a good starting point to demonstrate our method and introduce some notation. Let us suppose that the set of users only contains two vertices, $s\in E$, a.k.a Alice, and $t\in E$, also known as Bob. A possible target state could be a maximally entangled state  $\ket{\Phi^d}_{M_sM_t}=\frac{1}{\sqrt{d}}\sum_{i=1}^d\ket{ii}_{M_sM_t}$ with $\log d$ ebits. We also use the notation $\Phi^d_{M_sM_t}=\ket{\Phi^d}\bra{\Phi^d}_{M_sM_t}$. In the case of $d=2$, this state is called a Bell state. The target state could also be a general private state \cite{horodecki2005secure,horodecki2009general}, which is of the form $\gamma^d_{K_sK_tS_sS_t}=U^{\text{twist}}|\Phi^d\>\<\Phi^d|_{K_sK_t}\otimes\sigma_{S_sS_t}U^{\text{twist}\dagger}$, where $\sigma_{S_sS_t}$ is an arbitrary state and $U^{\text{twist}}=\sum_{ik}|ik\>\<ik|_{K_sK_t}\otimes U^{(ik)}_{S_sS_t}$ is a controlled unitary that `twists' the entanglement in the subsystem $K_sK_t$ to a more involved form also including the subsystem $S_sS_t$. It has been shown that, by measuring the `key part' $K_sK_t$, while keeping the `shield part' $S_sS_t$ away from an eavesdropper Eve, $\log d$ bits of a private key can be obtained. The number of ebits or private bits is treated as the figure of merit.

We can now define a quantum network capacity ${\cal Q}_{ \{p_e \}_{e \in E} } (G, \{ {\cal N}^e \}_{e \in E} )$ per time [${\cal Q} (G, \{ {\cal N}^e \}_{e \in E} )$ per total channel use] as the largest rate $\langle\log d^{(k)}\rangle_{k}/n$ achievable by an adaptive $(n,\epsilon, \{p_e\}_{e\in E})$ protocol such that after $n$ uses the finally obtained state $\rho_{M_sM_t}^{(n,k)}$ is $\epsilon$-close to $\Phi^{d^{(k)}}_{M_sM_t}$, in the limit $n\to\infty$ and $\epsilon\to0$ [maximized over all user frequencies $p_e\geq0$ such that $\sum_ep_e=1$]. Here $k$ is a vector keeping the track of outcomes of the LOCC rounds and the notation $\langle\cdots\rangle_{k}$ corresponds to averaging over all LOCC outcomes. Similarly, we define a private network capacity ${\cal P}_{ \{p_e \}_{e \in E} } (G, \{ {\cal N}^e \}_{e \in E} )$ per time [${\cal P} (G, \{ {\cal N}^e \}_{e \in E} )$ per total channel use] as the largest rate $\langle\log d^{(k)}\rangle_{k}/n$ achievable by an adaptive $(n,\epsilon, \{p_e\}_{e\in E})$ protocol such that after $n$ uses the state $\rho_{K_sK_tS_sS_t}^{(n,k)}$ is $\epsilon$-close to $\gamma^{d^{(k)}}_{K_sK_tS_sS_t}$, in the limit $n\to\infty$ and $\epsilon\to0$ [maximized over all user frequencies $p_e\geq0$ such that $\sum_ep_e=1$].

As the class of private states includes maximally entangled states, the private capacity is an upper bound on the quantum capacity \cite{horodecki2005secure,horodecki2009general}. Our main results in this section will be efficiently computable upper and lower bounds on the private and quantum capacities, respectively. 

In a number of recent works upper bounds on private network capacities have been obtained \cite{AML16,pirandola2019capacities,rigovacca2018versatile}. The main idea behind those results is to assign nonnegative edge capacities to each edge $e$ and to find the minimum edge capacity cut between $s$ and $t$. By cut between $s$ and $t$ we mean a set of the edges whose removal disconnects $s$ and $t$. The edge capacity of a cut can be defined as the sum of edge capacities of the edges in the cut. For details see the Methods section. If the edge capacity $c(e)$ of an edge $e$ is given by the usage frequency $p_e$ of channel $\mathcal{N}^e$, multiplied by an entanglement measure $\mathcal{E}(\mathcal{N}^e)$ upper bounding the private capacity of $\mathcal{N}^e$, which is continuous near the target state and cannot be increased by amortization (see properties P1 and P2 of Ref. \cite{rigovacca2018versatile} or Supplementary Note 2), the minimum edge capacity cut provides an upper bound on the private network capacity. Examples of suitable quantities $\mathcal{E}(\mathcal{N}^e)$ include the squashed entanglement $E_{\sq}$ \cite{christandl2004squashed}, the max-relative entropy of entanglement $E_{\max}$ \cite{datta2009min} and, for a particular class of so-called Choi-stretchable channels/teleportation-simulable \cite{BDSW96,HHH99,Mul12,PLOB17}, the relative entropy of entanglement $E_R$ \cite{vedral1997quantifying} of the channel. If we know such quantities for all channels constituting the network, all that is left to do is finding the minimum edge capacity cut, which is a well known problem in graph theory. However, it is not necessarily efficient to solve this optimization directly, because there is a case where we need to maximize further such a minimized edge capacity. 
For instance, it is not clear a priori how to maximize over channel frequencies the minimum edge capacity to find the capacity of the network per channel use. To tackle this issue we resort to the duality of the problem.

In particular, using the max-flow min-cut theorem \cite{EFS96,ford1956maximal}, we rephrase the problem of finding the minimum edge capacity cut as a network flow maximization problem in the undirected graph $G'$. Thanks to this, it becomes sufficient for us to consider maximization only, in every case. In a network flow maximization problem in an undirected graph, the idea is to assign a variable $f_{vw}$ and $f_{wv}$ to each undirected edge $\{vw\}$ which can take non-negative values. $f_{vw}$ is interpreted as an abstract flow of some commodity from vertex $v$ to vertex $w$. As such, it has to fulfill the following constraint: Interpreting the edge capacity $c'(\{vw\})$ of an undirected edge $\{vw\}$ as the capacity of its edge to transmit an abstract commodity, we require that the sum of edge flows $f_{vw}$ and $f_{wv}$ does not exceed the edge capacity $c'(\{vw\})$. We call this the edge capacity constraint. 

Having defined a flow of an abstract commodity through an edge, the obvious next step is to consider a flow through the entire network. Namely, we mark two vertices, the source $s$ and the sink $t$ and define a flow from $s$ to $t$ as the sum of all `outgoing' flows $f_{sv}$, where $v$ is a vertex adjacent to $s$, such that for every edge the edge capacity constraint is fulfilled and that for every vertex $w\notin\{s,t\}$ the sum over $v$ of `incoming' edge flows $f_{vw}$ is equal to the sum over $v$ of `outgoing' flows $f_{wv}$, where $v$ are the vertices adjacent to $w$, which is known as flow conservation constraint. If the graph is undirected, the roles of the source and the sink can be exchanged, without changing the value of the flow. As both the edge capacity and the flow conservation constraint are linear, the maximization of the flow from $s$ to $t$ can be efficiently computed by means of linear programming \cite{murty1983linear}. The max-flow min-cut theorem now states that the minimum edge capacity cut that separates $s$ and $t$ is equal to the maximum flow from $s$ to $t$. Figure \ref{fig:unicast} illustrates this with an example. For detailed definitions of cuts, flows and the max-flow min-cut theorem see the Methods section. 

\begin{figure}
\centering
\includegraphics[width=\textwidth, trim={3cm 23cm 3cm 3cm},clip=true]{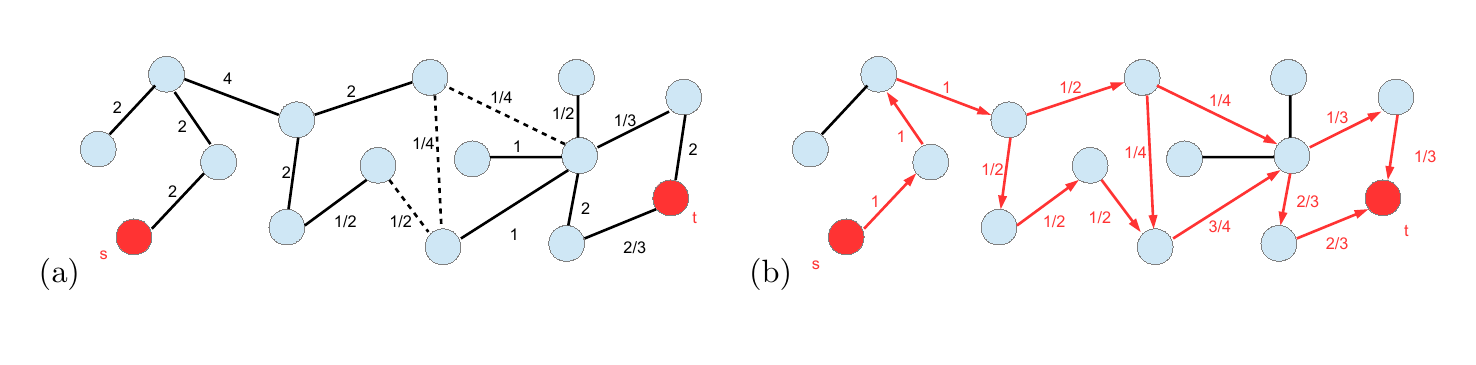}
%(a)\includegraphics[width=0.4\textwidth, trim={2cm 10cm 2cm 2cm},clip=true]{Unicast.pdf}     
%(b)\includegraphics[width=0.4\textwidth, trim={2cm 10cm 2cm 2cm},clip=true]{Unicast_flow.pdf}
\caption{Example of the max-flow min-cut theorem (a) Min-cut: Example of an undirected graph with a source-sink pair (in red). The labels of the edges denote their edge capacities. The edges in dashed lines represent a source-sink cut, i.e. their removal completely disconnects the source from the sink. The capacity of the cut is given by the sum over the edge capacities in the cut, in this case equal to $1$, which is the minimum capacity of all source-sink cuts in this network. In other words, the min-cut is equal to $1$. Note that the minimizing cut is not unique. (b) Max-flow: By the max-flow min-cut theorem the min-cut is equal to the maximum flow from the source to the sink. Here we have provided an example of a flow from the source to the sink. The labels denote the directed edge flows $f_e$. The flow from the source to the sink achieves the min-cut value of $1$.}
\label{fig:unicast}
\end{figure}

Precisely, we use the max-flow min-cut theorem to transform the min-cut upper bounds on the private network capacity given in \cite{rigovacca2018versatile} into an efficiently computable linear program. To do so we define directed edge capacities $c(e)=p_{e}\mathcal{E}(\mathcal{N}^{e})$, for every directed edge $e\in E$. Thus, entanglement takes the role of the abstract commodity considered above. 

The interpretation of entanglement as a commodity raises the question if there exists a protocol that can distribute entanglement in a way similar to the flow of a commodity through a network. Ideally, one could construct such a protocol using the edge flows obtained in the flow maximization. To some extend this can be achieved by a quantum routing protocol, such as the aggregated repeater protocol introduced in \cite{AK16}. The aggregated repeater protocol consists of two steps: First each channel is used to distribute Bell states at a rate $p_eR^\leftrightarrow(\mathcal{N}^e)$ such that $R^\leftrightarrow(\mathcal{N}^e)$ reaches the quantum capacity $Q^\leftrightarrow(\mathcal{N}^e)$ in the asymptotic limit of many channel uses. This results in a network of Bell states, which can be described by an indirected multigraph, where each edge corresponds to one Bell pair. The second step of the protocol is to find edge-disjoint paths from Alice to Bob and connect them by means of entanglement swapping. The number of Bell pairs is thus equal to the number of edge disjoint paths between Alice and Bob in the Bell network. Hence, in order to obtain a lower bound on the capacity, one would have to find the number of edge disjoint paths in the the multigraph corresponding to the Bell network.

Finding the maximum number of edge-disjoint paths between $s$ and $t$ in a multigraph is the same as maximizing the flow in a graph where the edge capacities are given by the number of parallel edges in the multigraph, however with the additional constraint that each edge flow takes integer values \cite{nishizeki1990planar}. Such an integer flow maximization can no longer be formulated as a LP. Physically, the integer constraint corresponds to the fact that there is no such thing as `half a Bell state'. Hence, if a LP provides an edge flow of 0.5 for some edge and 1 for another, we cannot translate this into a protocol distributing half a Bell sate over the first edge and one Bell state over the other. What we can do, however, is to multiply all edge flows obtained in the optimization by a factor of $2$, and distribute one Bell state along the edge where we have obtained flow $0.5$ and two Bell states along the edge where we have obtained $1$. For rational edge flows obtained in the optimization and finite graphs, we can always find a large enough number to multiply the edge flows with to obtain integer values associated with each edge that can be translated into a number of Bell pairs distributed along this edge. If the edge flows obtained are real numbers, we can approximate them by rational numbers with arbitrary accuracy. This allows us to compute lower bounds on the quantum network capacities by means of maximizing the flow over a network with edge capacities given by $c(e)=p_{e}{Q^\leftrightarrow}(\mathcal{N}^{e})$, providing us with a lower bounds that can be efficiently computed by linear programming. Finally, we can include an optimization over usage frequencies $p_e$ into both LPs, providing us with:

\begin{theorem}\label{Theo_bip}
For a network described by a finite directed graph $G$ and an undirected graph $G'$ as defined above, the private and quantum network capacities per total channel use, $\mathcal{P}\left(G,\{\mc{N}^e\}_{e\in E}\right)$ and $\mathcal{Q}\left(G,\{\mc{N}^e\}_{e\in E}\right)$, satisfy
\begin{equation}
\bar{f}^{s\to t}_{\max}(G',\{Q^\leftrightarrow(\mc{N}^e)\}_{e\in E})\leq\mathcal{Q}\left(G,\{\mc{N}^e\}_{e\in E}\right)\leq\mathcal{P}\left(G,\{\mc{N}^e\}_{e\in E}\right)\leq \bar{f}^{s\to t}_{\max}(G',\{\mathcal{E}(\mc{N}^e)\}_{e\in E}),
\end{equation}
where $\bar{f}^{s\to t}_{\max}$ is given by the linear program Eq. (\ref{LP1bar}) in the Methods section.

Further, $\cal E$ can be chosen to be the squashed entanglement $E_{\sq}$, the max-relative entropy of entanglement $E_{\max}$ and, for a Choi-stretchable channels, the relative entropy of entanglement $E_R$. 
\end{theorem}
For the proof see Supplementary Note 2. As described in the Methods section, the LPs scale polynomially with the size of the network.

Note that for a subset of Choi-stretchable channels, known as distillable channels, which include erasure channels, dephasing channels, bosonic quantum amplifier channels and lossy optical channels, the relative entropy of entanglement of the channel $\mc{N}^e$ (and its Choi state $\sigma^e$) is equal to the two-way classical assisted quantum capacity \cite{PLOB17}, $E_R(\mc{N}^e)=E_R(\sigma^e)=Q^\leftrightarrow(\mc{N}^e)$. Hence the bounds in Theorem \ref{Theo_bip} become tight.

\subsection{Multiple pairs of users}\label{sec:multiuni}

We now move on to the scenario of multiple pairs of users $(s_1,t_1), \cdots  ,(s_r,t_r)$ who wish to establish maximally entangled states or private states concurrently, i.e. we have  target states of the form $\bigotimes_{i=1}^r\Phi^{d_i}_{M_{s_i}M_{t_i}}$ or
$\bigotimes_{i=1}^r\gamma^{d_i}_{K_{s_i}K_{t_i}S_{s_i}S_{t_i}}$. This would be a typical scenario in a future `quantum internet', where a number of user pairs might wish to perform QKD in parallel. In contrast to the bipartite scenario discussed in the previous section, where the goal is to simply optimize the rate at which entanglement is distributed between a user pair, there are a number of different figures of merit in the multi-pair scenario. We define the following three figures of merit: (1) a total multi-pair quantum (private) network capacity ${\cal Q}^{\rm total}(G, \{ {\cal N}^e \}_{e \in E} )$ per total channel use [${\cal Q}^{\rm total}_{\{p_e\}_{e \in E}} (G, \{ {\cal N}^e \}_{e \in E} )$ per time] (${\cal P}^{\rm total}(G, \{ {\cal N}^e \}_{e \in E} )$ per total channel use [${\cal P}^{\rm total}_{\{p_e\}_{e \in E}} (G, \{ {\cal N}^e \}_{e \in E} )$ per time]), defined as the largest sum, over all user pairs, of the entanglement distribution rates achievable by an adaptive $(n,\epsilon, \{p_e\}_{e\in E})$ protocol such that after $n$ uses we are $\epsilon$-close to the target state, again taking the limit $n\to\infty$ and $\epsilon\to0$ [and maximizing over all user frequencies $p_e\geq0$ such that $\sum_ep_e=1$]. Whereas maximizing the sum of rates is a good approach when the goal is to distribute as much entanglement as possible, it has the drawback that the protocol can be unfair in the sense that some pairs might get more entanglement than others, while some might not get anything at all. This drawback can be overcome by using our second figure of merit: (2) a worst-case multi-pair quantum (private) network capacity ${\cal Q}^{\rm worst}(G, \{ {\cal N}^e \}_{e \in E} )$ per total channel use [${\cal Q}^{\rm worst}_{\{p_e\}_{e \in E}} (G, \{ {\cal N}^e \}_{e \in E} )$ per time] (${\cal P}^{\rm worst}(G, \{ {\cal N}^e \}_{e \in E} )$ per total channel use [${\cal P}^{\rm worst}_{\{p_e\}_{e \in E}} (G, \{ {\cal N}^e \}_{e \in E} )$ per time]), i.e. the least entanglement distribution rate that can be achieved by any user pair [by maximizing over all user frequencies $p_e \ge 0$ with $\sum_e p_e =1$]. This approach is good in a scenario where the goal is to distribute entanglement in a fair way, in the sense that the amount of entanglement that each user pair obtains is maximized. Finally we consider (3) the case where we assign weight $q_i$ to each user pair $(s_i,t_i)$. This approach can be used if user pairs are given different priorities. We call the  corresponding figure of merit weighted multi-pair quantum (private) network capacity ${\cal Q}^{q_1,\cdots, q_r} (G, \{ {\cal N}^e \}_{e \in E} )$ per total channel use [${\cal Q}^{q_1,\cdots, q_r}_{\{p_e\}_{e \in E}} (G, \{ {\cal N}^e \}_{e \in E} )$ per time] (${\cal P}^{q_1,\cdots, q_r} (G, \{ {\cal N}^e \}_{e \in E} )$ per total channel use [${\cal P}^{q_1,\cdots, q_r}_{\{p_e\}_{e \in E}} (G, \{ {\cal N}^e \}_{e \in E} )$ per time]) and define it as the largest achievable weighted sum of rates [with maximization over all user frequencies $p_e \ge 0$ with $\sum_e p_e =1$]. We will now present our results for the total and worst-case scenario. For bounds on the weighted multi-pair network capacities see Supplementary Note 3.

Let us begin with scenario (1). As in the bipartite case, we can assign edge capacity $c(e)=p_eE_{\sq}(\mathcal{N}^e)$ to each edge $e$ in the graph corresponding to the network. From \cite{bauml2017fundamental} we can obtain upper bounds on the total multi-pair private network capacity which are given the minimum capacity multicut. A multicut is defined as a set of edges whose removal disconnects all pairs. The capacity of a multicut is defined by summing over the edge capacities of all edges in the multicut. Whereas this is a straightforward generalization of the problem of finding the minimum capacity cut that separates a single pair, there is no exact generalization of the max-flow min-cut theorem to multicuts. In fact, finding the minimum multicut in a general graph has been shown to be NP-hard \cite{garg1997primal}. 

It is however possible to upper bound the minimum multicut by means of a total multi-commodity flow optimization, also known as total multi-commodity flow, up to a factor $g_{\text{t}}(r)$ of order $\mathcal{O}(\log r)$ \cite{garg1996approximate} . A multi-commodity flow is a generalization of a flow to more than one source-sink pair, each exchanging a separate abstract `commodity'. In order to maximize the total multi-commodity flow one introduces separate edge flow variables $f_{e}^{(i)}$ for each commodity $i$ as well as each edge $e$ and maximizes the sum of flows from $s_i$ to $t_i$ over all commodities $i\in\{1,...,r\}$. In the optimization, one requires that for each commodity $i$ the flow is conserved in all edges except at the corresponding source $s_i$ and sink $t_i$, resulting in $r$ separate flow conservation constraints. Thus, it is ensured that for each commodity the net flow leaving the source will reach the corresponding sink. A multi-commodity flow is concurrent if all commodities can be distributed in parallel without exceeding the edge capacities in any edge. In order to ensure this, one adds the constraint that for each undirected edge $\{vw\}$ the sum of flows of all commodities passing through the edge, $\sum_{i=1}^r\left(f_{vw}^{(i)}+f_{wv}^{(i)}\right)$ does not exceed the edge capacity $c'(\{vw\})$. For details on multicuts and multi-commodity flows and the gaps that separate them see the Methods section. Figure \ref{fig:multiunicast} (a) shows an example of a minimum multicut separating all three source-sink pairs pairs. Figure \ref{fig:multiunicast} (b) shows a corresponding concurrent multi-commodity flow. The value of the minimum multicut  is 3.5, which is equal to the sum of the three source-sink flows. So in this simple example there is no gap. 

\begin{figure}
\centering
\includegraphics[width=\textwidth, trim={3cm 20cm 3cm 3cm},clip=true]{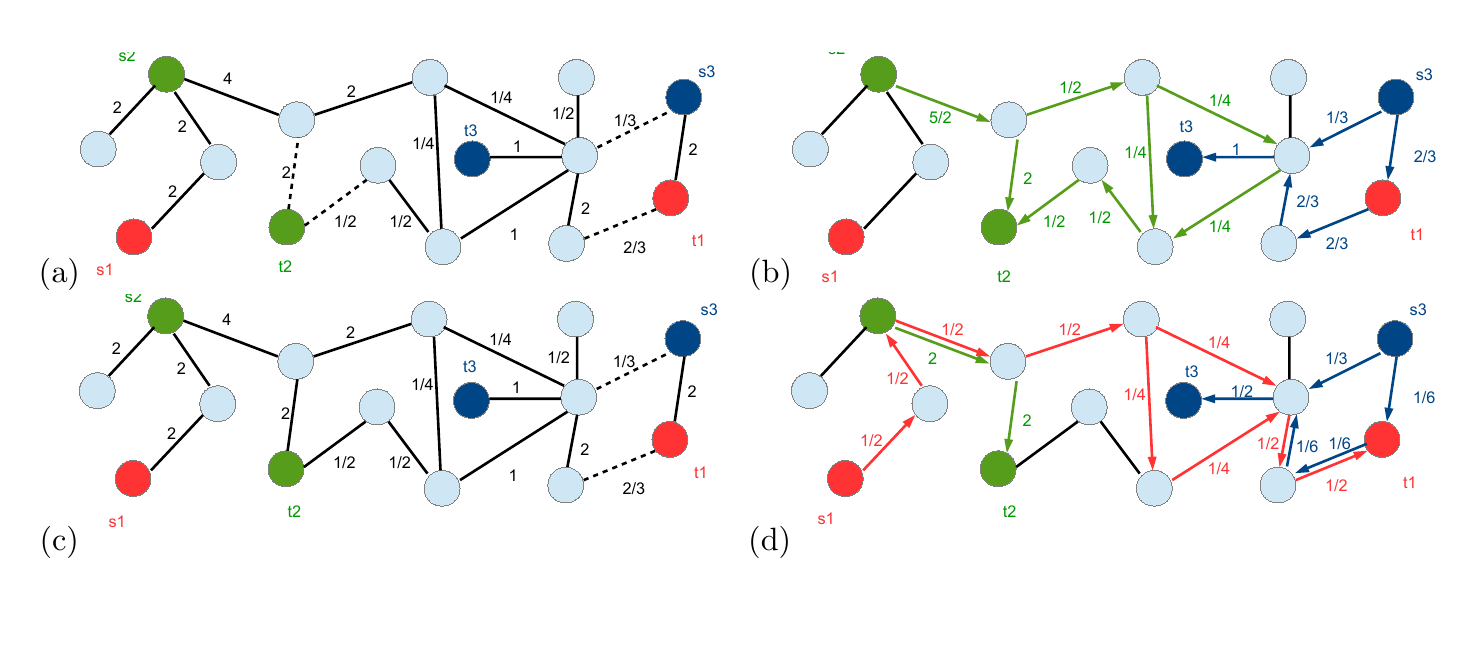}
%(a)\includegraphics[width=0.4\textwidth, trim={2cm 10cm 2cm 2cm},clip=true]{Multicut.pdf}     
%(b)\includegraphics[width=0.4\textwidth, trim={2cm 10cm 2cm 2cm},clip=true]{Total.pdf}
%\vspace{1cm}
%(c)\includegraphics[width=0.4\textwidth, trim={2cm 10cm 2cm 2cm},clip=true]{CutRatio.pdf}     
%(d)\includegraphics[width=0.4\textwidth, trim={2cm 10cm 2cm 2cm},clip=true]{WorstCase.pdf}
\caption{Example of a multi-user scenario with source-sink pairs, denoted by the pairs of red, green and blue vertices respectively. (a) A multicut (dashed edges) that separates all three source sink pairs. The capacity of the multicut is equal to $7/2$, which is the minimum value possible in this network. (b) A concurrent multi-commodity flow instance, with values $0$, $5/2$ and $1$ for the red, green and blue pairs, respectively. The total multi-commodity flow is hence equal to minimum multicut capacity, however at the price that there is no flow for the red pair. (c) Same network and same edge capacities as in (a), with an example of a bipartite cut, denoted in dashed lines, with capacity $1$ that separates two source-sink pairs, the red and the blue ones. Hence its cut ratio is given by $1/2$, which is also the minimum cut ratio in this graph. (d) Example of a corresponding multi-commodity flow instance, with concurrent flows of values $1/2$, $2$ and $1/2$ for the red, green and blue source-sink pairs. Hence, the worst-case multi-commodity flow is equal to $1/2$, in this case matching the minimum cut ratio. Whereas the flows sum up to $3$, which is less than the sum of flows in (b), this  multi-commodity flow instance is fairer than the one in (b) as it also provides a flow for the red user pair.
}\label{fig:multiunicast}
\end{figure}

Applying the aggregated repeater protocol \cite{AK16} to multiple user pairs and using the same reasoning as in the bipartite case, we also obtain lower bounds in terms of the maximum concurrent multi-commodity flows, providing us with the following efficiently computable bounds:

\begin{theorem}\label{Theo_Total}
In a network described by a graph $G$ with associated undirected graph $G'$and a scenario of $r$ user pairs $(s_1,t_1),...,(s_r,t_r)$, the total multi-pair quantum and private network capacities per total channel use, ${\cal Q}^{\rm total}(G, \{ {\cal N}^e \}_{e \in E} )$ and ${\cal P}^{\rm total}(G, \{ {\cal N}^e \}_{e \in E} )$, satisfy
\begin{align}
&\bar{f}^{\text{\emph{total}}}_{\max}(G',\{Q^{\leftrightarrow}(\mc{N}^e)\}_{e\in E})\leq{\cal Q}^{\rm total}(G, \{ {\cal N}^e \}_{e \in E} )\leq{\cal P}^{\rm total}(G, \{ {\cal N}^e \}_{e \in E} )\leq g_{\text{\emph{t}}}(r)\bar{f}^{\text{\emph{total}}}_{\max}(G',\{E_{\sq}(\mc{N}^e)\}_{e\in E}),
\end{align}
where $\bar{f}^{\emph{total}}_{\max}$ is given by the polynomial sized linear program Eq. (\ref{LP3bar}) presented in the Methods section and $g_{\rm t}(r)$ is of order $\mathcal{O}(\log r)$ as described in \cite{garg1996approximate}.
%\begin{align}
%\bar{f}^{\text{total}}_{\max}\left(G',\{C(e)\}_{e\in E}\right)=\textrm{max} &\ \sum_{i=1}^r\sum_{v:\{s_iv\}\in E'} \left(f^{(i)}_{s_iv}-f^{(i)}_{vs_i}\right)\label{LP3}\\
%\forall \{vw\}\in E':\ &\  \sum_{i=1}^r\left(f^{(i)}_{vw}+f^{(i)}_{wv}\right) \leq p_{wv}{C}({wv})+p_{vw}C({vw})\label{LP3a}\\
%\forall i,\ \forall w\in V,w\neq s_i,t_i:\ &\ \sum_{v:\{vw\}\in E'} \left(f^{(i)}_{vw} -f^{(i)}_{wv}\right)=0\label{LP3b}
%\end{align}
%where the maximization is over edge flows $f^{(i)}_{vw}\geq 0$ and $f^{(i)}_{wv}\geq 0$ for all edges $\{vw\}\in E'$ and for all $i\in\{1,...,r\}$ and over usage frequencies $0\leq p_e \leq 1,\ \sum_ep_e=1$ for all $e\in E$. If $e\notin E$ we set $p_e=0$. 
\end{theorem}
For the proof see Supplementary Note 3. 

%Let us now describe LP Eq. (\ref{LP3}) - Eq. (\ref{LP3b}) in detail: The objective Eq. (\ref{LP3}) is the sum of source-sink flows over all user pairs/commodities $i=1,...,r$. For each user pair, the source sink-flow is expressed analogously to Eq. (\ref{LP2}). The edge capacity constraints Eq. (\ref{LP3a}) involve a summation over all commodities that pass through given edge in both directions. The flow conservation constraints Eq. (\ref{LP3b}) are of the same form as Eq. (\ref{LP2b}), but have to be observed for all commodities $i=1,...,r$.

Let us now move on to scenario (2). Let us, again, describe the network by a capacitated graph with edge capacities $c(e)=p_e\mathcal{E}(\mathcal{N}^e)$, where $\mathcal{E}(\mathcal{N}^e)$ can be chosen to be the squashed entanglement $E_{\sq}$, the max-relative entropy of entanglement $E_{\max}$ and, for Choi-stretchable channels, the relative entropy of entanglement $E_R$ of the channel. Using the results of \cite{bauml2017fundamental,rigovacca2018versatile}, it is possible to show that the worst-case multi-pair private network capacity is upper bounded by the so-called minimum cut ratio with unit demands of the capacitated graph. Given a (bipartite) cut, which separates the set of vertices into two subsets, the cut ratio is defined as its capacity of the cut, i.e. the sum over edge capacities of the edges, divided by the demand across the cut, in this case the number of pairs separated by the cut. The minimum cut ratio is obtained by a minimization over all bipartite cuts. See Figure \ref{fig:multiunicast} (c) for an example of a minimum cut ratio. As for the minimum multicut discussed above, the computation of the minimum cut ratio is an NP hard problem in general graphs \cite{aumann1998log}. 

Whereas, as in the case of multicuts, there is no exact version of the max-flow min-cut theorem for the minimum cut ratio, there is again a connection to concurrent multi-commodity flows up to a factor $g_{\text{w}}(r)$, which can be of order up to $\mathcal{O}(\log r)$ \cite{aumann1998log}. Namely, it has been shown that the minimum cut ratio is upper bounded by $g_{\text{w}}(r)$ times what we call the maximum worst-case multi-commodity flow, also known as maximum concurrent multi-commodity flow, which corresponds to the maximum flow that can be achieved by any of the commodities concurrently, with respect to the same edge capacity and flow conservation constraints as in the case of the total multi-commodity flow, discussed previously. Figure \ref{fig:multiunicast} (d) contains an example of a maximum worst-case multi-commodity flow that achieves the cut ratio in figure \ref{fig:multiunicast} (c). Note that this flow is different from the one achieving the minimum multicut in figure \ref{fig:multiunicast} (b). In particular it is `fairer' in the sense that it also provides a flow for the red user pair $(s_1,t_1)$. See the Methods section for a detailed definition of the minimum cut ratio, the worst-case multi-commodity flow and the gap that separates them.

As in the previous scenarios, we can obtain a lower bound by application of the aggregated repeater protocol \cite{AK16} to multiple user pairs and include an optimization over usage frequencies, resulting in the following result:

\begin{theorem}\label{Theo_WC}
In a network described by a graph $G$ with associated undirected graph $G'$ and a scenario of $r$ user pairs $(s_1,t_1),...,(s_r,t_r)$, the worst-case multi-pair quantum and private network capacities per total channel use, ${\cal Q}^{\rm worst}(G, \{ {\cal N}^e \}_{e \in E} )$ and ${\cal P}^{\rm worst}(G, \{ {\cal N}^e \}_{e \in E} )$, satisfy
\begin{align}
&\bar{f}^{\text{\emph{worst}}}_{\max}(G',\{Q^{\leftrightarrow}(\mc{N}^e)\}_{e\in E})\leq{\cal Q}^{\rm worst}(G, \{ {\cal N}^e \}_{e \in E} )\leq{\cal P}^{\rm worst}(G, \{ {\cal N}^e \}_{e \in E} )\leq g_{\text{\emph{w}}}(r)\bar{f}^{\text{\emph{worst}}}_{\max}(G',\{{\cal E}(\mc{N}^e)\}_{e\in E}),
\end{align}
where $\bar{f}^{\text{\emph{worst}}}_{\max}$ is given by the polynomially sized linear program Eq. (\ref{LP4bar}) presented in the Methods section.
%\begin{align}
%\bar{f}^{\text{worst}}_{\max}\left(G',\{C(e)\}_{e\in E}\right)=\textrm{max} &\ f\label{LP4}\\
%\forall i:\ &\ f-\sum_{v:\{s_iv\}\in E'} \left(f^{(i)}_{s_iv}-f^{(i)}_{vs_i}\right)\leq0\label{LP4a}\\
%\forall \{vw\}\in E':\ &\  \sum_{i=1}^r\left(f^{(i)}_{vw}+f^{(i)}_{wv}\right) \leq p_{wv}{C}({wv})+p_{vw}C({vw})\label{LP4b}\\
%\forall i,\ \forall w\in V,w\neq s_i,t_i:\ &\ \sum_{v:\{vw\}\in E'} \left(f^{(i)}_{vw} -f^{(i)}_{wv}\right)=0\label{LP4c}
%\end{align}
%where the maximization is over edge flows $f^{(i)}_{vw}\geq 0$ and $f^{(i)}_{wv}\geq 0$ for all edges $\{vw\}\in E'$ and for all $i\in\{1,...,r\}$ and over usage frequencies $0\leq p_e \leq 1,\ \sum_ep_e=1$ for all $e\in E$. If $e\notin E$ we set $p_e=0$. 
Further $g_{\text{\emph{w}}}(r)$  is the flow-cut gaps described in the Methods section. $\cal E$ can be chosen to be the squashed entanglement $E_{\sq}$, the max-relative entropy of entanglement $E_{\max}$ and, for a Choi-stretchable channels, the relative entropy of entanglement $E_R$. 
\end{theorem}

For the proof, see Supplementary Note 3. As a proof of principle demonstration, we have numerically computed the worst-case and total multi-commodity flows for an example network. See Supplementary Note 5 for details and plots.

\subsection{Multipartite target states}\label{sec:multi}

In this section we present our results on the distribution of multipartite entanglement. Let us consider a set of disjoint users $S=\{s_1,...,s_l\}$, who wish to establish a multipartite target state, such as a GHZ state \cite{greenberger1989going} $\ket{\Phi^{\text{GHZ},d}}_{M_{s_1}...M_{s_l}}=\frac{1}{\sqrt{d}}\sum_{i=0}^{d-1}\ket{i}_{M_{s_1}}\otimes\cdots\otimes\ket{i}_{M_{s_l}}$ or a multipartite private state \cite{augusiak2009multipartite}, $\gamma^d_{K_{s_1}S_{s_1}...K_{s_l}S_{s_l}}=U^{\text{twist}}|\Phi^{\text{GHZ},d}\>\<\Phi^{\text{GHZ},d}|_{K_{s_1}...K_{s_l}}\otimes\sigma_{S_{s_1}...S_{s_l}}U^{\text{twist}\dagger}$, where $\sigma_{S_{s_1}...S_{s_l}}$ is an arbitrary state and $U^{\text{twist}}=\sum_{i_1,\cdots,i_l}|i_1,\cdots,i_l\>\<i_1,\cdots,i_l|_{K_{s_1}...K_{s_l}}\otimes U^{(i_1,\cdots,i_l)}_{S_{s_1}...S_{s_l}}$ is a controlled unitary operation. The corresponding multipartite quantum and private network capacities $\mathcal{Q}^S$ and $\mathcal{P}^S$ are defined analogously to the bipartite case.

As a consequence of \cite{bauml2017fundamental}, the private capacity is upper bounded by the connectivity of the set $S$ of user vertices in the graph capacitated by $p_eE_{sq}(\mathcal{N}^e)$. By connectivity of the set $S$ we mean the minimum edge capacity cut that separates any two vertices $s_i$ and $s_j$ ($i\neq j$) in $S$. Such a cut is also known as minimum $S$-cut or minimum Steiner cut with respect to $S$.  See figure \ref{fig:multicast} (a) for an example of a minimum Steiner cut with respect to the set of the  red, green, blue and yellow nodes. The computation of the connectivity consists of a minimization of all possible disjoint vertex pairs within $S$ as well as a cut minimization. See Eq. (\ref{eq:connectivity}) in the Methods section. Applying the max-flow min cut theorem for every possible pair $s_i$ and $s_j$ ($i\neq j$) in $S$, we can transform the computation of the connectivity of $S$ into another linear program that upper bounds the multipartite network private capacity. See figure \ref{fig:multicast} (b) for an example.

\begin{figure}

\centering
\includegraphics[width=\textwidth, trim={3cm 20cm 3cm 3cm},clip=true]{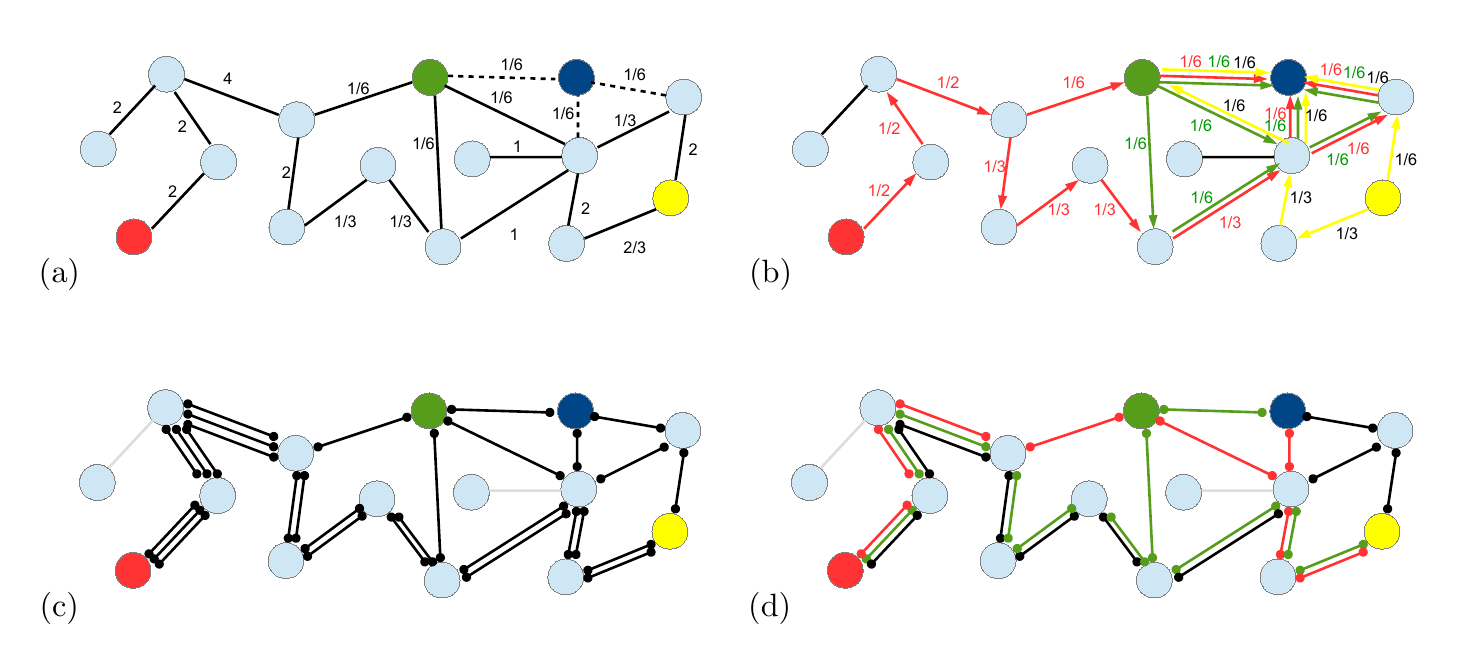}
%(a)\includegraphics[width=0.4\textwidth, trim={2cm 10cm 2cm 2cm},clip=true]{Steinercut.pdf}     
%(b)\includegraphics[width=0.4\textwidth, trim={2cm 10cm 2cm 2cm},clip=true]{ConnectivityFlow.pdf}

%\vspace{1cm}
%(c)\includegraphics[width=0.4\textwidth, trim={2cm 10cm 2cm 2cm},clip=true]{Bellnetwork.pdf}     
%(d)\includegraphics[width=0.4\textwidth, trim={2cm 10cm 2cm 2cm},clip=true]{Steinertrees.pdf}     

\caption{Example of a setting where a group of four user nodes (red, green, blue, yellow) wishes to establish a Greenberger-Horne-Zeilinger (GHZ)  state. (a) The graph $G'$ with labeled edge capacities. The dashed edges correspond to a minimum Steiner cut with respect to the set of the four users, i.e. it is a smallest capacity cut that separates at least one pair of vertices in the set. In this case it separates the red-blue, green-blue and yellow-blue pairs and has capacity $1/2$. In other words, the set of users is $1/2$-connected. (b) Part of a flow instance corresponding to linear program (LP) given by Eq. (\ref{LP5})-Eq. (\ref{LP5d}). Here the flows from the red, green and yellow vertices to the blue vertex are shown in red, green and yellow, respectively. For simplicity, flows between other nodes are not shown in this picture. The directed edge-flows correspond to the variables $f^{(ij)}_{vw}$ of the LP given by Eq. (\ref{LP5})-Eq. (\ref{LP5d}). By providing flows of value of at least $1/2$ between all pairs in the set of users, the LP shows that the set of users is $1/2$-connected. (c) Aggregated Repeater Protocol: Assuming that the edges in (a) correspond to quantum channels (of some direction) and their capacities to non-asymptotic quantum capacities, one could, by using each channel (at most) $6$ times, create a network of Bell states that is described by a $3$-connected undirected multigraph. (d) Steiner Trees: In our example the multigraph contains two edge-disjoint Steiner trees, depicted in red and green. The Bell pairs forming the Steiner trees can then be connected by means of a generalized entangled swapping protocol to form $2$ qubit GHZ states among the four users.
}\label{fig:multicast}
\end{figure}

Finding a lower bound on the multipartite network quantum capacity, i.e. the maximum rate at which we can distribute a GHZ state among $S$, is slightly more involved than in the previously considered scenarios. As we did in all previous scenarios, we begin by performing an aggregated repeater protocol to create a network of Bell states that can be described by an undirected multigraph. See figure \ref{fig:multicast} (c) for an example. Whereas it is possible to create a GHZ state locally in one of the nodes and use chains of Bell pairs to teleport the respective subsystems of the GHZ state to all other nodes in $S$, it is easy to find a network where this is not the optimal strategy. Instead, the idea is to generalize the concept of paths linking two nodes to Steiner trees spanning the set $S$ of users. In an undirected multigraph, a Steiner tree spanning $S$, or short $S$-tree, is an acyclic subgraph that connects all nodes in $S$. See figure \ref{fig:multicast} (d) for an example of two edge disjoint Steiner trees spanning the set of the red, green, blue and yellow nodes. See the Methods section for more information on Steiner trees. 

A Steiner tree spanning $S$ in the network of (qubit) Bell states can be transformed into a (qubit) GHZ  state among all nodes in $S$ by means of a protocol introduced in \cite{wallnofer20162d}, which can be seen as a generalization of entanglement swapping. Hence, the number of (qubit) GHZ states obtainable from the Bell state network is equal to the number of edge disjoint Steiner trees spanning $S$. Computing this number is a referred to as a Steiner tree packing, which is anther NP-complete problem \cite{cheriyan2006hardness}. However, the number of edge-disjoint Steiner trees in a multigraph can be lower bounded by its $S$-connectivity up to constant factor $1/2$ and an additive constant \cite{kriesell2003edge,lau2004approximate,petingi2009packing}. Combining this with the max-flow min cut theorem, this allows us to derive a linear-program lower bound on the multipartite quantum network capacity. Hence we obtain the following:

\begin{theorem}\label{Theo_multi}
In a network described by a graph $G$ with associated undirected graph $G'$ and a scenario of a set $S=\{s_1,...,s_r\}$ of users, the quantum and private network capacities per total channel use, $\mathcal{Q}^{S}\left(G,\{\mc{N}^e\}_{e\in E}\right)$ and $\mathcal{P}^{S}\left(G,\{\mc{N}^e\}_{e\in E}\right)$, satisfy
\begin{equation}
\frac{1}{2}\bar{f}^{S}_{\max}(G',\{Q^{\leftrightarrow}(\mc{N}^e)\}_{e\in E})\leq\mathcal{Q}^{S}\left(G,\{\mc{N}^e\}_{e\in E}\right)\leq\mathcal{P}^{S}\left(G,\{\mc{N}^e\}_{e\in E}\right)\leq\bar{f}^{S}_{\max}(G',\{E_{\sq}(\mc{N}^e)\}_{e\in E}),
\end{equation}
where $\bar{f}^{S}_{\max}$ is given by the polynomially sized linear program Eq. (\ref{LP5bar}) presented in the Methods section.
%\begin{align}
%\bar{f}^{S}_{\max}\left(G',\{C(e)\}_{e\in E}\right)=\textrm{max} &\ f\label{LP6}\\
%\forall i,j> i:\ &\ f-\sum_{v:\{s_iv\}\in E'} \left(f^{(ij)}_{s_iv}-f^{(ij)}_{vs_i}\right)\leq0\label{LP6a}\\
%\forall i,j> i,\{vw\}\in E':\ &\  f^{(ij)}_{vw}+f^{(ij)}_{wv}\leq p_{wv}{C}({wv})+p_{vw}C({vw})\label{LP6b}\\
%\forall i,j> i,\ \forall w\in V,w\neq s_i,s_j:\ &\ \sum_{v:\{vw\}\in E'} \left(f^{(ij)}_{vw} -f^{(ij)}_{wv}\right)=0,\label{LP6c}
%\end{align}
%where the maximization is over edge flows $f^{(ij)}_{vw}\geq 0$ and $f^{(ij)}_{wv}\geq 0$  for all edges $\{vw\}\in E'$ and for all disjoint pairs $s_i,s_j\in S$ with $j>i$, as well as over usage frequencies $0\leq p_e \leq 1,\ \sum_ep_e=1$ for all $e\in E$. If $e\notin E$ we set $p_e=0$.
\end{theorem}
For the proof see Supplementary Note 4.

\section{Discussion}\label{sec:concl}

We have provided linear-program upper and lower bounds on the entanglement and key generation capacities in quantum networks for various user scenarios. We have done so by reducing the corresponding network routing problems to flow optimizations, which can be written as linear programs. The user scenarios we have considered are the distribution of Bell or private states between a single pair of users, the parallel distribution of such states between multi-pairs of users  and the distribution of GHZ or multipartite private states among a group of multiple users. The size of the linear programs scales polynomially in the parameters of the networks, and hence the LPs can be computed in polynomial time. In order to perform the LPs, upper and lower bounds on the two-way assisted private or quantum capacities of all the channels constituting the network have to be provided as input parameters. Thus the problem of bounding capacities for the entire network is reduced to bounding capacities of single channels, as well as performing an LP which scales polynomially in the network parameters.

For a large class of practical channels, including erasure channels, dephasing channels, bosonic quantum amplifier channels and lossy optical channels, tight bounds can be obtained in the bipartite case. In the multi-pair case, however, there still remains a gap of order up to $\log r^*$ between the upper and lower bounds. This gap, also known as flow-cut gap, is due to the lack of an exact max-flow min-cut theorem for multi-commodity flows. From a complexity theory standpoint, the flow-cut gap separates the NP-hard problem of determining the minimum cut ratio from the problem of finding the maximum concurrent multi-commodity flow, which can be done in polynomial time \cite{aumann1998log}. From a network theoretic view the gap also leaves room for a possible advantage of network coding over network routing in undirected networks, which is still an open problem \cite{li2004network,harvey2004comparing}. Another gap, of value $1/2$, occurs between our upper and lower bounds in the multi-pair case. As in the multiple-pair case, this gap is significant in terms of computational complexity, as it separates our polynomial LP from the problem of Steiner tree packing, which is NP-complete \cite{cheriyan2006hardness}.

While our linear programs cover an important set of user scenarios and tasks, we believe that our recipe will find broader use. In the bipartite case, we could assign costs to the links and consider the problem of minimizing the total cost for a given set of user demands \cite{ford2015flows}. In the multipartite case, we could apply it to the distribution of multipartite entanglement between multiple groups of users, for which one could leverage results connecting the minimum ratio Steiner cut problem and the Steiner multicut problem with concurrent Steiner flows \cite{klein1997approximation}. As another example, beyond network capacities, many algorithms for graph clustering and community detection in complex networks rely on the sparsest cut of graph \cite{kannan2004clusterings,schaeffer2007graph}. This quantity is bounded from below by the uniform multi-commodity flow problem, which is an instance of our multi-pair entanglement distribution maximizing the worst-case multi-commodity flow, and from above by the same quantity multiplied by a value that scales logarithmically with the number of nodes in the network. Hence, the direct solution of this instance could be used to solve the analogous problem in complex networks where the links are evaluated for their capability to transmit quantum information or private classical information. Although we have focused on a linear program to bound capacities, rather than actual rates in practical scenarios with other imperfections, such as storage limitation or overheads, 
we believe that our program could be the basis to develop an algorithm to treat such practical scenarios as well.

\section{Methods}

\subsection{Bipartite user scenario}
In this section we will explicitly define all quantities that occur in our result for bipartite user scenarios, Theorem \ref{Theo_bip}, and briefly review the main ingredient in its proof, the max-flow min-cut theorem. Let us begin with the definition of the capacities: The quantum and private network capacities per total channel use that occur in Theorem \ref{Theo_bip} are defined as
\begin{align}
&\mathcal{Q}\left(G,\{\mc{N}^e\}_{e\in E}\right)=\max_{p_e\geq0,\sum_ep_e=1}\lim_{\epsilon\to0}\lim_{n\to\infty}\sup_{\Lambda}\left\{\frac{\langle\log d^{(k)}\rangle_{k}}{n}:\left\|\rho_{M_sM_t}^{(n,k)}-{\Phi^{d^{(k)}}}_{M_sM_t}\right\|_1\leq\epsilon\right\},\label{cap2}\\
&\mathcal{P}\left(G,\{\mc{N}^e\}_{e\in E}\right)=\max_{p_e\geq0,\sum_ep_e=1}\lim_{\epsilon\to0}\lim_{n\to\infty}\sup_{\Lambda}\left\{\frac{\langle\log d^{(k)}\rangle_{k}}{n}:\left\|\rho_{K_sK_tS_sS_t}^{(n,k)}-\gamma^{d^{(k)}}_{K_sK_tS_sS_t}\right\|_1\leq\epsilon\right\},\label{cap1}
\end{align}
where the suprema are over all adaptive $(n,\epsilon, \{p_e\}_{e\in E})$ protocols $\Lambda$. Further $k=(k_1,  \ldots  ,k_{n+1})$ is a vector keeping the track of outcomes of the $n+1$ LOCC rounds in $\Lambda$, the averaging, denoted by the parenthesis $\langle...\rangle_{k}$, is over all those outcomes and $\rho_{M_sM_t}^{(n,k)}$ is the final state of $\Lambda$ for given outcomes $k$.

Let us discuss the difference between the above quantities and network capacities introduced in \cite{pirandola2016capacities,pirandola2019capacities}, which consider rates per network use. There are two strategies considered in \cite{pirandola2016capacities,pirandola2019capacities}, sequential (or single-path) routing and multi-path routing. Both strategies are adaptive in the same sense as defined above, i.e. the channel uses are interleaved by LOCC operations among all nodes, the number of LOCC rounds being equal to the total number of channel uses. 

In the case of sequential (or single-path) routing, one use of the network involves usage of channels along a single path from Alice to Bob. The path, and its length, can change with every use of the network. This strategy could correspond to the external provider offering a path for the users (similar to the paradigm of circuit switching networks \cite{leon2003communication}) instead of allowing the users to precisely determine the usage frequencies of each channel. 

In the case of multi-path routing, a flooding strategy is applied, where during each use of the network each channel is used exactly once. Hence the total number of channel uses is given by $|E|$ times the number of network uses. As shown in \cite{pirandola2016capacities,pirandola2019capacities}, there are examples of networks, such as the so-called diamond network, for which such a strategy provides an advantage over single-path routing. The multi-path scenario could correspond to a private quantum network where the users are willing to use the whole of their resources each clock cycle to implement the desired communication task. 

In the present paper, an alternative approach is taken. Instead of considering rates per use of the network, we consider rates per the total number of channel uses. By setting our usage frequencies  $p_e$ constant for all nodes $e\in E$ in the network, we can incorporate the flooding strategy used in the multi-path routing scenario of \cite{pirandola2016capacities,pirandola2019capacities}. Hence, although phrased with the channel use metric, our results also can be used for the network use metric generalizing the original results in \cite{pirandola2016capacities,pirandola2019capacities} to multipartite settings. There is, however, no direct relation between our capacities and the single path capacities. In fact they can differ by a factor $\mc{O}(|E|)$, which is the order of the number of vertices in the network, as shown in figure \ref{fig:SinglePath}. 
\begin{figure}

\centering
\includegraphics[width=0.7\textwidth, trim={1cm 10cm 2cm 2cm},clip=true]{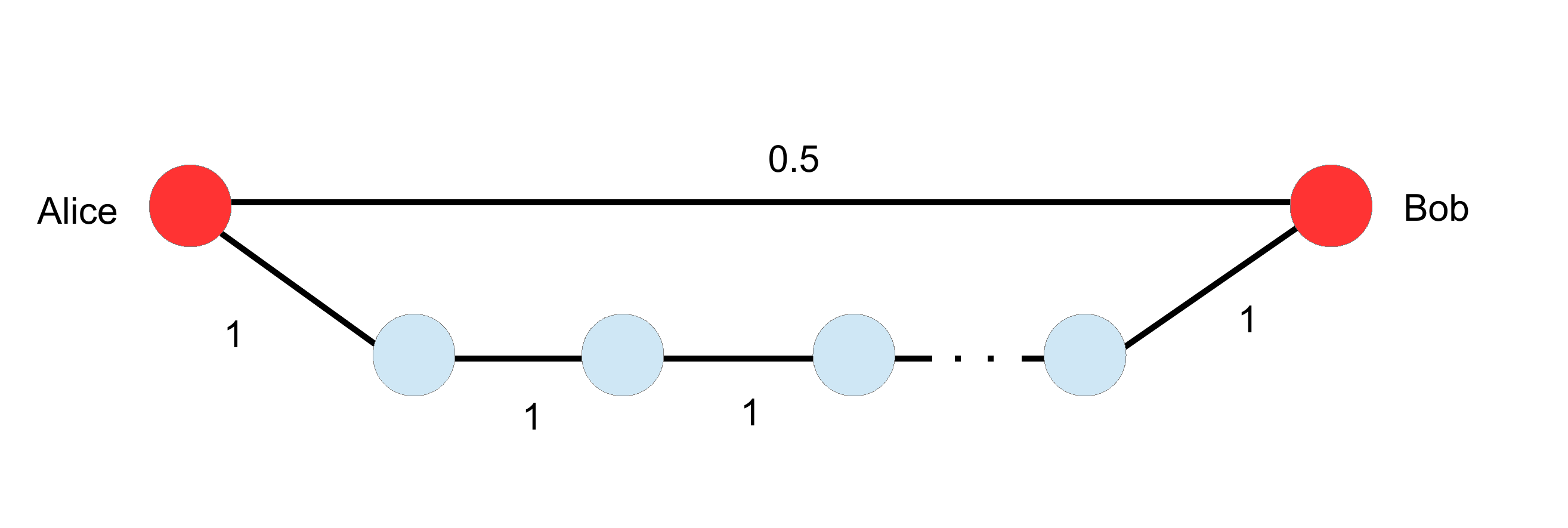}     

\caption{Per-network-use versus per-channel-use capacities: Simple example of a network, where our capacity Eq. (\ref{cap2}) can differ arbitrarily from the single-path capacity introduced in \cite{pirandola2016capacities,pirandola2019capacities}. The numbers refer to capacities of the single channel. When the goal is to maximize the transmission per total number of channel uses, the upper route is preferable. It can achieve a transmission of $0.5$ using a single channel, i.e. a rate per channel use of $0.5$. The lower route can achieve a transmission of $1$ using $n$ (greater than two) channels, i.e. a rate per channel use of $1/n$. When the goal is to maximize the transmission per uses of the network over a single path as in \cite{pirandola2016capacities,pirandola2019capacities}, the lower route is preferable as it can achieve a transmission of $1$ per use of the network, whereas the upper route can achieve $0.5$.}\label{fig:SinglePath}
\end{figure}

We will now introduce the linear program that provides upper and lower bounds on the capacities Eq. (\ref{cap2}) and Eq. (\ref{cap1}), respectively. Let us consider undirected graph $G'=(V,E')$, as defined at the beginning of the Result section, with edge capacities $c'(\{vw\})$ for all $\{vw\}\in E'$.  We assume that we have two special nodes $s,t\in V$, which we call the source and the sink. As the entanglement across each edge can be used in both directions, we assign two edge flows $f_{wv}\geq0$ and $f_{vw}\geq0$ to each edge $\{wv\}\in E'$, where $f_{wv}$ corresponds to a flow from $w$ to $v$ and $f_{vw}$ to a flow in the opposite direction. 

The goal is now to maximize the flow from $s$ to $t$ over the graph $G'$. In order to be a feasible flow, it should not exceed the capacity of each edge. Namely, for each edge $\{vw\}$ we need
\begin{equation}\label{capconstr}
f_{wv}+f_{vw}\leq c'(\{wv\}).%,
\end{equation}
We also need that for each edge  $w\neq s,t$ 
\begin{equation}\label{flowcons}
\sum_{v:\{vw\}\in E'} f_{vw} = \sum_{v:\{vw\}\in E'}f_{wv},
\end{equation}
which is known as flow conservation. By this flow conservation the flow from $s$ to $t$ is equal to the flow leaving the source minus the flow entering the source,
\begin{equation}\label{obj}
f^{s\to t}=\sum_{v:\{sv\}\in E'}(f_{sv}-f_{vs}).
\end{equation}
In order to obtain the maximum flow from $s$ to $t$ over the graph $G'$, we need to maximize Eq. (\ref{obj}) over edge flows with respect to constraints Eq. (\ref{capconstr}) and Eq. (\ref{flowcons}), which is a linear program:

\begin{align}\label{LP1}
f^{s\to t}_{\max}(G',\{c'(\{wv \}) \}_{\{wv\}\in E'})=\textrm{max}& \sum_{v:\{sv\}\in E'} (f_{sv}-f_{vs})\\
&\forall \{vw\}\in E':\ f_{wv}+f_{vw}\leq c'(\{wv\})\nonumber\\
&\forall \{vw\}\in E':\ f_{wv}, f_{vw}\geq 0\nonumber\\
&\forall w\in V: w\neq s,t,\ \sum_{v:\{vw\}\in E'} f_{vw} = \sum_{v:\{vw\}\in E'}f_{wv}. \nonumber
\end{align}
%where 
%\begin{equation}
%c'(\{wv\})=\begin{cases}
%c(wv)+c(vw)\text{ if }wv\in E\text{ and }vw\in E\\
%c(wv)\text{ if }wv\in E\text{ and }vw\notin E\end{cases}
%\end{equation}
In Theorem \ref{Theo_bip}, we set the capacities $c'(\{wv\})$ to
\begin{align}
&c'_{C}(\{vw\},p_{wv},p_{vw})=p_{wv}{C}(\mathcal{N}^{wv})+p_{vw}{C}(\mathcal{N}^{vw})\label{cap_lower},
\end{align}
where $C=Q^\leftrightarrow$ for the lower bound and $C=\mathcal{E}$ for the upper bound, respectively. In the case where $wv\in E$ but $vw\notin E$, we set $p_{vw}{C}(\mathcal{N}^{vw})=0$. Further we add an optimization over the usage frequencies
\begin{equation}\label{LP1bar}
\bar{f}^{s\to t}_{\max}(G',\{C(\mc{N}^e)\}_{e\in E})=\max_{0\leq p_e \leq 1,\ \sum_ep_e=1}f^{s\to t}_{\max}(G',\{c'_{C}(\{vw\},p_{wv},p_{vw})\}_{\{vw\}\in E'})
\end{equation}

In the following we will make use of the max-flow min-cut theorem: Given a subset $V'\subset V$ we define a cut of $G'$  as the set 
\begin{equation}
\partial(V'):=\{\{vw\}\in E':  v\in V', w\in V\setminus V'\}.
\end{equation}
If, for given vertices $s$ and $t$, and a set $V_{s;t}\subset V$,  $s\in V_{s;t}$ and $t\in V\setminus V_{s;t}$ we call $\partial(V_{s;t})$ an $st$-cut. Let us note that the first and second indices in the subscript of $V_{s;t}$ have different meanings. 
%Given vertices $s$ and $t$ and sets $V_s$ and $V_t$ such that $s\in V_s\subset V$ and $t\in V_t:=V\setminus V_s$, we define an $st$\emph{-cut} of $G'$ as
%\begin{equation}\label{eq:stcut}
%V_s\leftrightarrow V_t=\partial(V_s).
%\end{equation}
%with vertex sets $V_s,V_t\subset V$, where $V_t=V\setminus V_s$, such that $s\in V_s$ and $t\in V_t$. 
The minimum $st$-cut of $G'$ is defined as
\begin{equation}\label{mincut}
%\min_{V_s\leftrightarrow V_t}\sum_{\{vw\}\in V_s\leftrightarrow V_t}c'(\{vw\}).
\min_{V_{s;t}}\sum_{\{vw\}\in \partial(V_{s;t})}c'(\{vw\}),
\end{equation}
where the minimization is over all $V_{s;t}\subset V$ such that $s\in V_{s;t}$ and $t\in V\setminus V_{s;t}$. By the max-flow min-cut theorem \cite{dantzig1955max,EFS96} it holds 
\begin{equation}\label{maxFlowminCut}
f^{s\to t}_{\max}(G',\{c'(\{wv \}) \}_{\{wv\}\in E'})=\min_{V_{s;t}}\sum_{\{vw\}\in \partial(V_{s;t})}c'(\{vw\}).
%\min_{V_s\leftrightarrow V_t}\sum_{\{vw\}\in V_s\leftrightarrow V_t}c'(\{vw\}).
\end{equation}
See figure \ref{fig:unicast} for an example illustrating the connection between cuts and flows. 
%\if0
%Let us also define a directed path from $s$ to $t$ in $G'$ as
%\begin{equation}
%P_{s\to t}=\{v_jv_{j+1}:i,j=1,\ldots,l-1,v_j\in V,\{v_jv_{j+1}\}\in E',v_1=s,v_{l}=t,v_j \neq v_i, i \neq j\}.
% \end{equation}
%In a finite graph with nonzero flow $f^{s\to t}$, obtained from a solution $\{f_{vw},f_{wv} \}_{ \{vw\} \in E' }$ of Eq. Eq. (\ref{LP1}), we can always find a finite number $N$ of paths $P_{s\to t}^{(1)},...,P_{s\to t}^{(N)}$ whose flow consists of path-flows $f^{(1)},...,f^{(N)}$ such that $f^{s\to t}=\sum_{i=1}^Nf^{(i)}$  \cite{ford1956maximal}\footnote{The authors of \cite{ford1956maximal} use the notion of chains (i.e. undirected paths) linking $s$ and $t$. Here we add a direction such that they are directed from $s$ to $t$.}. It will be convenient to define for every path $P_{s\to t}^{(i)}$ and every edge $vw\in E$ the quantity
%\begin{equation}\label{eq:fvwi}
%f_{vw}^{(i)}:=\begin{cases}
%&f^{(i)}\text{ if }vw\in P_{s\to t}^{(i)}\\&0\text{ else.}
%\end{cases}
%\end{equation}
%Note that an edge $vw$ can be part of more than one paths. The sum of path-flows passing through the edge, however, has to be upper bounded by the edge flow $f_{vw}$. Hence it holds
%\begin{equation}\label{eq:sumfvwi}
%\sum_{i=1}^Nf_{vw}^{(i)}\leq f_{vw}.
%\end{equation} 
%The edge flow $f_{vw}$, in turn, is constraint by the capacity constraint Eq. (\ref{capconstr}) \cite{ford1956maximal}. If all $f^{(i)}$ take integer values, there exist $f^{s\to t}=\sum_{i=1}^Nf^{(i)}$ \emph{edge-disjoint paths} from $s$ to $t$ in the multigraph $G_{\lfloor c' \rfloor}''$. 

%\fi

\subsection{Multiple pairs of users}
In this section we will explicitly define all quantities that occur in our results for multiple pairs of users, Theorems \ref{Theo_Total} and \ref{Theo_WC}. We also briefly introduce multi-commodity flows and the corresponding generalizations of the max-flow min-cut theorem, which are used in the proofs of Theorems \ref{Theo_Total} and \ref{Theo_WC}.
 
We begin by defining a total multi-pair quantum network capacity and a worst-case multi-pair quantum network capacity per total channel use respectively as
\begin{align}
&{\cal Q}^{\rm total}\left(G,\{\mc{N}^e\}_{e\in E}\right)=\max_{\substack{ p_e\geq0\\ \sum_ep_e=1}}\lim_{\epsilon\to0}\lim_{n\to\infty}\sup_{\Lambda}\left\{\frac{\sum_{i=1}^r\langle\log d_i^{(k)}\rangle_{k}}{n}:\left\|\rho_{M_{s_1}M_{t_1}  \cdots  M_{s_r}M_{t_r}}^{(n,k)}-\bigotimes_{i=1}^r\Phi^{d_i^{(k)}}_{M_{s_i}M_{t_i}}\right\|_1\leq\epsilon\right\},\label{cap3a}\\
&{\cal Q}^{\rm worst}\left(G,\{\mc{N}^e\}_{e\in E}\right)=\max_{\substack{ p_e\geq0\\ \sum_ep_e=1}}\lim_{\epsilon\to0}\lim_{n\to\infty}\min_{i\in\{1 \cdots  r\}}\sup_{\Lambda}\left\{\frac{\langle\log d_i^{(k)}\rangle_{k}}{n}:\left\|\rho_{M_{s_1}M_{t_1}  \cdots  M_{s_r}M_{t_r}}^{(n,k)}-\bigotimes_{i=1}^r\Phi^{d_i^{(k)}}_{M_{s_i}M_{t_i}}\right\|_1\leq\epsilon\right\}\label{cap3},
\end{align}
where the supremum is over all adaptive $(n,\epsilon, \{p_e\}_{e\in E})$ protocols $\Lambda$ and $k=(k_1,  \ldots  ,k_{m+1})$ is a vector of outcomes of the $m+1$ LOCC rounds in $\Lambda$, the averaging is over all those outcomes and $\rho^{(n,k)}$ is the final state of $\Lambda$ for given outcomes $k$. The corresponding private capacities ${\cal P}^{\rm total}\left(G,\{\mc{N}^e\}_{e\in E}\right)$ and ${\cal P}^{\rm worst}\left(G,\{\mc{N}^e\}_{e\in E}\right)$  are defined by replacing $\Phi^{d_i^{(k)}}_{M_{s_i}M_{t_i}}$ by $\gamma^{d_i^{(k)}}_{K_{s_i}S_{s_i}K_{t_i}S_{t_i}}$ in Eqs. Eq. (\ref{cap3a}) and Eq. (\ref{cap3}).

%In scenario (3), we assign nonnegative weights $q_1,\cdots,q_r$, with $\sum_iq_i=1$, to user pairs $(s_1,t_1), \cdots  ,(s_r,t_r)$ and define a \emph{weighted multi-pair quantum network capacity} as
%\begin{equation}\label{cap3c}
%\mathcal{Q}^{(s_1,t_1), \cdots  ,(s_r,t_r)}_{q_1,\cdots,q_r}\left(G,\{\mc{N}^e\}_{e\in E}\right)=\max_{\substack{ p_e\geq0\\ \sum_ep_e=1}}\lim_{\epsilon\to0}\lim_{n\to\infty}\sup_{\Lambda}\left\{\frac{\sum_{i=1}^rq_i\langle\log d_i^{(k)}\rangle_{k}}{n}:\left\|\rho_{M_{s_1}M_{t_1}  \cdots  M_{s_r}M_{t_r}}^{(k)}-\bigotimes_{i=1}^r\Phi^{d_i^{(k)}}_{M_{s_i}M_{t_i}}\right\|_1\leq\epsilon\right\}.
%\end{equation}
%In the case where $q_i=1/r$ for all $i=1,...,r$, Eq. (\ref{cap3c}) reduces to Eq. (\ref{cap3a}), up to a normalization factor. By replacing $\bigotimes_{i=1}^r\Phi^{d_i^{(k)}}_{M_{s_i}M_{t_i}}$ with $\bigotimes_{i=1}^r\gamma^{d_i}_{K_{s_i}K_{t_i}S_{s_i}S_{t_i}}$ in Eq. (\ref{cap3a}), Eq. (\ref{cap3}) and Eq. (\ref{cap3c}), we can define the corresponding private capacities, $\mathcal{P}^{(s_1,t_1), \cdots  ,(s_r,t_r)}_{\text{total}}$, $\mathcal{P}^{(s_1,t_1), \cdots  ,(s_r,t_r)}_{\text{worst}}$ and $\mathcal{P}^{(s_1,t_1), \cdots  ,(s_r,t_r)}_{q_1,\cdots,q_r}$, respectively.

The bounds on Eqs. Eq. (\ref{cap3a}) and Eq. (\ref{cap3}) given in Theorems \ref{Theo_Total} and \ref{Theo_WC}, respectively are in terms of multi-commodity flow optimizations, which we will introduce in this section. A flow instance involving multiple sources and sinks $s_1,\cdots,s_r$ and $t_1,\cdots ,t_r$ is known as a multi-commodity flow, each flow $f^{(i)}$ from $s_i$ to $t_i$ being considered to be a separate commodity. The the maximum total multi-commodity flow is then obtained by maximizing the sum over all single-commodity flows. Generalizing LP Eq. (\ref{LP1}) accordingly, we obtain the following LP:
\begin{align}
f^{\text{total}}_{\max}\left(G',\{c'(\{vw \}) \}_{\{vw\}\in E'}\right)=\textrm{max} &\sum_{i=1}^r\sum_{v:\{s_iv\}\in E'} \left(f^{(i)}_{s_iv}-f^{(i)}_{vs_i}\right)\\
\forall \{vw\}\in E':\ &\  \sum_{i=1}^r\left(f^{(i)}_{vw}+f^{(i)}_{wv}\right)\le  c'(\{vw\})\nonumber\\
\forall \{vw\}\in E', \ \forall i:\ &\  f^{(i)}_{vw}, f^{(i)}_{wv}\geq 0\nonumber\\
\forall i,\ \forall w\in V,w\neq s_i,t_i:\ &\ \sum_{v:\{vw\}\in E'} \left(f^{(i)}_{vw} -f^{(i)}_{wv}\right)=0.\nonumber
\end{align}
Again, we can define
\begin{equation}\label{LP3bar}
\bar{f}^{\text{total}}_{\max}(G',\{C(\mc{N}^e)\}_{e\in E})=\max_{0\leq p_e \leq 1,\ \sum_ep_e=1}f^{\text{total}}_{\max}(G',\{c'_{C}(\{vw\},p_{wv},p_{vw})\}_{\{vw\}\in E'}).
\end{equation}
Further, the maximum worst-case multi-commodity flow is obtained by adding additional variable $f$ and maximizing  $f$, while demanding that every single-commodity flow is greater or equal to $f$. Hence $f$ corresponds to the least flow between any source sink pair. This provides us with the following linear program \cite{shahrokhi1990maximum}:
\begin{align}
f^{\text{worst}}_{\max}\left(G',\{c'(\{vw \}) \}_{\{vw\}\in E'}\right)=\textrm{max} &\ f\\
\forall i:\ &\ f-\sum_{v:\{s_iv\}\in E'} \left(f^{(i)}_{s_iv}-f^{(i)}_{vs_i}\right)\leq0\nonumber\\
\forall \{vw\}\in E':\ &\  \sum_{i=1}^r\left(f^{(i)}_{vw}+f^{(i)}_{wv}\right)\le  c'(\{vw\})\nonumber\\
\forall \{vw\}\in E', \ \forall i:\ &\  f^{(i)}_{vw}, f^{(i)}_{wv}\geq 0\nonumber\\
\forall i,\ \forall w\in V,w\neq s_i,t_i:\ &\ \sum_{v:\{vw\}\in E'} \left(f^{(i)}_{vw} -f^{(i)}_{wv}\right)=0.\nonumber
\end{align}
Again, we can define
\begin{equation}\label{LP4bar}
\bar{f}^{\text{worst}}_{\max}(G',\{C(\mc{N}^e)\}_{e\in E})=\max_{0\leq p_e \leq 1,\ \sum_ep_e=1}f^{\text{worst}}_{\max}(G',\{c'_{C}(\{vw\},p_{wv},p_{vw})\}_{\{vw\}\in E'}).
\end{equation}

Next, we will consider generalisations of the max-flow min-cut theorem to multiple source sink pairs: Given source sink pairs $(s_1,t_1), \cdots  ,(s_r,t_r)$, one can define a multicut $\{S\}\leftrightarrow \{T\}$ as a set of edges in $E'$ whose removal disconnects all source sink pairs and the capacity of a multicut as the sum over the capacity of its edges $\{S\}\leftrightarrow \{T\}$, namely
\begin{equation}
c'(\{S\}\leftrightarrow \{T\})=\sum_{\{vw\}\in \{S\}\leftrightarrow \{T\}} c'(\{vw\}).
\end{equation}
Whereas there is no known exact max-flow minimum cut-ratio theorem in the case of multiple flows, there exists a relation between the minimum multicut and the maximum total multi-commodity flow up to a factor $g_\text{t}(r)$ that scales as $\mc{O}(\log r)$ \cite{garg1996approximate}. Namely it holds 
\begin{equation}\label{maxFlowMinMultiCut}
f^{\text{total}}_{\max}\left(G',\{c'(\{vw \}) \}_{\{vw\}\in E'}\right)\leq \min_{\{S\}\leftrightarrow \{T\}}c'(\{S\}\leftrightarrow \{T\})\leq g_\text{t}(r)f^{\text{total}}_{\max}\left(G',\{c'(\{vw \}) \}_{\{vw\}\in E'}\right),
\end{equation}
An example of the relation Eq. (\ref{maxFlowMinMultiCut}) is given in figure \ref{fig:multiunicast} (a) and (b). In the example $g_{\rm t}(r)=1$. In the case of the maximum worst-case multi-commodity flow there exists a similar relation with the minimum cut ratio, which is defined as
\begin{equation}\label{cutratio}
R_{\min}\left(G',\{c'(\{vw \}) \}_{\{vw\}\in E'}\right)=\min_{V'\subset V}\frac{\sum_{\{vw\} \in \partial V'}c'(\{vw\})}{d(\partial (V'))},
\end{equation} 
where the minimization is over (bipartite) cuts $\partial V'$ and
%\begin{equation}\label{bipcuts}
%\partial V'=\{\{v_1v_2\}:v_1\in V_1\subset V,v_2\in V_2=V\setminus V_1 \}
%\end{equation}
\begin{equation}
d(\partial (V'))=\left|\{ i : (s_i \in V', t_i \in V \setminus V') \vee  (t_i \in V', s_i \in V \setminus V') \}\right|%\sum_{s_i\in V',t_i\in V\setminus V' \lor t_i\in V',s_i\in V\setminus V'}1
\end{equation}
describes the demand across a cut $\partial V'$. Note that in the case of only one source sink pair the minimum cut ratio Eq. (\ref{cutratio}) reduces to the min-cut Eq. (\ref{mincut}). Whereas there is no known exact max-flow minimum cut-ratio theorem in the case of multiple flows, there is a relation up to some factor $g_{\text{w}}(r)$ \cite{leighton1999multicommodity},
\begin{equation}\label{maxFlowMinCutRatio}
f^{\text{worst}}_{\max}\left(G',\{c'(\{vw \}) \}_{\{vw\}\in E'}\right)\leq R_{\min} \left(G',\{c'(\{vw \}) \}_{\{vw\}\in E'}\right)\leq g_{\text{w}}(r)f^{\text{worst}}_{\max}\left(G',\{c'(\{vw \}) \}_{\{vw\}\in E'}\right).
\end{equation}
An example of the relation Eq. (\ref{maxFlowMinCutRatio}) is given in figure \ref{fig:multiunicast} (c) and (d). In the example $g_{\text{w}}(r)=1$. The gap $g_{\text{w}}(r)$ is known as the flow-cut gap. In \cite{leighton1999multicommodity} it has been shown to be of $\mc{O}(\log |E|)$. This was then improved to $\mc{O}(\log r)$, where $r$ is the number of source sink pairs, in \cite{linial1995geometry,aumann1998log}. In the case of overlapping source and sink vertices, i.e.  $s_i=s_j$, $s_i=t_j$, $t_i=s_j$ or $t_i=t_j$ for some $i\neq j$, the flow-cut gap has further been improved to $\mc{O}(\log r^*)$, where $r^*$ is the size of the smallest set of vertices that contains at least one of such $s_i$ or $t_i$ for all $i=1,...,r$ \cite{gunluk2007new}. For a number of particular classes of graphs, it has been shown that the flow-cut gap can even be of $\mc{O}(1)$ \cite{gupta2004cuts,chekuri2006embedding,lee2009geometry,chakrabarti2012cut,salmasi2017constant}.

\subsection{Multipartite target states}

In this section we will explicitly define all quantities that occur in our result for multipartite target states, Theorem \ref{Theo_multi}. We also briefly introduce the concept of Steiner cuts and Steiner trees, which are used in the proof of Theorem \ref{Theo_multi}.

Again, we begin with the definition of the capacities: Given a set $S\subset V$ of users that wish to establish a GHZ or multipartite private state, the multipartite quantum and private network capacities are defined as
\begin{align}
&\mathcal{Q}^S\left(G,\{\mc{N}^e\}_{e\in E}\right)=\max_{p_e\geq0,\sum_ep_e=1}\lim_{\epsilon\to0}\lim_{n\to\infty}\sup_{\Lambda}\left\{\frac{\langle\log d^{(k)}\rangle_{k}}{n}:\left\|\rho_{M_{s_1}...M_{s_l}}^{(n,k)}-\Phi^{d^{(k)}}_{M_{s_1}...M_{s_l}}\right\|_1\leq\epsilon\right\},\label{multicap2}\\
&\mathcal{P}^S\left(G,\{\mc{N}^e\}_{e\in E}\right)=\max_{p_e\geq0,\sum_ep_e=1}\lim_{\epsilon\to0}\lim_{n\to\infty}\sup_{\Lambda}\left\{\frac{\langle\log d^{(k)}\rangle_{k}}{n}:\left\|\rho_{K_{s_1}S_{s_1}...K_{s_l}S_{s_l}}^{(n,k)}-\gamma^{d^{(k)}}_{K_{s_1}S_{s_1}...K_{s_l}S_{s_l}}\right\|_1\leq\epsilon\right\},\label{mulricap1}
\end{align}
where the suprema are over all adaptive $(n,\epsilon, \{p_e\}_{e\in E})$ protocols $\Lambda$. As the class of multipartite private states includes GHZ states, the multipartite private capacity is an upper bound on the multipartite quantum capacity.

Let us mow introduce the concept of Steiner cuts and Steiner trees: For a subset $S\subset V$ of vertices in $G'$ we define a Steiner cut with respect to $S$, in short $S$-cut, as a cut $\partial (V_S)$ with respect to a set $V_S\subset V$ such that there is at least one pair of vertices $s_i,s_j\in S$ with $s_i\in V_S$ and $s_j\in V\setminus V_S$. When considering a minimization of the capacity over all $S$-cuts, we can divide the minimization into a minimization over pairs of vertices in $S$ and a minimization over cuts separating the pairs,
\begin{equation}
\min_{V_S}\sum_{\{vw\}\in \partial (V_S)}c'(\{vw\})=\min_{s_i,s_j\in S,s_i\neq s_j}\min_{V_{s_i;s_j}}\sum_{\{vw\}\in \partial(V_{s_i;s_j})}c'(\{vw\}),
%\min_{\substack{\partial V' \\ S\text{-cut}}}\sum_{\{vw\}\in \partial V'}c'(\{vw\})=\min_{s_i,s_j\in S,s_i\neq s_j}\min_{V_{s_i}\leftrightarrow V_{s_j}}\sum_{\{vw\}\in V_{s_i}\leftrightarrow V_{s_j}}c'(\{vw\}),
\end{equation}
where $\min_{V_S}$ is a minimization over all $V_S\subset V$ such that there is at least one pair of vertices $s_i,s_j\in S$ with $s_i\in V_S$ and $s_j\in V\setminus V_S$. Further $\min_{V_{s_i;s_j}}$ is a minimization over all $V_{s_i;s_j}\subset V$ such that $s_i\in V_{s_i;s_j}$ and $s_j\in V\setminus V_{s_i;s_j}$.
%where the notation $V_{s_i}\leftrightarrow V_{s_j}$ means an $s_is_j$-cut as defined by Eq. (\ref{eq:stcut}). 
Note that, as $\min_{s_i,s_j\in S,s_i\neq s_j}\min_{V_{s_i;s_j}}\sum_{\{vw\}\in \partial(V_{s_i;s_j})}c'(\{vw\})$ % $\min_{V_{s_i}\leftrightarrow V_{s_j}}\sum_{\{vw\}\in V_{s_i}\leftrightarrow V_{s_j}}c'(\{vw\})$ 
does not depend on the order, we can, without loss of generality restrict to disjoint $s_i$ and $s_j$ with $j>i$, reducing the number of resources needed in the outer minimization. We can then apply the max-flow min-cut theorem Eq. (\ref{maxFlowminCut}) to the inner minimization,
\begin{equation}\label{minSCutMaxSflow}
\min_{s_i,s_j\in S,s_i\neq s_j}\min_{V_{s_i;s_j}}\sum_{\{vw\}\in \partial(V_{s_i;s_j})}c'(\{vw\})
%\min_{s_i,s_j\in S,s_i\neq s_j}\min_{V_{s_i}\leftrightarrow V_{s_j}}\sum_{\{vw\}\in V_{s_i}\leftrightarrow V_{s_j}}c'(\{vw\})
=\min_{\substack{s_i,s_j\in S,s_i\neq s_j\\j> i}}f^{s_i\to s_j}_{\max}(G',\{c'(\{wv \}) \}_{\{wv\}\in E'}),
\end{equation}
where $f^{s_i\to s_j}_{\max}(G',\{c'(\{wv \}) \}_{\{wv\}\in E'})$ is given by LP Eq. (\ref{LP1}). As there are finitely many disjoint $s_i,s_j$-pairs in $S$, we could solve $f^{s_i\to s_j}_{\max}(G',\{c'(\{wv \}) \}_{\{wv\}\in E'})$ for every pair and then find the smallest solution. A more efficient way is to introduce flow variables $f_e^{(ij)}$ for every disjoint $s_i,s_j$-pair (and every edge) and maximize a slack variable $f$, while requiring the flow value for every $s_i,s_j$-pair to be greater or equal than $f$ and all other constraints of LP Eq. (\ref{LP1}) to be fulfilled for every disjoint $s_i,s_j$-pair:

\begin{equation}\label{eq:connectivity}
\min_{\substack{s_i,s_j\in S,s_i\neq s_j\\j> i}}f^{s_i\to s_j}_{\max}(G',\{c'(\{wv \}) \}_{\{wv\}\in E'})=f^{S}_{\max}\left(G',\{c'(\{vw \}) \}_{\{vw\}\in E'}\right),
\end{equation}
where
\begin{align}
f^{S}_{\max}\left(G',\{c'(\{vw \}) \}_{\{vw\}\in E'}\right)=\textrm{max} &\ f\label{LP5}\\
\forall i,j> i:\ &\ f-\sum_{v:\{s_iv\}\in E'} \left(f^{(ij)}_{s_iv}-f^{(ij)}_{vs_i}\right)\leq0\label{LP5a}\\
\forall i,j> i,\{vw\}\in E':\ &\  f^{(ij)}_{vw}+f^{(ij)}_{wv}\le  c'(\{vw\})\label{LP5b}\\
\forall i,j> i,\forall \{vw\}\in E':\ &\  f^{(ij)}_{vw},f^{(ij)}_{wv}\geq 0\label{LP5c}\\
\forall i,j> i,\ \forall w\in V,w\neq s_i,s_j:\ &\ \sum_{v:\{vw\}\in E'} \left(f^{(ij)}_{vw} -f^{(ij)}_{wv}\right)=0.\label{LP5d}
\end{align}
Adding a maximization over usage frequencies, we obtain
\begin{equation}\label{LP5bar}
\bar{f}^{S}_{\max}(G',\{C(\mc{N}^e)\}_{e\in E})=\max_{0\leq p_e \leq 1,\ \sum_ep_e=1}f^{S}_{\max}(G',\{c'_{C}(\{vw\},p_{wv},p_{vw})\}_{\{vw\}\in E'}).
\end{equation}

It will be convenient to introduce an undirected multigraph $G_{\lfloor c' \rfloor}''$, by replacing each edge $\{vw\}\in E'$ with $\lfloor c'(\{vw\}) \rfloor$ identical edges with unit-capacity connecting $v$ and $w$. An $S$-cut in an undirected unit-capacity multigraph $G''$ is defined as a set of edges whose removal disconnects at least two vertices in $S$. The size $\lambda_S(G'')$ of the minimum $S$-cut in $G''$ is called the $S$-connectivity of $G''$.

In $G''$ we can also define a Steiner tree spanning $S$, in short $S$-tree, as a subgraph of $G''$ that contains all vertices in $S$ and is a tree, i.e. does not contain any cycles. If $S$ only consists of two vertices, we call an $S$-tree a path. We call two Steiner trees edge-disjoint, if they do not contain a common edge. The problem of finding the number $t_S(G'')$ of edge-disjoint Steiner trees in a general undirected multigraph is NP-complete \cite{cheriyan2006hardness}. However, there is a connection between $S$-connectivity and the number of edge-disjoint $S$-trees in an undirected unit-capacity multigraph \cite{kriesell2003edge,lau2004approximate,petingi2009packing}:
\begin{equation}\label{KreiselConj}
t_S(G'')\geq\lfloor g_1 \lambda_S(G'')\rfloor-g_2.
\end{equation}
In \cite{kriesell2003edge} it has been conjectured that Eq. (\ref{KreiselConj}) holds for $g_1=\frac{1}{2}$ and $g_2=0$. In \cite{lau2004approximate} it has been shown that the relation holds for $g_1=\frac{1}{26}$ and $g_2=0$, whereas the authors of \cite{petingi2009packing} show that it holds for $g_1=\frac{1}{2}$ and $g_2=\frac{|V\setminus S|}{2}+1$, which is finite in the graphs we are considering.

\subsection{On complexity}

Let us briefly discuss the computational complexity of our linear programs Eq. (\ref{LP1bar}), Eq. (\ref{LP3bar}), Eq. (\ref{LP4bar}) and Eq. (\ref{LP5bar}). Using interior point methods, e.g. \cite{ye1991n3l}, a linear program in standard form
\begin{align}\label{LPs}
\textrm{min} &\ c^Tx\\
& Ax=b, \ x\geq 0,\nonumber
\end{align}
where $c,x\in\mathbb{R}^N$, $b\in\mathbb{R}^M$ and $A\in\mathbb{R}^{M\times N}$, can be solved using ${\cal O}(\sqrt{N}L)$ iterations and ${\cal O}(N^3L)$ total arithmetic operations. Here $L$ is the size of the problem data, $A,b,c$, which scales as ${\cal O}(MN+M+N)$ \cite{wright1997primal}. If we assume $A$ to be of full rank, it holds $M\leq N$, and hence, $L$ scales as ${\cal O}(N^2)$. Using slack variables \cite{ye1991n3l}, all inequality constraints in our linear programs can be converted into equality constraints. Linear equality constraints can be easily written in the form $Ax=b$. Hence $N$ can be obtained by adding the number of variables and the number of inequality constraints in our linear programs. 

For LP Eq. (\ref{LP1bar}) we have $N=3|E'|+|E|$. LP Eq. (\ref{LP3bar}) has $2r|E'|+|E|$ variables and $|E'|+|E|$ inequality constraints. Thus $N=(2r+1)|E'|+2|E|$ for LP Eq. (\ref{LP3bar}). LP Eq. (\ref{LP4bar}) has $2r|E'|+|E|+1$ variables and $|E'|+|E|+r$ inequality constraints. Thus $N=(2r+1)|E'|+2|E|+1+r$ for LP Eq. (\ref{LP4bar}). LP Eq. (\ref{LP5bar}) has $2{{|S|}\choose{2}}|E'|+|E|+1$ variables and ${{|S|}\choose{2}}|E'|+|E|+{{|S|}\choose{2}}$ inequality constraints. Thus $N=3{{|S|}\choose{2}}|E'|+2|E|+1+{{|S|}\choose{2}}$ for LP Eq. (\ref{LP5bar}). Hence, all our linear programs, the number of iterations as well as the number of total arithmetic operations scale polynomially with the size of the network. 

\newpage
\section{Supplementary Information}
%%%%%%%%%%%%%%%%%%%%%%%%%%%%%%%%%%%%%%%%%
\subsection{Preliminaries}\label{sec:main}

Let a quantum network be given by a directed graph $G=(V,E)$, where $V$ denotes the set of the finite vertices and $E$ the set of the finite directed edges, which represent quantum channels. Each directed edge $e\in E$ has tail $v\in V$ and head $w\in V$.  We also denote $e$ by $vw$.  $\mc{N}^e=\mc{N}^{vw}$ corresponds to a channel with input in $v$ and output in $w$. We can also assign graph theoretic capacity functions $c:E\to\mathbb{R}^+_0$ to each edge. We assume that each vertex has the capability to store and process quantum information locally and that all vertices are connected by public lines of classical communication, the use of both of which is considered to be a free resource. Let us assume there is a subset $U\subset V$ of the vertices, the users who wish to establish a target state $\theta$ containing the desired resource, whereas the remaining vertices serve as repeater stations. In the following section we will elaborate on the exact form of $\theta$.

We assume that initially there is no entanglement between any of the vertices. In order to obtain $\theta$, all vertices apply an adaptive protocol consisting of (generally probabilistic) local operations and classical communications (LOCC) among the nodes in the network interleaved by channel uses. In particular, during each round of LOCC it is determined which channel is used next and which state is inserted into the channel \cite{AML16,pirandola2016capacities}. We describe a protocol by given upper bound $n_e$ on the average of the number of uses of each channel ${\cal N}^e$, which is associated with a set of usage frequencies $\{p_e\}_{e\in E}$, where $p_e := n_e/ n (\ge 0)$ of each channel ${\cal N}^e$ for a single parameter $n$ which can be regarded as time or an upper bound on the average of total channel uses with $\sum_{e \in E} p_e =1$ (see \cite{AK16}), and an error parameter $\epsilon$ such that after the final round of LOCC a state $\epsilon$-close in trace distance to $\theta$ is obtained. By average we mean that parameters of a protocol are averaged over all possible LOCC outcomes. We call such a protocol an $(n,\epsilon, \{p_e\}_{e\in E})$ adaptive protocol. In the asymptotic limit where $n\to\infty$ it then holds ${n}_e\to\infty$ for edge $e$ with $p_e>0$ while $\{p_e\}_{e\in E}$ remains fixed \cite{AK16}.

Note that whereas quantum channels are directed, the direction does not play a role when we use them to distribute entanglement under the free use of (two-way) classical communication. For example, if a channel is used to distribute a Bell state, it is invariant under permutations of nodes across the channel. This motivates the introduction of an \emph{undirected} graph $G'=(V,E')$, where $E'$ is obtained from $E$ as follows:
For any edge $vw\in E$ with $wv \in E$, the directed edges $vw$ and $wv$ are replaced by single undirected edge $\{vw\}$ (or, equivalently $\{wv\}$) with $c'(\{vw\})=c(vw)+c(wv)$, while, for any edge $vw\in E$ with $wv \notin E$, the directed edge $vw$ is replaced by undirected edge $\{vw\}$ with $c'(\{vw\})=c(vw)$.

%This motivates the introduction of an \emph{undirected} graph $G'=(V,E')$, where $E'$ is obtained from $E$ as follows: If for an edge $vw\in E$ there exists no edge $wv\in E$, we replace $vw$ by an undirected edge $\{vw\}$ (or, equivalently $\{wv\}$) with capacity $c'(\{vw\})=c(vw)$. If for $vw\in E$ there exists an edge $wv\in E$ we merge the two edges, replacing them by one undirected edge $\{vw\}$ with capacity $c'(\{vw\})=c(vw)+c(wv)$. 

In order to describe networks consisting only of Bell states, it will also be convenient to introduce an undirected unit-capacity multigraph $G_{\lfloor c'\rfloor}''=(V,E_{\lfloor c'\rfloor}'')$, which we derive from $G'$ by replacing every edge $\{vw\}$ in $G$ by $\lfloor c'(\{vw\})\rfloor$ unit-capacity undirected edges linking $v$ and $w$. See Supplementary Figure \ref{figG} for an example of the various graphs.

\begin{figure}

\centering
(a) \includegraphics[width=0.3\textwidth, trim={1cm 10cm 10cm 2cm},clip=true]{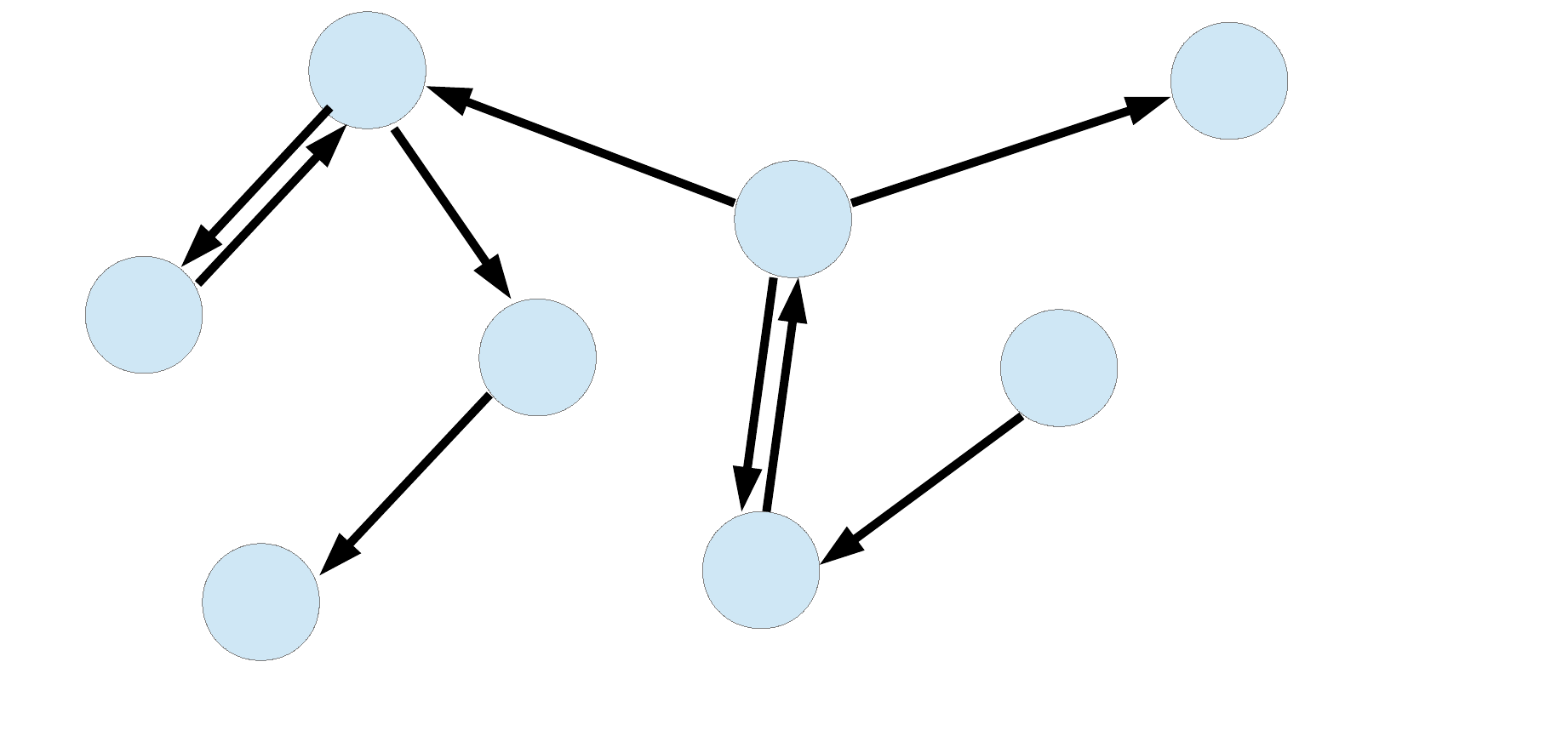}     
(b)\includegraphics[width=0.3\textwidth, trim={1cm 10cm 10cm 2cm},clip=true]{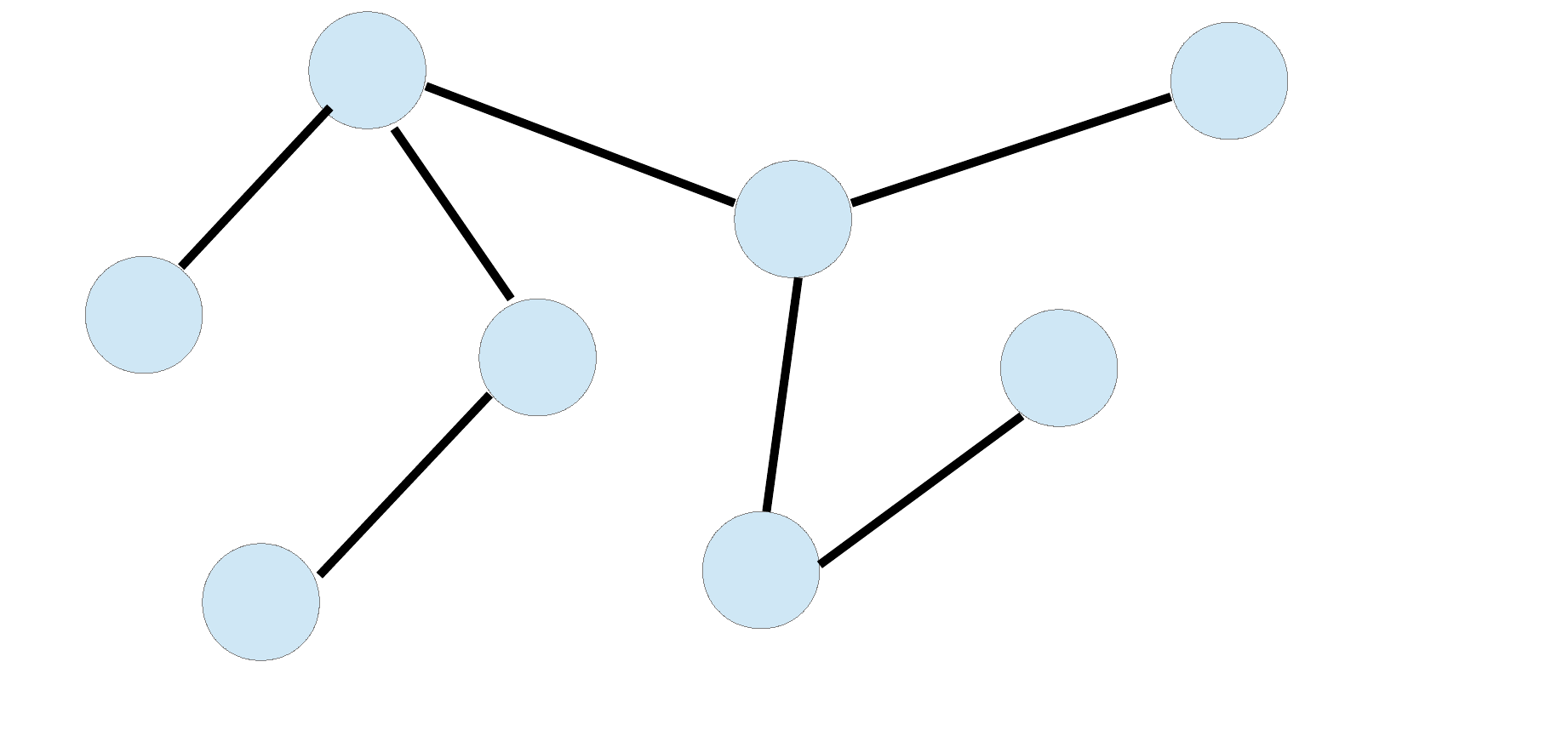}     
(c) \includegraphics[width=0.3\textwidth, trim={1cm 10cm 10cm 2cm},clip=true]{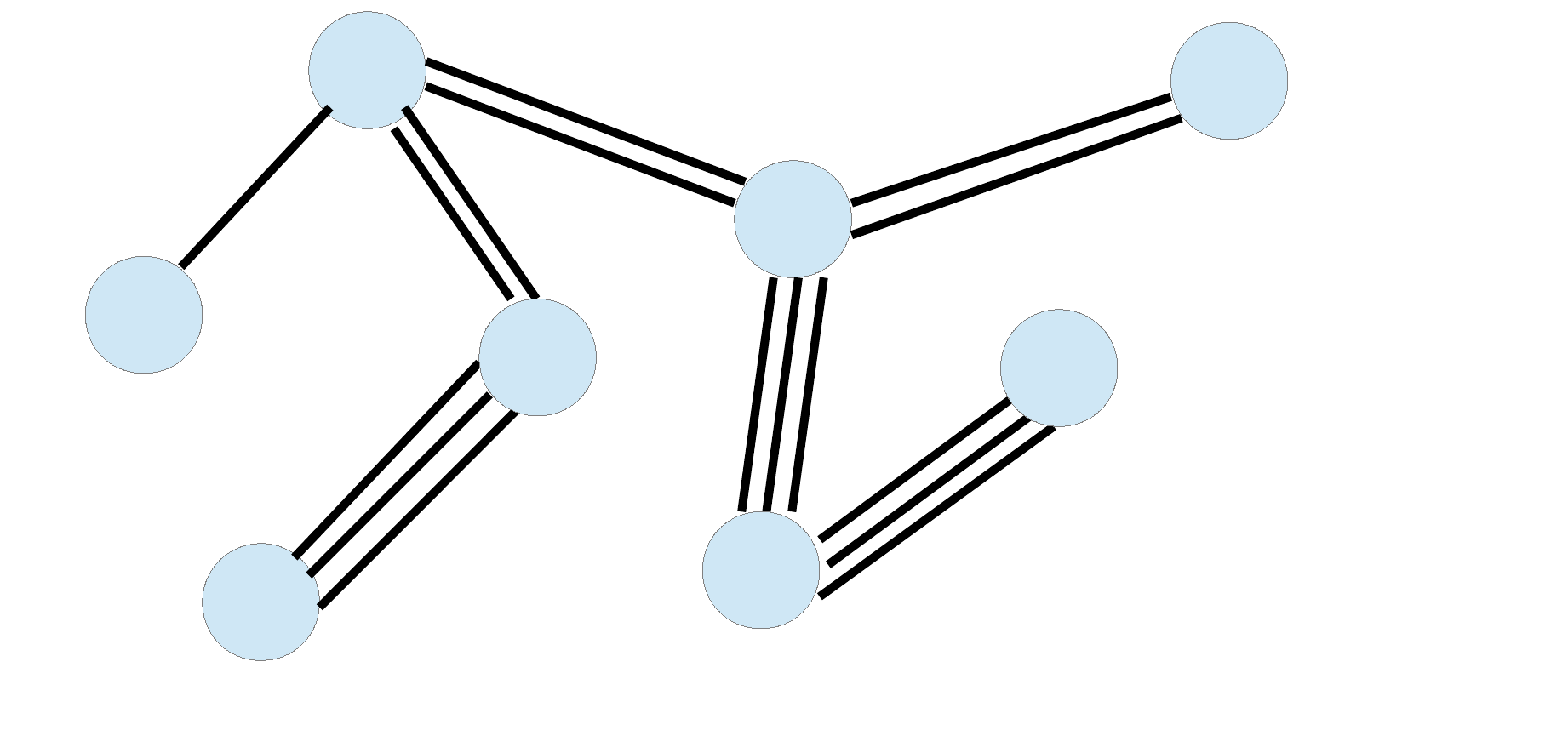}     
\caption{Our notation for graphs: (a) Directed graph $G$, with the edges representing quantum channels. (b) Undirected graph $G'$ as defined. (c) Undirected multigraph $G''$, with the edges corresponding to Bell states. For the definitions of $G$, $G'$ and $G''$, see \ref{sec:main}.}\label{figG}
\end{figure}

\subsection{Bipartite user scenario}\label{sec:bip}

Let us suppose that the set $U$ of users only contains two vertices, $s\in E$, called Alice, and $t\in E$, called Bob. We can now define a \emph{network capacity with respect to target state $\theta$ with fixed average usage frequencies $\{p_e\}_{e\in E}$} as the maximum asymptotic average rate at which we can obtain the target state by means of adaptive operations,
\begin{equation}\label{cap1}
\mathcal{C}^{\theta}_{\{p_e\}_{e\in E}}\left(G,\{\mc{N}^e\}_{e\in E}\right)=\lim_{\epsilon\to0}\lim_{n\to\infty}\sup_{\Lambda(n,\epsilon, \{p_e\}_{e\in E})}\left\{\frac{\langle\log d^{(k)}\rangle_{k}}{n}:\left\|\rho_{X_sX_t}^{(n,k)}-\theta^{d^{(k)}}_{X_sX_t}\right\|_1\leq\epsilon\right\},
\end{equation}
where the supremum is over all adaptive $(n,\epsilon, \{p_e\}_{e\in E})$ protocols $\Lambda$. Further $k=(k_1,  \ldots  ,k_{m+1})$ is a vector keeping track of outcomes of the $m+1$ LOCC rounds in $\Lambda$, the averaging, denoted by the parenthesis $\langle...\rangle_{k}$, is over all those outcomes and $\rho_{X_sX_t}^{(n,k)}$ is the final state of $\Lambda$ for given outcomes $k$. A more general quantity will be the \emph{network capacity with respect to target state} $\theta$, which is obtained by maximizing Supplementary Equation (\ref{cap1}) over the usage frequencies,
\begin{equation}\label{cap2}
\mathcal{C}^{\theta}\left(G,\{\mc{N}^e\}_{e\in E}\right)=\max_{p_e\geq0,\sum_ep_e=1}\mathcal{C}_{\{p_e\}_{e\in E}}^{\theta}\left(G,\{\mc{N}^e\}_{e\in E}\right).
\end{equation}
If the target is a maximally entangled state $\Phi^d_{M_sM_t}$, we refer to Supplementary Equations (\ref{cap1}) and (\ref{cap2}) as \emph{quantum capacities} of the network, while if it is a private state $\gamma^d_{K_sS_sK_tS_t}$, we refer to them as \emph{private capacities} of the network. Respectively, we also use the notation $\mc{Q}:=\mc{C}^\Phi$ and $\mc{P}:=\mc{C}^\gamma$.

Let us discuss the qualitative difference between the capacities given by Supplementary Equations (\ref{cap1}) and (\ref{cap2}). In the first scenario of Supplementary Equation (\ref{cap1}), the channel frequencies are fixed. In this case the optimization task reduces to finding the protocol that achieves the largest rate per channel use while using the channels with the given frequencies. This could be, in practice, related with the rate of entanglement distribution per time $n$. In particular, in this scenario, $p_e$ represents how frequently the use of channel $\mc{N}^e$ occurs for a given time $n$. In the scenario of Supplementary Equation (\ref{cap2}), on the other hand, we allow the users to choose the frequency usage with $\sum_{e \in E} p_e=1$ for maximizing the rate at which the desired communication task can be performed. This would be meaningful whenever we want to minimize the total channel uses $n=\sum_e n_e $ to obtain one resource state. The number $n$ of channel uses could be related with a cost (or usage fee), which has to be minimized. Note that, the solution for Supplementary Equation (\ref{cap2}) is always achieved by a single path of repeaters with (in general) different frequencies on the path. Both scenarios could correspond to the implementation of a communication task over a quantum network that belongs to an external provider who charges per time or number of channel uses.
\subsubsection{Upper bounding the capacity}
We will now show that the capacities given by Supplementary Equations (\ref{cap1}) and (\ref{cap2}) can be upper bounded by linear programs. In  \cite{rigovacca2018versatile}, it has been shown that with any bipartite entanglement measure $\cal E$ that satisfies the properties
\begin{enumerate}
\item{if $\|\rho_{AB}-\theta^d_{AB}\|_1=\epsilon$, there exist real functions $f$ and $g$, satisfying $\lim_{\epsilon\to0}f(\epsilon)=0$ and $\lim_{\epsilon\to0}g(\epsilon)=1$,  such that ${\cal E}(\rho_{AB})\geq g(\epsilon)\log d-f(\epsilon)$,}
\item{for $\tilde\rho_{AB'B}=\mathcal{N}_{A'\to B'}(\rho_{AA'B})$, it holds 
\begin{equation}
{\cal E}(\tilde\rho_{AB'B})\leq {\cal E}(\mathcal{N})+{\cal E}(\rho_{AA'B}),
\end{equation}
where ${\cal E}(\mathcal{N})=\max_{\rho_{AA'}}{\cal E}(\mathcal{N}_{A'\to B'}(\rho_{AA'}))$, i.e. $\cal E$ cannot be increased by amortization \cite{kaur2017amortized},}
\end{enumerate}
it is possible to upper bound the fixed usage frequencies the capacity given by Supplementary Equation (\ref{cap1}) as follows:
\begin{equation}\label{eq:upbound}
\mathcal{C}^{\theta}_{\{p_e\}_{e\in E}}\left(G,\{\mc{N}^e\}_{e\in E}\right)\leq\min_{V_{s;t}} \sum_{e\in E:\{e\} \in \partial(V_{s;t})} p_{e} {\cal E}(\mc{N}^{e}),
\end{equation}
where the minimization is over all $V_{s;t}\subset V$ such that $s\in V_{s;t}$ and $t\in V\setminus V_{s;t}$, and $\partial(V_{s;t})$ is an $st$-cut as defined in equations (13)-(14) of the Methods section. In the case of private states, entanglement measures satisfying properties 1 and 2 include the squashed entanglement $E_{\sq}$ \cite{TGW14,takeoka2014squashed,wilde2016squashed}, the max-relative entropy of entanglement $E_{\max}$ \cite{christandl2017relative} and, for a teleportation simulable/Choi-stretchable channels \cite{BDSW96,HHH99,Mul12,PLOB17}, the relative entropy of entanglement $E_R$ \cite{horodecki2009general,PLOB17}.

Using the max-flow min-cut theorem (see the Methods sections for details), the question of finding an upper bound on the network capacity is reduced to analyzing the entangling properties of single channels combined with finding the maximum flow from $s$ to $t$ in the graph $G'$, which can be formulated as a linear program (LP) of the form
\begin{align}
f^{s\to t}_{\max}(G',\{c'(\{wv\})\}_{\{wv\}\in E'})=\textrm{max}& \sum_{v:\{sv\}\in E'} (f_{sv}-f_{vs})\label{LP2}\\
&\forall \{vw\}\in E':\ f_{wv}+f_{vw}\leq c'(\{wv\})\label{LP2a}\\
&\forall w\in V: w\neq s,t,\ \sum_{v:\{vw\}\in E'} f_{vw} = \sum_{v:\{vw\}\in E'}f_{wv}\label{LP2b}
\end{align}
where the maximization is over edge flows $f_{vw}\geq 0$ and $f_{wv}\geq 0$ for all edges $\{vw\}\in E'$ and $c'(\{wv\})$ denotes an edge capacity of edge $\{wv\}$. Let us discuss the the LP given by Supplementary Equations (\ref{LP2})-(\ref{LP2b}) in detail: The objective in Supplementary Equation (\ref{LP2}) is the net sum of all outgoing flow at the source $s$. Namely, we maximize the sum over all vertices $v$, that are adjacent to the source vertex $s$, of the differences between the corresponding outgoing flows $f_{sv}$ and incoming flows $f_{vs}$. The constraint given by Supplementary Equation (\ref{LP2a}) is the capacity constraint. Namely, for every undirected edge $\{vw\}$ the sum of flows that pass through the edge in both directions have to be less than the capacity. The constraint given by Supplementary Equation (\ref{LP2b}) is the edge flow conservation constraint, ensuring that for every vertex other than the source or sink the sum of outgoing all flows is equal to the sum of all incoming flows. Because of the constraint given by Supplementary Equation (\ref{LP2b}), the net sum of outgoing flow at the source is equal to the flow from source to sink. Let us now set the edge capacities to
\begin{equation}\label{cE}
c'(\{wv\})=c'_{\cal E}(\{wv\}, p_{vw},p_{wv}):=p_{wv}{\cal E}(\mc{N}^{wv})+p_{vw}{\cal E}(\mc{N}^{vw}),
\end{equation}
for all $\{vw\}\in E'$, where $p_{vw}{\cal E}(\mc{N}^{vw})=0$ if $vw\notin E$. Note that Supplementary Equation (\ref{cE}) is a linear function in $p_{vw}$ and $p_{wv}$. Hence, we can formulate the following result:

\begin{proposition}\label{uni:upper}
For any entanglement measure $\cal E$ satisfying properties 1 and 2, it holds
\begin{equation}
\mathcal{C}^{\theta}_{\{p_e\}_{e\in E}}\left(G,\{\mc{N}^e\}_{e\in E}\right)\leq f^{s\to t}_{\max}(G',\{c'_{\cal E}(\{wv\},p_{vw},p_{wv})\}_{\{wv\}\in E'}).
\end{equation}
\end{proposition}

In order to upper bound capacities of the form of Supplementary Equation (\ref{cap2}), we can include an optimization over all usage frequencies $p_e$ into the optimization given by Supplementary Equations (\ref{LP2})-(\ref{LP2b}). To this end we treat the $p_e$ as variables, define
\begin{equation}\label{fbar}
\bar{f}^{s\to t}_{\max}(G',\{{\cal E}(\mc{N}^e)\}_{e\in E})=\max_{p_e\geq0,\sum_ep_e=1} f^{s\to t}_{\max}(G',\{c'_{\cal E}(\{wv\},p_{vw},p_{wv})\}_{\{wv\}\in E'}),
\end{equation}
and obtain

\begin{corollary}\label{Corr1}
For a network described by a finite directed graph $G$ and an undirected graph $G'$ as defined above, it holds
\begin{equation}
\mathcal{C}^{\theta}\left(G,\{\mc{N}^e\}_{e\in E}\right)\leq \bar{f}^{s\to t}_{\max}(G',\{{\cal E}(\mc{N}^e)\}_{e\in E}),
\end{equation}
for any entanglement measure $\cal E$ satisfying properties 1 and 2.
\end{corollary}

\subsubsection{Lower bounding the capacity}

For the lower bound, we set the edge capacities in the LP given by Supplementary Equations (\ref{LP2})-(\ref{LP2b}) to
\begin{equation}\label{cQ}
c'_{Q^\leftrightarrow}(\{wv\}, p_{vw},p_{wv})=p_{wv}{Q^\leftrightarrow}(\mc{N}^{wv})+p_{vw}{Q^\leftrightarrow}(\mc{N}^{vw}),
\end{equation}
for all $\{vw\}\in E'$, where $p_{vw}{Q^\leftrightarrow}(\mc{N}^{vw})=0$ if $vw\notin E$. We can now show

\begin{proposition}\label{lowerbound}
For a network described by a finite directed graph $G=(E,V)$ it holds\begin{equation}
\mathcal{Q}_{\{p_e\}_{e\in E}}\left(G,\{\mc{N}^e\}_{e\in E}\right)\geq f^{s\to t}_{\max}(G',\{c'_{Q^\leftrightarrow}(\{wv\},p_{vw},p_{wv})\}_{\{vw\}\in E'}),
\end{equation}
where $G'$ is defined above.% and the r.h.s. is defined analogously to Supplementary Equation (\ref{LP2})-(\ref{LP2b}).
\end{proposition}

For distillable channels, which is a subclass of Choi-stretchable channels, it has been shown that $Q^\leftrightarrow=E_R$ \cite{PLOB17}. Hence combining Proposition \ref{uni:upper} and Proposition \ref{lowerbound} provides
\begin{corollary}
For a network described by a finite directed graph $G=(E,V)$ with distillable channels $\mc{N}^e$ for $e\in E$ it holds
\begin{equation}
\mathcal{Q}_{\{p_e\}_{e\in E}}\left(G,\{\mc{N}^e\}_{e\in E}\right)= f^{s\to t}_{\max}(G',\{c'_{Q^\leftrightarrow}(\{wv\},p_{vw},p_{wv})\}_{\{vw\}\in E'}).
\end{equation}
\end{corollary}

Before proving Proposition \ref{lowerbound}, we need to define a directed path from $s$ to $t$ in $G'$ as
\begin{equation}
P_{s\to t}=\{v_jv_{j+1}:i,j=1,\ldots,l-1,v_j\in V,\{v_jv_{j+1}\}\in E',v_1=s,v_{l}=t,v_j \neq v_i, i \neq j\}.
 \end{equation}
In a finite graph with nonzero flow $f^{s\to t}$, obtained from a solution $\{f_{vw},f_{wv} \}_{ \{vw\} \in E' }$ of the LP given by Supplementary Equations (\ref{LP2})-(\ref{LP2b}), we can always find a finite number $N$ of paths $P_{s\to t}^{(1)},...,P_{s\to t}^{(N)}$ whose flow consists of path-flows $f^{(1)},...,f^{(N)}$ such that $f^{s\to t}=\sum_{i=1}^Nf^{(i)}$  \cite{ford1956maximal}\footnote{The authors of \cite{ford1956maximal} use the notion of chains (i.e. undirected paths) linking $s$ and $t$. Here we add a direction such that they are directed from $s$ to $t$.}. It will be convenient to define for every path $P_{s\to t}^{(i)}$ and every edge $vw\in E$ the quantity
\begin{equation}\label{eq:fvwi}
f_{vw}^{(i)}:=\begin{cases}
&f^{(i)}\text{ if }vw\in P_{s\to t}^{(i)}\\&0\text{ else.}
\end{cases}
\end{equation}
Note that an edge $vw$ can be part of more than one paths. The sum of path-flows passing through the edge, however, has to be upper bounded by the edge flow $f_{vw}$. Hence it holds
\begin{equation}\label{eq:sumfvwi}
\sum_{i=1}^Nf_{vw}^{(i)}\leq f_{vw}.
\end{equation} 
The edge flow $f_{vw}$, in turn, is constraint by the capacity the constraint given by Supplementary Equation (\ref{LP2a}) \cite{ford1956maximal}. If all $f^{(i)}$ take integer values, there exist $f^{s\to t}=\sum_{i=1}^Nf^{(i)}$ \emph{edge-disjoint paths} from $s$ to $t$ in the multigraph $G_{\lfloor c' \rfloor}''$. We can now show:

\begin{lemma}\label{intFlow}
Let us assume we have a finite undirected graph $G'$ and $N$ the number of directed paths  from $s$ to $t$. Let $k,m\in\mathbb{N}$ and $c'(\{vw\})$ capacities that can depend on $m$ and $k$. Then we can obtain, in the unit-capacity multigraph $G''_{\lfloor mkNc'\rfloor}$, $F^{s\to t}$ edge-disjoint paths from $s$ to $t$, where
\begin{equation}
F^{s\to t}\geq mkN\left(f^{s\to t}_{\max}\left(G',\{c'(\{vw \}) \}_{\{vw\}\in E'}\right)-\frac{1}{k}\right).
\end{equation}
\end{lemma}
\begin{proof}
Let $m,k\in\mathbb{N}$ and $\{f_{vw}\}_{\{vw\}\in E'}$ be the set of edge flows maximizing the LP given by Supplementary Equations (\ref{LP2})-(\ref{LP2b}) for $G'$ for capacities $c'(\{vw\})$, which can depend on $m$ and $k$. In particular all $\{f_{vw}\}_{\{vw\}\in E'}$ can depend on $m$ and $k$. As $G'$ is finite, we can always find a finite number $N$ of directed paths $P_{s \to t}^{(i)}$ from $s$ to $t$.  For each path we can assign a path-flow $f^{(i)}\geq0$, such that $f^{s\to t}_{\max}\left(G',\{c'(\{e\})\}_{\{e\}\in E'}\right)=\sum_{i=1}^Nf^{(i)}$ \cite{ford1956maximal}. By Supplementary Equation (\ref{eq:sumfvwi}) it holds for every edge $\{vw\}\in E'$ that
\begin{align}
&f_{wv}\geq\sum_{i=1}^Nf^{(i)}_{wv}=\sum_{i=1}^Nf^{(i)}\delta(i,wv),\label{eq20}\\
&f_{vw}\geq\sum_{i=1}^Nf^{(i)}_{vw}=\sum_{i=1}^Nf^{(i)}\delta(i,vw),\label{eq21}
\end{align}
where $f^{(i)}_{vw}$ is defined by Supplementary Equation (\ref{eq:fvwi}) and
\begin{equation}
\delta(i,uv)=\begin{cases}
1 \text{ if }uv\in P^{(i)}\\ 0\text{ else.}
\end{cases}
\end{equation}
Then for each $f^{(i)}\geq0$ there exists $\bar{n}^{(i)}\in\mathbb{N}_0$ such that
\begin{equation}
f^{(i)}-\frac{1}{kN}\leq \frac{\bar{n}^{(i)}}{kN}\leq  f^{(i)}.
\end{equation}
Let us also define $F^{(i)}=mkN\frac{\bar{n}^{(i)}}{kN}=m\bar{n}^{(i)}$ and $F_{vw}=\sum_{i=1}^NF^{(i)}\delta(i,vw)$. As the $f_{vw}$ are feasible solutions of the LP given by Supplementary Equations (\ref{LP2})-(\ref{LP2b}), from Supplementary Equations (\ref{eq20}) and (\ref{eq21}) it holds for any edge $\{vw\}\in E'$ that
\begin{equation}
F_{wv}+F_{vw}\leq \lfloor mkN( f_{wv}+f_{vw})\rfloor \leq \lfloor mkNc'(\{vw\})\rfloor
\end{equation}
and for all $w\in V$ with $w\neq s,t$ that
\begin{equation}\label{FlowCons}
\sum_{v:\{vw\}\in E'}\left( F_{vw} - F_{wv}\right)=\sum_{i=1}^NF^{(i)}\sum_{v:\{vw\}\in E'}\left(\delta(i,vw)  -\delta(i,wv)\right)=0.
\end{equation}
To see the last equality, let us consider a path $P_{s\to t}^{(i)}$. If the path does not pass through vertex $w$, $\delta(i,vw)$ and $\delta(i,vw)$ vanish for all vertices $v$. Since $w\neq s,t$, if the path does pass through $w$, there will be two distinct vertices $v_0$ and $v_1$ such that $\delta(i,v_0w)=1$ and $\delta(i,wv_1)=1$. By definition, the path can only pass through $w$ once, and hence $\delta(i,vw)$ vanishes for all $v\neq v_{0,1}$. Hence $\sum_{v:\{vw\}\in E'}\left(\delta(i,vw)  -\delta(i,wv)\right)=0$ for every $i$ and $w\neq s,t$. 

Hence, $\{F_{vw}\}_{\{wv\}\in E'}$ is a feasible solution of the LP given by Supplementary Equations (\ref{LP2})-(\ref{LP2b}) with capacities $\lfloor mkNc'(\{vw\})\rfloor$, providing a total flow of
\begin{equation}\label{**}
F^{s\to t}=\sum_{i=1}^NF^{(i)}\geq m kN\left(f^{s\to t}_{\max}\left(G',\{c'(\{e\})\}_{\{e\}\in E'}\right)-\frac{1}{k}\right).
\end{equation}
As any integer flow of value $F^{(i)}$ corresponds to $F^{(i)}$ edge-disjoint paths in $G''_{\lfloor mkNc'\rfloor}$, we can conclude that there are $F^{s\to t}$ edge-disjoint paths from $s$ to $t$.
\end{proof}

In order to prove Proposition \ref{lowerbound}, we will now show that, given the solution of the LP given by Supplementary Equations (\ref{LP2})-(\ref{LP2b}), we can physically construct a network of Bell states corresponding to a graph where we can apply Lemma \ref{intFlow}. 

\begin{proof} (of Proposition \ref{lowerbound})
Let $k,m\in\mathbb{N}$. Let $G'$ be the graph as defined above and $N$ the number of directed paths $P_{s \to t}^{(i)}$ from $s$ to $t$. Without loss of generality we can assume $N\geq1$. Following \cite{AK16}, we can employ the following \emph{aggregated repeater protocol}: Across each channel $\mc{N}^e$, we perform Bell state generation protocols assisted two-way classical communication, using the channel $\lfloor m kNp_e\rfloor$ times. This provides us with states $\rho^e$ across $e$ such that $\left\|\rho^e-{\Phi_e^+}^{\otimes \lfloor m kNp_e\rfloor R_{\epsilon}^\leftrightarrow(\mc{N}^e)}\right\|_1\leq\epsilon$, where $\Phi_e^+ :=\ket{\Phi^2} \bra{\Phi^2}_e$ is a maximally entangled state across $e$ and $R_{\epsilon}^\leftrightarrow(\mc{N}^e)$ denotes the rate at which Bell states are generated across edge $e$ with some error $\epsilon>0$. Hence for the entire network we have
\begin{equation}
\left\|\bigotimes_{e\in E}\rho^e-\bigotimes_{e\in E}{\Phi_e^+}^{\otimes \lfloor m kNp_e\rfloor R_{\epsilon}^\leftrightarrow(\mc{N}^e)}\right\|_1\leq|E|\epsilon.
\end{equation}
It further holds for any $e\in E$ with $p_e>0$ and any $k\geq1$ that
\begin{equation}
 \lfloor m kNp_e\rfloor\geq  m kNp_e-1\geq \left(m-\frac{1}{p_e}\right)kNp_e\geq \left(m-\left\lceil\frac{1}{p_e}\right\rceil\right)kNp_e\geq \left(m-\tilde{m}\right)kNp_e,
\end{equation}
where we have defined $\tilde{m}:=\max_{e\in E, p_e>0}\left\lceil\frac{1}{p_e}\right\rceil$. Hence the state $\bigotimes_{e\in E}{\Phi_e^+}^{\otimes \lfloor m kNp_e\rfloor R_{\epsilon}^\leftrightarrow(\mc{N}^e)}$ can be transformed into $\bigotimes_{e\in E}{\Phi_e^+}^{\otimes \lfloor (m-\tilde{m}) kNp_e R_{\epsilon}^\leftrightarrow(\mc{N}^e)\rfloor}$ by removal of Bell pairs. Let us from now on assume that $m\geq\tilde{m}$. The resulting state can be interpreted as a network of Bell states, which can be described by the unit-capacity multigraph $G''_{\lfloor (m-\tilde{m}) kNc'_{R_{\epsilon}^\leftrightarrow}\rfloor}$, where we have defined $c'_{R_{\epsilon}^\leftrightarrow}$ by Supplementary Equation (\ref{cE}) with $\bar{c}(vw)=R_{\epsilon}^\leftrightarrow(\mc{N}^{vw})$. Let us note that $R_{\epsilon}^\leftrightarrow(\mc{N}^{vw})$ can depend on $m$ and $k$. By Lemma \ref{intFlow} there exist
\begin{equation}
F_{\epsilon}^{s\to t}\geq (m-\tilde{m})  kN\left(f^{s\to t}_{\max}\left(G',\{c'_{R_{\epsilon}^\leftrightarrow}(\{wv\},p_{vw},p_{wv})\}_{\{vw\}\in E'}\right)-\frac{1}{k}\right)
\end{equation}
edge-disjoint paths from $s$ to $t$ in $G''_{\lfloor(m-\tilde{m}) kNc'_{R_{\epsilon}^\leftrightarrow}\rfloor}$. Each edge-disjoint path corresponds to a chain of Bell states from $s$ to $t$. By means of entanglement swapping, we can connect these chains, providing us with a rate of entanglement generation between $s$ and $t$ of

\begin{equation}\label{rate}
\frac{F_{\epsilon}^{s\to t}}{m kN}\geq \left(1-\frac{\tilde{m}}{m}\right) f^{s\to t}_{\max}\left(G',\{c'_{R_{\epsilon}^\leftrightarrow}(\{wv\},p_{vw},p_{wv})\}_{\{vw\}\in E'}\right)-\frac{1}{k}.
\end{equation}
Taking the limit of $m\to \infty$ followed by the limit of $\epsilon \to 0$, it holds 
\begin{equation}
\lim_{\epsilon\to0}\lim_{m\to\infty}R_{\epsilon}^\leftrightarrow(\mc{N}^e)=Q^\leftrightarrow(\mc{N}^e).
\end{equation}
Note that the  two-way assisted quantum capacities $Q^\leftrightarrow(\mc{N}^e)$ no longer depend on $m$ and $k$. Using the fact that the optimal value of the objective of a parametric linear program of the form of Supplementary Equations (\ref{LP2})-(\ref{LP2b}) is a continuous function of the parameters \cite{meyer1979continuity}, we can see that 
\begin{equation}
 \mathcal{Q}_{\{p_e\}_{e\in E}}\left(G,\{\mc{N}^e\}_{e\in E}\right)\geq\lim_{\epsilon\to0}\lim_{m\to\infty}\frac{F_{\epsilon}^{s\to t}}{m kN}\geq   f^{s\to t}_{\max}(G',\{c'_{Q^\leftrightarrow}(\{wv\},p_{vw},p_{wv})\}_{\{vw\}\in E'})-\frac{1}{k}.
\end{equation}
Taking the limit $k\to\infty$ finishes the proof.
\end{proof}

In order to lower bound capacities of the form of Supplementary Equation (\ref{cap2}), we can again include an optimization over all usage frequencies $p_e$ into the optimization given by Supplementary Equations (\ref{LP2})-(\ref{LP2b}), defining
\begin{equation}\label{fbarQ}
\bar{f}^{s\to t}_{\max}(G',\{{Q^{\leftrightarrow}}(\mc{N}^e)\}_{e\in E})=\max_{p_e\geq0,\sum_ep_e=1} f^{s\to t}_{\max}(G',\{c'_{Q^{\leftrightarrow}}(\{wv\},p_{vw},p_{wv})\}_{\{wv\}\in E'}).
\end{equation}
Then we have the following:
\begin{corollary}\label{Corr2}
For a network described by a finite directed graph $G$ and an undirected graph $G'$ as defined above, it holds
\begin{equation}
\mathcal{Q}\left(G,\{\mc{N}^e\}_{e\in E}\right)\geq \bar{f}^{s\to t}_{\max}(G',\{Q^{\leftrightarrow}(\mc{N}^e)\}_{e\in E}).
\end{equation}
Hence, for distillable channels it holds
\begin{equation}
\mathcal{Q}\left(G,\{\mc{N}^e\}_{e\in E}\right)=\bar{f}^{s\to t}_{\max}(G',\{Q^{\leftrightarrow}(\mc{N}^e)\}_{e\in E}).
\end{equation}
\end{corollary}
%%%%%%%%%%%%%%

\subsection{Multiple pairs of users}\label{sec:multiuni}
We now move on to the scenario of multi-pairs of users $s_1 \cdots  s_r$ and $t_1 \cdots  t_r$ who wish to establish maximally entangled states or private states concurrently. In this scenario the target state is of the form
$\theta^{d_1\cdots d_r}_{X_{s_1}X_{t_1}  \cdots  X_{s_r}X_{t_r}}=\bigotimes_{i=1}^r\theta^{d_i}_{X_{s_i}X_{t_i}}$. The states $\theta^{d_i}_{X_{s_i}X_{t_i}}$ can be maximally entangled states or private states. There are several ways to measure the performance of a protocol performing concurrent entanglement distribution. We consider the following three figures of merit: (1) the total multi-pair capacity, i.e. the sum of achievable rates over all user pairs. The drawback of this approach is that it does not distinguish between fair protocols where each user pair gets a similar amount of the resource and unfair ones where some user pairs get more then others. This drawback can be overcome by using our second figure of merit: (2) worst-case multi-pair capacity, i.e. the least achievable rate that is guaranteed for any user pair. Finally we consider (3) the case where we assign weights to each user pair independently. 

For case (1), we define the \emph{total multi-pair network capacity with respect to target state $\theta^{d_1 \cdots  d_r}_{X_{s_1}X_{t_1}  \cdots  X_{s_r}X_{t_r}}$ with fixed average usage frequencies $\{p_e\}_{e\in E}$} as 
\begin{equation}\label{cap3a}
\mathcal{C}^{\theta,\text{total}}_{\{p_e\}_{e\in E}}\left(G,\{\mc{N}^e\}_{e\in E}\right)=\lim_{\epsilon\to0}\lim_{n\to\infty}\sup_{\Lambda(n,\epsilon, \{p_e\}_{e\in E})}\left\{\frac{\sum_{i=1}^r\langle\log d_i^{(k)}\rangle_{k}}{n}:\left\|\rho_{X_{s_1}X_{t_1}  \cdots  X_{s_r}X_{t_r}}^{(n,k)}-\theta^{d^{(k)}_1 \cdots  d^{(k)}_r}_{X_{s_1}X_{t_1}  \cdots  X_{s_r}X_{t_r}}\right\|_1\leq\epsilon\right\},
\end{equation}
where, again, the supremum is over all adaptive $(n,\epsilon, \{p_e\}_{e\in E})$ protocols $\Lambda$ and $k=(k_1,  \ldots  ,k_{m+1})$ is a vector of outcomes of the $m+1$ LOCC rounds in $\Lambda$, the averaging is over all those outcomes and $\rho_{X_{s_1} X_{t_1} \ldots X_{s_r} X_{t_r} } ^{(n,k)}  $ is the final state of $\Lambda$ for given outcomes $k$. For case (2), we define the \emph{worst-case multi-pair network capacity with respect to target state $\theta^{d_1 \cdots  d_r}_{X_{s_1}X_{t_1}  \cdots  X_{s_r}X_{t_r}}$ with fixed average usage frequencies $\{p_e\}_{e\in E}$} as 
\begin{equation}\label{cap3}
\mathcal{C}^{\theta,\text{worst}}_{\{p_e\}_{e\in E}}\left(G,\{\mc{N}^e\}_{e\in E}\right)=\lim_{\epsilon\to0}\lim_{n\to\infty}\sup_{\Lambda(n,\epsilon, \{p_e\}_{e\in E})}\min_{i\in\{1 \cdots  r\}}\left\{\frac{\langle\log d_i^{(k)}\rangle_{k}}{n}:\left\|\rho_{X_{s_1}X_{t_1}  \cdots  X_{s_r}X_{t_r}}^{(n,k)}-\theta^{d^{(k)}_1 \cdots  d^{(k)}_r}_{X_{s_1}X_{t_1}  \cdots  X_{s_r}X_{t_r}}\right\|_1\leq\epsilon\right\}.
\end{equation}
In scenario (3), where some user pairs are more important than others, we can also assign  nonnegative weights $q_1,\cdots,q_r$, with $\sum_iq_i=1$, to user pairs $(s_1,t_1), \cdots  ,(s_r,t_r)$ and define a \emph{weighted multi-pair network capacity} as
\begin{equation}\label{cap3c}
\mathcal{C}^{\theta,q_1,\cdots,q_r}_{\{p_e\}_{e\in E}}\left(G,\{\mc{N}^e\}_{e\in E}\right)=\lim_{\epsilon\to0}\lim_{n\to\infty}\sup_{\Lambda(n,\epsilon, \{p_e\}_{e\in E})}\left\{\frac{\sum_{i=1}^rq_i\langle\log d_i^{(k)}\rangle_{k}}{n}:\left\|\rho_{X_{s_1}X_{t_1}  \cdots  X_{s_r}X_{t_r}}^{(n,k)}-\theta^{d^{(k)}_1 \cdots  d^{(k)}_r}_{X_{s_1}X_{t_1}  \cdots  X_{s_r}X_{t_r}}\right\|_1\leq\epsilon\right\}.
\end{equation}
In the case where $q_i=1/r$ for all $i=1,...,r$, Supplementary Equation (\ref{cap3c}) reduces to Supplementary Equation (\ref{cap3a}), up to a normalization factor. Again, we can optimize over usage frequencies as
\begin{align}\label{cap4}
&\mathcal{C}^{\theta,\text{total}}\left(G,\{\mc{N}^e\}_{e\in E}\right)=\max_{p_e\geq0,\sum_ep_e=1}\mathcal{C}^{\theta,\text{total}}_{\{p_e\}_{e\in E}}\left(G,\{\mc{N}^e\}_{e\in E}\right),\\
&\mathcal{C}^{\theta,\text{worst}}\left(G,\{\mc{N}^e\}_{e\in E}\right)=\max_{p_e\geq0,\sum_ep_e=1}\mathcal{C}^{\theta,\text{worst}}_{\{p_e\}_{e\in E}}\left(G,\{\mc{N}^e\}_{e\in E}\right),\\
&\mathcal{C}^{\theta,q_1,\cdots,q_r}\left(G,\{\mc{N}^e\}_{e\in E}\right)=\max_{p_e\geq0,\sum_ep_e=1}\mathcal{C}^{\theta,q_1,\cdots,q_r}_{\{p_e\}_{e\in E}}\left(G,\{\mc{N}^e\}_{e\in E}\right).
\end{align}
If the target state is a product of maximally entangled (or private) states, we use the notation $\mathcal{C}^{\bigotimes_i\Phi_i}=\mc{Q}$ (or $\mathcal{C}^{\bigotimes_i\gamma_i}=\mc{P}$) and speak of quantum (or private) multi-pair network capacities.

\subsubsection{Upper bounding the capacity}
Several upper bounds on various multi-user pair capacities have been obtained \cite{AML16,pirandola2016capacities,bauml2017fundamental}. In particular it follows as a special case of Theorem 2 in \cite{bauml2017fundamental} (and also from equation (4) in \cite{AML16}) that for an $(m,\epsilon, \{p_e\}_{e\in E})$ key generation protocol, it holds for every $V'\subset V$ that
\begin{equation}\label{21}
\frac{1}{m}\sum_{i=1}^r\delta_{i|\partial(V')}\langle\log d_i^{(k)}\rangle_{k}\leq\frac{1}{1-b\epsilon}\left(\sum_{e\in E:\{e\} \in\partial(V')} p_e E_{\sq}(\mc{N}^e)+g(\epsilon)\right),
\end{equation}
where $b>0$, $g(\epsilon)\to0$ as $\epsilon\to0$ and
\begin{equation}\delta_{i|\partial(V')}=\begin{cases}
1\text{ if }(s_i,t_i)\in(V'\times V\setminus V')\text{ or }(s_i,t_i)\in(V\setminus V'\times  V')\\
0\text{ else,}
\end{cases}
\end{equation}
where $\times$ denotes a cartesian product. Using the same reasoning as in \cite{AML16}, it is also possible to extend the results of \cite{rigovacca2018versatile} to the multi-pair case, which includes the bound given by Supplementary Equation (\ref{21}). Given a cut $\partial(V')$, we consider all user pairs $(s_i,t_i)$ such that $(s_i,t_i)\in(V'\times V\setminus V')$ or $(s_i,t_i)\in(V\setminus V'\times  V')$. Assume that in the protocol, each such pair $(s_i,t_i)$ obtains $\langle\log d_i^{(k)}\rangle_{k}$ target bits. As the target states are invariant under permutation of the parties we can, without loss of generality, relabel the users in the following way: If $(s_i,t_i)\in(V'\times V\setminus V')$, we define $(\tilde{s}_i,\tilde{t}_i):=(s_i,t_i)$, while if $(s_i,t_i)\in(V\setminus V'\times  V')$, we define $(\tilde{s}_i,\tilde{t}_i):=(t_i,s_i)$, such that it always holds $(\tilde{s}_i,\tilde{t}_i)\in(V'\times V\setminus V')$. Let us now assume a hypothetical scenario where all $\tilde{s}_i$ are in the same place and thus have the full control of their quantum systems, forming a `superuser' $\tilde{s}$ and similarly their partners $\tilde{t}_i$ in $V\setminus V'$ can form a `superuser' $\tilde{t}$. By combining their outcomes the pair of superusers $(\tilde{s},\tilde{t})$ can achieve at least $\sum_{i=1}^r\delta_{i|\partial(V')}\langle\log d_i^{(k)}\rangle_{k}$ target bits. As, by assumption, $\partial(V')$ separates superusers $\tilde{s}$ and $\tilde{t}$, the total number of obtainable target bits between the pair $(\tilde{s},\tilde{t})$ is upper bounded by $\sum_{e\in E:\{e\} \in\partial(V')} p_e {\cal E}(\mc{N}^e)$ for entanglement measures satisfying properties 1 and 2. Hence, we have
\begin{equation}\label{19}
\frac{1}{m}\sum_{i=1}^r\delta_{i|\partial(V')}\langle\log d_i^{(k)}\rangle_{k}\leq\frac{1}{1-f(\epsilon)}\left(\sum_{e\in E:\{e\} \in\partial(V')} p_e {\cal E}(\mc{N}^e)+g(\epsilon)\right),
\end{equation}
where $f(\epsilon)\to 0$ and $g(\epsilon)\to 0$ as $\epsilon\to0$. Combining Supplementary Equation (\ref{19}) with the results of \cite{aumann1998log}, we can, up to a factor $g_{\rm w}(r)=\mathcal{O}(\log r^*)$, upper bound the worst-case capacity using a worst-case multi-commodity flow optimization, which can be expressed as a linear program of the following form (see the Methods section for details):
\begin{align}
f^{\text{\text{worst}}}_{\max}\left(G',\{c'(\{vw\})\}_{\{wv\}\in E'}\right)=\textrm{max} &\ f\label{LP4}\\
\forall i:\ &\ f-\sum_{v:\{s_iv\}\in E'} \left(f^{(i)}_{s_iv}-f^{(i)}_{vs_i}\right)\leq0\label{LP4a}\\
\forall \{vw\}\in E':\ &\  \sum_{i=1}^r\left(f^{(i)}_{vw}+f^{(i)}_{wv}\right) \leq c'(\{wv\})\label{LP4b}\\
\forall i,\ \forall w\in V,w\neq s_i,t_i:\ &\ \sum_{v:\{vw\}\in E'} \left(f^{(i)}_{vw} -f^{(i)}_{wv}\right)=0,\label{LP4c}
\end{align}
where the maximization is over edge flows $f^{(i)}_{vw}\geq 0$ and $f^{(i)}_{wv}\geq 0$ for all edges $\{vw\}\in E'$. In the LP given by Supplementary Equations (\ref{LP4})-(\ref{LP4c}) we have introduced an additional variable $f$, which we maximize under the constraint that it lower bounds the source-sink flows for all pairs $i=1,...,r$, as expressed in  Supplementary Equation (\ref{LP4a}). This ensures that $f$ is the `worst-case' flow that can be achieved for any user pair. %Constraints Supplementary Equation (\ref{LP4b}) and (\ref{LP4c}) are the same as in Proposition \ref{Theo_Total}. 

This allows us to prove the following result:
\begin{proposition}\label{multiuni:upper}
In a network described by a graph $G$ with associated undirected graph $G'$ it holds in a scenario of $r$ user pairs $(s_1,t_1),...,(s_r,t_r)$ for any entanglement measure with properties 1 and 2,
\begin{align}
\mc{P}^{\text{\emph{worst}}}_{\{p_e\}_{e\in E}}\left(G,\{\mc{N}^e\}_{e\in E}\right)\leq g_{\rm w}(r)f^{\text{\emph{worst}}}_{\max}\left(G',\{c'_{\cal E}(\{vw\},p_{vw},p_{wv})\}_{\{wv\}\in E'}\right),
\end{align}
where  $g_{\rm w}(r)=\mathcal{O}(\log r^*)$ and $f^{\text{\emph{worst}}}_{\max}$ is given by the LP defined by Supplementary Equations (\ref{LP4})-(\ref{LP4c}). 
\end{proposition}

\begin{proof}
It holds for any cut with $\delta_{i|\partial(V')}>0$ for at least one user pair,
\begin{equation}
\min_{j\in\{1 \cdots  r\}}\langle\log d_j^{(k)}\rangle_{k}\leq\frac{\sum_{i=1}^r\delta_{i|\partial(V')}\langle\log d_i^{(k)}\rangle_{k}}{\sum_{i=1}^r\delta_{i|\partial(V')}}.
\end{equation}
Using Supplementary Equation (\ref{19}) and taking the limit $m\to\infty$ and $\epsilon\to0$, we obtain
\begin{align}\label{20}
&\mc{P}^{\text{worst}}_{\{p_e\}_{e\in E}}\left(G,\{\mc{N}^e\}_{e\in E}\right)\leq\min_{V'}\frac{\sum_{e\in E:\{e\} \in\partial(V')} p_e {\cal E}(\mc{N}^e)}{\sum_{i=1}^r\delta_{i|\partial(V')}}=R_{\min} \left(G',\{c'_{\cal E}(\{vw\},p_{vw},p_{wv})\}_{\{wv\}\in E'}\right),
\end{align}
where we have used that in the case of unit demands $\sum_{i=1}^r\delta_{i|\partial(V')}=d(\partial(V'))$. Application of the results of \cite{aumann1998log} (equation (26) in the Methods section) finishes the proof.
\end{proof}

Let us now consider the total throughput scenario. It follows from Theorem 2 in \cite{bauml2017fundamental} that for every multicut $\{S\}\leftrightarrow \{T\}$ (See the Methods section for the definition) it holds \footnote{Let us note that in Theorem 2 of \cite{bauml2017fundamental} we have used the multipartite squashed entanglement with respect to the partition defined by the multicut. However, as every channel $\mc{N}^e$ is a point-to-point channel linking only two parts of the partition, it reduces to the usual bipartite squashed entanglement. This becomes obvious from the squashed entanglement with respect to a partition defined in Eq. (13)-(15) in \cite{bauml2017fundamental}.}

\begin{equation}\label{multicutversion}
\frac{1}{m}\sum_{i=1}^r\langle\log d_i^{(k)}\rangle_{k}\leq\frac{1}{1-b\epsilon}\left(\sum_{e\in E:\{e\} \in \{S\}\leftrightarrow \{T\}} p_e E_{\sq}(\mc{N}^e)+g(\epsilon)\right).
\end{equation}

As was shown in \cite{garg1996approximate}, we can, up to a factor $g_{\rm t}(r)$ of order $\mathcal{O}(\log r)$, upper bound the minimum multicut by a multi-commodity flow optimization, which can be formulated as a linear program of the following form (See the Methods section for details):
\begin{align}
f^{\text{total}}_{\max}\left(G',\{c'(\{vw\})\}_{\{vw\}\in E'}\right)=\textrm{max} &\ \sum_{i=1}^r\sum_{v:\{s_iv\}\in E'} \left(f^{(i)}_{s_iv}-f^{(i)}_{vs_i}\right)\label{LP3}\\
\forall \{vw\}\in E':\ &\  \sum_{i=1}^r\left(f^{(i)}_{vw}+f^{(i)}_{wv}\right) \leq c'(\{vw\})\label{LP3a}\\
\forall i,\ \forall w\in V,w\neq s_i,t_i:\ &\ \sum_{v:\{vw\}\in E'} \left(f^{(i)}_{vw} -f^{(i)}_{wv}\right)=0\label{LP3b}.
\end{align}
Let us now describe the LP given by Supplementary Equations (\ref{LP3})-(\ref{LP3b}) in detail: The objective given by Supplementary Equation (\ref{LP3}) is the sum of source-sink flows over all user pairs/commodities $i=1,...,r$. For each user pair, the source-sink flow is expressed analogously to Supplementary Equation (\ref{LP2}). The edge capacity constraints given by Supplementary Equation (\ref{LP3a}) involve a summation over all commodities that pass through given edge in both directions. The flow conservation constraints given by Supplementary Equation (\ref{LP3b}) are of the same form as Supplementary Equation (\ref{LP2b}), but have to be observed for all commodities $i=1,...,r$. This allows us to show the following:
\begin{proposition}\label{multiuni:upper2}
In a network described by a graph $G$ with associated undirected graph $G'$ and a scenario of $r$ user pairs $(s_1,t_1),...,(s_r,t_r)$ 
\begin{align}
\mc{P}^{\text{\emph{total}}}_{\{p_e\}_{e\in E}}\left(G,\{\mc{N}^e\}_{e\in E}\right)\leq g_{\rm t}(r)f^{\text{\emph{total}}}_{\max}\left(G',\{c'_{E_{\sq}}(\{vw\},p_{vw},p_{wv})\}_{\{vw\}\in E'}\right),
\end{align}
where $f^{\text{\emph{total}}}_{\max}$ is given by the LP defined by Supplementary Equations (\ref{LP3})-(\ref{LP3b}).
\end{proposition}

\begin{proof}
Using Supplementary Equation (\ref{multicutversion}) and taking the limit $m\to\infty$ and $\epsilon\to0$, we obtain
\begin{align}
&\mc{P}^{\text{total}}_{\{p_e\}_{e\in E}}\left(G,\{\mc{N}^e\}_{e\in E}\right)\leq\min_{\{S\}\leftrightarrow \{T\}}\sum_{e\in E:\{e\} \in \{S\}\leftrightarrow \{T\}} p_e E_{\sq}(\mc{N}^e).
\end{align}
Application of the results of \cite{garg1996approximate} (equation (23) in the Methods section) finishes the proof.
\end{proof}

%%%%%%%%%%%%%%%%%%
As in the case of a single pair of users we can include an optimization over all usage frequencies $p_e$ into the optimizations  given by Supplementary Equations (\ref{LP4})-(\ref{LP4c}) and (\ref{LP3})-(\ref{LP3b}), resulting in LPs 
\begin{align}
&\bar{f}^{\text{total}}_{\max}\left(G',\{\bar{c}(e)\}_{e\in E}\right)=\max_{p_e\geq0,\sum_ep_e=1} f^{\text{total}}_{\max}(G',\{c'_{E_{\sq}}(\{wv\},p_{vw},p_{wv})\}_{\{wv\}\in E'}),\label{fbartotal}\\
&\bar{f}^{\text{worst}}_{\max}\left(G',\{\bar{c}(e)\}_{e\in E}\right)=\max_{p_e\geq0,\sum_ep_e=1} f^{\text{worst}}_{\max}(G',\{c'_{\cal E}(\{wv\},p_{vw},p_{wv})\}_{\{wv\}\in E'}).\label{fbarworst}
\end{align}
This provides us with the following:
\begin{corollary}\label{Corr3}
In a network described by a graph $G$ with associated undirected graph $G'$ it holds in a scenario of $r$ user pairs $(s_1,t_1),...,(s_r,t_r)$,
\begin{align}&\mathcal{P}^{\text{\emph{total}}}\left(G,\{\mc{N}^e\}_{e\in E}\right)\leq g_{\rm t}(r)\bar{f}^{\text{\emph{total}}}_{\max}(G',\{E_{\sq}(\mc{N}^e)\}_{e\in E}),\\
&\mathcal{P}^{\text{\emph{worst}}}\left(G,\{\mc{N}^e\}_{e\in E}\right)\leq g_{\rm w}(r)\bar{f}^{\text{\emph{worst}}}_{\max}(G',\{{\cal E}(\mc{N}^e)\}_{e\in E}),
\end{align}
where ${\cal E}$ is an entanglement measure with properties 1 and 2 and $\bar{f}^{\text{\emph{worst}}}_{\max}$ and $\bar{f}^{\text{\emph{total}}}_{\max}$ are given by  Supplementary Equations (\ref{LP4}) and (\ref{LP3a}) with added optimization over usage frequencies, respectively.
\end{corollary}

Finally, let us consider the scenario of weighted user pairs. Given any subset $U\subset\left\{(s_1,t_1), \cdots  ,(s_r,t_r)\right\}$, we can define a LP 
\begin{align}\label{LP3c}
f^{U}_{\max}\left(G',\{c'(\{vw \}) \}_{\{vw\}\in E'}\right)=\textrm{max} &\sum_{i\in I_U}\sum_{v:\{s_iv\}\in E'} \left(f^{(i)}_{s_iv}-f^{(i)}_{vs_i}\right)\\
\forall \{vw\}\in E':\ &\  \sum_{i\in I_U}\left(f^{(i)}_{vw}+f^{(i)}_{wv}\right)\le  c'(\{vw\})\nonumber\\
\forall \{vw\}\in E', \ \forall i:\ &\  f^{(i)}_{vw}, f^{(i)}_{wv}\geq 0\nonumber\\
\forall i\in I_U,\ \forall w\in V,w\neq s_i,t_i:\ &\ \sum_{v:\{vw\}\in E'} \left(f^{(i)}_{vw} -f^{(i)}_{wv}\right)=0,\nonumber
\end{align}
where $I_U:=\{i:(s_i,t_i)\in U\}$. For subset $U$, we can also define  a multicut $\{S_U\}\leftrightarrow \{T_U\}$ as a set of edges in $E'$ whose removal disconnects all source sink pairs in $U$. By \cite{garg1996approximate}, it then holds 
\begin{equation}\label{maxFlowMinMultiCut2}
%f^{\text{total}}_{\max}\left(G',\{c'(\{vw \}) \}_{\{vw\}\in E'}\right)\leq 
\min_{\{S_U\}\leftrightarrow \{T_U\}}c'(\{S_U\}\leftrightarrow \{T_U\})\leq g_{\rm t}(|U|)f^{U}_{\max}\left(G',\{c'(\{vw \}) \}_{\{vw\}\in E'}\right),
\end{equation}
if $|U|\geq2$. If $|U|=1$, this reduces to the max-flow min-cut theorem. For every user pair $(s_i,t_i)$, where $i\in\{1,...,r\}$, it holds by Corollary \ref{Corr1} that for all $m\in\mathbb{N}$ and $\epsilon>0$
\begin{equation}\label{polytope1}
R_i^{m,\epsilon}\leq_{\epsilon} \bar{f}^{s_i\to t_i}_{\max}\left(G',\{{\cal E}(\mc{N}^e) \}_{e\in E}\right),
\end{equation}
where $R_i^{m,\epsilon}:=\langle\log d_i^{(k)}\rangle_{k}/m$ and ${\cal E}$ is an entanglement measure satisfying properties 1 and 2 and the r.h.s. given by the LP defined by Supplementary Equations (\ref{LP2})-(\ref{LP2b}). Further, by Corollary \ref{Corr3}, it holds for any subset $U\subset\left\{(s_1,t_1), \cdots ,(s_r,t_r)\right\}$ with $|U|\geq2$ of user pairs that
\begin{equation}\label{polytope2}
\sum_{i\in I_U}R_i^{m,\epsilon}\leq_\epsilon g_{\rm t}(|U|)\bar{f}^U_{\max}\left(G',\{{E_{\sq}}(\mc{N}^e) \}_{e\in E}\right),
\end{equation}
where the r.h.s. is given by the LP defined by Supplementary Equations (\ref{LP3})-(\ref{LP3b}), considering only user pairs in $U$ and with added optimization over usage frequencies. As the inequalities given by Supplementary Equations (\ref{polytope1}) and (\ref{polytope2}), for all possible subsets $U\subset\left\{(s_1,t_1), \cdots ,(s_r,t_r)\right\}$, define facets of a polytope $\Pi^{m,\epsilon}$ in $\mathbb{R}^r_+$. Let us also define the `asymptotic' polytope $\Pi:=\lim_{\epsilon\to0}\lim_{m\to\infty}\Pi^{m,\epsilon}$. Optimization over the polytope provides us with the following result:
\begin{proposition}\label{multiuni:upper3}
In a network described by a graph $G$ with associated undirected graph $G'$ and a scenario of $r$ user pairs $(s_1,t_1),...,(s_r,t_r)$ with weights $q_1,...,q_r$ it holds
\begin{equation}
\mc{P}^{q_1,...,q_r}\left(G,\{\mc{N}^e\}_{e\in E}\right)\leq\max_{(\hat{R}_1,...,\hat{R}_r)\in\Pi}\sum_{i=1}^rq_i\hat{R}_i,
\end{equation}
which is a linear program.
\end{proposition}
\begin{proof}
This follows from the definition given by Supplementary Equation (\ref{cap3c}) of $\mc{P}^{q_1,...,q_r}_{\{p_e\}_{e\in E}}\left(G,\{\mc{N}^e\}_{e\in E}\right)$ as a weighted sum of asymptotic rates and the definition of the polytope $\Pi$.
\end{proof}

Let us note that instead of the weighted sum of rates $\sum_{i=1}^rq_i{R}_i$, we could also maximize a general concave target function $f({R}_1,...,{R}_r)$ over a polytope, which would still a convex optimization problem. Let us also note that whereas the number of all possible subsets of $U\subset\left\{(s_1,t_1), \cdots ,(s_r,t_r)\right\}$ scales exponentially in $r$, we can also obtain upper bounds by using only a small number of subsets $U\subset\left\{(s_1,t_1), \cdots ,(s_r,t_r)\right\}$ or only the single-pair bounds given by Supplementary Equation (\ref{polytope1}) to define the polytope.

\subsubsection{Lower bounding the capacity}
It is straightforward to extend our lower bound, Proposition \ref{lowerbound}, to multiple user scenarios:
\begin{proposition}\label{ConcLowerBound}
In a network described by a graph $G$ with associated undirected graph $G'$ it holds in a scenario of $r$ user pairs $(s_1,t_1),...,(s_r,t_r)$,
\begin{align}
&\mc{Q}^{\text{\emph{worst}}}_{\{p_e\}_{e\in E}}\left(G,\{\mc{N}^e\}_{e\in E}\right)\geq f^{\text{\emph{worst}}}_{\max}\left(G',\{c'_{Q^{\leftrightarrow}}(\{vw\},p_{vw},p_{wv})\}_{\{vw\}\in E'}\right),\\
&\mc{Q}^{\text{\emph{total}}}_{\{p_e\}_{e\in E}}\left(G,\{\mc{N}^e\}_{e\in E}\right)\geq f^{\text{\emph{total}}}_{\max}\left(G',\{c'_{Q^{\leftrightarrow}}(\{vw\},p_{vw},p_{wv})\}_{\{vw\}\in E'}\right),
\end{align}
where $f^{\text{\emph{worst}}}_{\max}$ and $f^{\text{\emph{total}}}_{\max}$ are given by the LP defined by Supplementary Equations (\ref{LP4})-(\ref{LP4c}) and the LP given by Supplementary Equations (\ref{LP3})-(\ref{LP3b}), respectively.
\end{proposition}

Before proving Proposition \ref{ConcLowerBound}, we need the following Lemmas:
\begin{lemma}\label{ConcIntFlow}
Let us assume we have a finite undirected graph $G'$ with capacities $c'(\{vw\})$ (that can in general depend on $m$ and $k$) and $r$ source-sink pairs $(s_1,t_1),...,(s_r,t_r)$. Let $N:=\max_{i\in\{1,...,r\}}N_i$, where $N_i$ are the numbers of directed paths from $s_i$ to $t_i$ that exist in $G'$. Let further $k,m\in\mathbb{N}$. Then we can, for any $i\in\{1,...,r\}$ concurrently, obtain, in $G''_{\lfloor mkNc'\rfloor}$, $F^{s_i\to t_i}$ edge-disjoint paths from $s_i$ to $t_i$, where
\begin{equation}
F^{s_i\to t_i}\geq mkN\left(f^{\text{\emph{worst}}}_{\max}\left(G',\{c'(\{vw \}) \}_{\{vw\}\in E'}\right)-\frac{1}{k}\right).
\end{equation}
\end{lemma}

\begin{proof}
Let $m,k\in\mathbb{N}$ and let $\{f^{(i)}_{vw}\}_{i\in\{1...r\},\{vw\}\in E'}$ be the set of edge flows maximizing the LP given by Supplementary Equations (\ref{LP4})-(\ref{LP4c}) for $G'$ with capacities $c'(\{vw\})$, which in general can depend on $m$ and $k$. As $G'$ is finite, we can, for any $i\in\{1,...,r\}$, always find a finite number $N_i$ of directed paths $P_{s_i \to t_i}^{(ij)}$, where $j\in\{1,...,N_i\}$, from $s_i$ to $t_i$. Hence $N:=\max_{i\in\{1,...,r\}}N_i$ is finite.  For each path $P_{s_i \to t_i}^{(ij)}$ we can assign a path-flow $f^{(ij)}\geq0$ such that for every $i\in\{1...r\}$ it holds
\begin{equation}
\sum_{j=1}^{N_i}f^{(ij)}=\sum_{v:\{s_iv\}\in E'}\left( f^{(i)}_{s_iv}-f^{(i)}_{vs_i}\right)\geq f^{\text{worst}}_{\max}\left(G',\{c'(\{e\})\}_{\{e\}\in E'}\right)
\end{equation}
(see also \cite{ford1956maximal}). Let us define in analogy to Supplementary Equation (\ref{eq:fvwi})
\begin{equation}\label{fvwij}
f^{(ij)}_{vw}=\begin{cases}f^{(ij)}\text{ if }vw\in P_{s_i \to t_i}^{(ij)}\\0\text{ else.}\end{cases}
\end{equation}
As in Supplementary Equation (\ref{eq:sumfvwi}), it holds for every edge $\{vw\} \in E'$ that
\begin{align}
&f^{(i)}_{wv}\geq\sum_{j=1}^{N_i}f^{(ij)}_{wv}=\sum_{j=1}^{N_i}f^{(ij)}\delta(ij,wv),\label{A2}\\
&f^{(i)}_{vw}\geq\sum_{j=1}^{N_i}f^{(ij)}_{vw}=\sum_{j=1}^{N_i}f^{(ij)}\delta(ij,vw),\label{A3}
\end{align}
where 
\begin{equation}
\delta(ij,uv)=\begin{cases}
1 \text{ if }uv\in P^{(ij)}_{s_i\to t_i}\\ 0\text{ else.}
\end{cases}
\end{equation}
Then for each $f^{(ij)}$ there exists $\bar{n}^{(ij)}\in\mathbb{N}_0$ such that
\begin{equation}
f^{(ij)}-\frac{1}{kN}\leq \frac{\bar{n}^{(ij)}}{kN}\leq  f^{(ij)}.
\end{equation}
Let us also define $F^{(ij)}=m\bar{n}^{(ij)}$ and $F^{(i)}_{vw}=\sum_{j=1}^{N_i}F^{(ij)}\delta(ij,vw)$. As the $f^{(i)}_{vw}$ are feasible solutions of the LP given by Supplementary Equations (\ref{LP4})-(\ref{LP4c}), it holds for any edge $\{vw\}\in E'$ that
\begin{equation}\label{A6}
\sum_{i=1}^r\left(F^{(i)}_{wv}+F^{(i)}_{vw}\right)\leq \lfloor mkN\sum_{i=1}^r( f^{(i)}_{wv}+f^{(i)}_{vw})\rfloor \leq \lfloor mkNc'(\{vw\})\rfloor.
\end{equation}
Further, for all $i\in\{1...r\}$, it holds for all $w\in V$ with $w\neq s_i,t_i$ that
\begin{equation}\label{A7}
\sum_{v:\{vw\}\in E'}\left( F^{(i)}_{vw} - F^{(i)}_{wv}\right)=\sum_{j=1}^{N_i}F^{(ij)}\sum_{v:\{vw\}\in E'}\left(\delta(ij,vw)  -\delta(ij,wv)\right)=0,
\end{equation}
where we use the same argument as explained after Supplementary Equation (\ref{FlowCons}). It also holds
\begin{align}
\sum_{v:\{s_iv\}\in E'}F^{(i)}_{s_iv}&=\sum_{j=1}^{N_i}F^{(ij)}\sum_{v:\{s_iv\}\in E'}\delta(ij,s_iv)\\
&=\sum_{j=1}^{N_i}F^{(ij)}\\
&\geq mkN\sum_{j=1}^{N_i}\left(f^{(ij)}-\frac{1}{kN}\right)\\
&\geq mkN\left(f^{\text{worst}}_{\max}\left(G',\{c'(\{e\})\}_{\{e\}\in E'}\right)-\frac{1}{k}\right).
\end{align}
Hence $\{F^{(i)}_{vw}\}_{i\in\{1...r\},\{vw\}\in E'}$ is a feasible solution of the LP given by Supplementary Equations (\ref{LP4})-(\ref{LP4c}) with capacities $\lfloor mkNc'(\{vw\})\rfloor$, providing, for any $i\in\{1,...,r\}$ concurrently, a flow of
\begin{equation}
F^{s_i\to t_i}=\sum_{j=1}^{N_i}F^{(ij)}\geq m kN\left(f^{\text{worst}}_{\max}\left(G',\{c'(\{e\})\}_{\{e\}\in E'}\right)-\frac{1}{k}\right).
\end{equation}
As any integer flow of value $F^{(ij)}$ corresponds to $F^{(ij)}$ edge-disjoint paths from $s_i$ to $t_i$ in $G''_{\lfloor mkNc'\rfloor}$, we can conclude that, for all $i\in\{1,...,r\}$ concurrently, there are $F^{s_i\to t_i}$ edge-disjoint paths from $s_i$ to $t_i$.
\end{proof}

\begin{lemma}\label{TpIntFlow}
Let us assume we have a finite undirected graph $G'$ with capacities $c'(\{vw\})$ (that can in general depend on $m$ and $k$) and $r$ source sink pairs $(s_1,t_1),...,(s_r,t_r)$. Let $N:=\sum_{i=1}^rN_i$, where $N_i$ are the numbers of directed paths from $s_i$ to $t_i$ that exist in $G'$. Let further $k,m\in\mathbb{N}$. Then we can, for any $i\in\{1,..,r\}$ concurrently, obtain, in $G''_{\lfloor mkNc'\rfloor}$, $F^{s_i\to t_i}$ edge-disjoint paths from $s_i$ to $t_i$, where
\begin{equation}
\sum_{i=1}^rF^{s_i\to t_i}\geq mkN\left(f^{\text{\emph{total}}}_{\max}\left(G',\{c'(\{vw \}) \}_{\{vw\}\in E'}\right)-\frac{1}{k}\right).
\end{equation}
\end{lemma}
\begin{proof}
Let $m,k\in\mathbb{N}$ and let $\{f^{(i)}_{vw}\}_{i\in\{1...r\},\{vw\}\in E'}$ be the set of edge flows maximizing the LP given by Supplementary Equations (\ref{LP3})-(\ref{LP3b}) for $G'$ with capacities $c'(\{vw\})$, which in general can depend on $m$ and $k$. As $G'$ is finite, we can, for any $i\in\{1,...,r\}$, always find a finite number $N_i$ of directed paths $P_{s_i \to t_i}^{(ij)}$, where $j\in\{1,...,N_i\}$, from $s_i$ to $t_i$. Hence $N:=\sum_{i=1}^rN_i$ is finite. For each path $P_{s_i \to t_i}^{(ij)}$ we can assign a path-flow $f^{(ij)}\geq0$ such that for every $i\in\{1...r\}$ it holds
\begin{equation}
\sum_{i=1}^r\sum_{j=1}^{N_i}f^{(ij)}=\sum_{i=1}^r\sum_{v:\{s_iv\}\in E'}\left( f^{(i)}_{s_iv}-f^{(i)}_{vs_i}\right)=f^{\text{total}}_{\max}\left(G',\{c'(\{e\})\}_{\{e\}\in E'}\right)
\end{equation}
(see also \cite{ford1956maximal}). Let us define $f^{(ij)}_{vw}$ as in Supplementary Equation (\ref{fvwij}). As in Supplementary Equation (\ref{eq:sumfvwi}), for every edge $\{vw\} \in E'$ Supplementary Equations (\ref{A2})-(\ref{A3}) are fulfilled.
Then for each $f^{(ij)}$ there exists $\bar{n}^{(ij)}\in\mathbb{N}_0$ such that
\begin{equation}
f^{(ij)}-\frac{1}{kN}\leq \frac{\bar{n}^{(ij)}}{kN}\leq  f^{(ij)}.
\end{equation}
Let us also define $F^{(ij)}=m\bar{n}^{(ij)}$ and $F^{(i)}_{vw}=\sum_{j=1}^{N_i}F^{(ij)}\delta(ij,vw)$. As the $f^{(i)}_{vw}$ are feasible solutions of the LP given by Supplementary Equations (\ref{LP3})-(\ref{LP3b}), Supplementary Equations (\ref{A6})-(\ref{A7}) are fulfilled. It also holds
\begin{align}
\sum_{i=1}^r\sum_{v:\{s_iv\}\in E'}F^{(i)}_{s_iv}&=\sum_{i=1}^r\sum_{j=1}^{N_i}F^{(ij)}\sum_{v:\{s_iv\}\in E'}\delta(ij,s_iv)\\
&=\sum_{i=1}^r\sum_{j=1}^{N_i}F^{(ij)}\\
&\geq\sum_{i=1}^r mkN\sum_{j=1}^{N_i}\left(f^{(ij)}-\frac{1}{kN}\right)\\
&\geq mkN\left(f^{\text{total}}_{\max}\left(G',\{c'(\{e\})\}_{\{e\}\in E'}\right)-\frac{1}{k}\right).
\end{align}
Hence $\{F^{(i)}_{vw}\}_{i\in\{1...r\},\{vw\}\in E'}$ is a feasible solution of the LP given by Supplementary Equations (\ref{LP3})-(\ref{LP3b}) with capacities $\lfloor mkNc'(\{vw\})\rfloor$, providing, for any $i\in\{1,...,r\}$ concurrently, a flow of $F^{s_i\to t_i}=\sum_{j=1}^{N_i}F^{(ij)}$. As any integer flow of value $F^{(ij)}$ corresponds to $F^{(ij)}$ edge-disjoint paths from $s_i$ to $t_i$ in $G''_{\lfloor mkNc'\rfloor}$, we can conclude that, for all $i\in\{1,...,r\}$ concurrently, there are $F^{s_i\to t_i}$ edge-disjoint paths from $s_i$ to $t_i$. Further it holds
\begin{equation}
\sum_{i=1}^rF^{s_i\to t_i}=\sum_{i=1}^r\sum_{j=1}^{N_i}F^{(ij)}\geq m kN\left(f^{\text{worst}}_{\max}\left(G',\{c'(\{e\})\}_{\{e\}\in E'}\right)-\frac{1}{k}\right),
\end{equation}
finishing the proof.
\end{proof}

\begin{proof} (of Proposition \ref{ConcLowerBound})
We start with the worst-case scenario. Let $N:=\max_{i\in\{1,...,r\}}N_i$, where $N_i$ are the numbers of directed paths from $s_i$ to $t_i$ that exist in $G'$. Let further $k,m\in\mathbb{N}$. As shown in the proof of Proposition \ref{lowerbound}, we can obtain a state $\epsilon |E|$-close in trace distance to a network of Bell states $\bigotimes_{e\in E}{\Phi_e^+}^{\otimes \lfloor (m-\tilde{m}) kNp_e R_{\epsilon}^\leftrightarrow(\mc{N}^e)\rfloor}$, by using each channel $\mc{N}^e$ $\lfloor mkNp_e\rfloor$ times, which can be associated with $G''_{\lfloor (m-\tilde{m})kNc'_{R_{\epsilon}^\leftrightarrow}\rfloor}$. Here, as in the proof of Proposition \ref{lowerbound}, we have defined $\tilde{m}:=\max_{e\in E, p_e>0}\left\lceil\frac{1}{p_e}\right\rceil$. By Lemma \ref{ConcIntFlow}, for each $i\in\{1...r\}$ concurrently, there exist  
\begin{equation}
F_{\epsilon}^{s_i\to t_i}\geq (m-\tilde{m})  kN\left(f^{\text{worst}}_{\max}\left(G',\{c'_{R_{\epsilon}^\leftrightarrow}(\{wv\},p_{vw},p_{wv})\}_{\{vw\}\in E'}\right)-\frac{1}{k}\right)
\end{equation}
edge-disjoint paths from $s_i$ to $t_i$, corresponding to chains of Bell states. By means of entanglement swapping, we can connect these chains, providing us with concurrent rates of entanglement generation between each $s_i$ and $t_i$ of

\begin{equation}\label{rate}
\frac{F_{\epsilon}^{s_i\to t_i}}{m kN}\geq \left(1-\frac{\tilde{m}}{m}\right) f^{\text{worst}}_{\max}\left(G',\{c'_{R_{\epsilon}^\leftrightarrow}(\{wv\},p_{vw},p_{wv})\}_{\{vw\}\in E'}\right)-\frac{1}{k}.
\end{equation}
Going to the limit $m\to\infty$ and $\epsilon \to0$, the rates $R_{\epsilon}^\leftrightarrow(\mc{N}^e)$ reach the two-way assisted quantum capacities $Q^\leftrightarrow(\mc{N}^e)$. Hence we have
\begin{equation}
 \mathcal{Q}^{\text{worst}}_{\{p_e\}_{e\in E}}\left(G,\{\mc{N}^e\}_{e\in E}\right)\geq\lim_{\epsilon\to0}\lim_{m\to\infty}\min_{i\in\{1...r\}}\frac{F_{\epsilon}^{s_i\to t_i}}{m kN}\geq   f^{\text{worst}}_{\max}(G',\{c'_{Q^{\leftrightarrow}}(\{vw\},p_{vw},p_{wv})\}_{\{vw\}\in E'})-\frac{1}{k}.
\end{equation}
Taking the limit $k\to\infty$ finishes the proof of the worst-case scenario. The total throughput case works analogously: We define $N:=\sum_{i=1}^rN_i$ and use Lemma \ref{TpIntFlow} to show that in $G''_{\lfloor (m-\tilde{m}) kNc'_{R_{\epsilon}^\leftrightarrow}\rfloor}$ there exist $F_{\epsilon}^{s_i\to t_i}$ edge-disjoint paths for each source sink pair $(s_i,t_i)$ such that 
\begin{equation}
\sum_{i=1}^rF_{\epsilon}^{s_i\to t_i}\geq (m-\tilde{m})  kN\left(f^{\text{total}}_{\max}\left(G',\{c'_{R_{\epsilon}^\leftrightarrow}(\{wv\},p_{vw},p_{wv})\}_{\{vw\}\in E'}\right)-\frac{1}{k}\right).
\end{equation}
Connecting the corresponding chains of Bell states by entanglement swapping, we can obtain the rate
\begin{equation}\label{ratetp}
\sum_{i=1}^r\frac{F_{\epsilon}^{s_i\to t_i}}{m kN}\geq \left(1-\frac{\tilde{m}}{m}\right)  f^{\text{total}}_{\max}\left(G',\{c'_{R_{\epsilon}^\leftrightarrow}(\{wv\},p_{vw},p_{wv})\}_{\{vw\}\in E'}\right)-\frac{1}{k}.
\end{equation}
Taking all the limits completes the proof.
\end{proof}

Again, we can include the optimization over usage frequencies into the optimization, yielding the following:
\begin{corollary}
In a network described by a graph $G$ with associated undirected graph $G'$ it holds in a scenario of $r$ user pairs $(s_1,t_1),...,(s_r,t_r)$,

\begin{align}
&\mathcal{Q}^{\text{\emph{worst}}}\left(G,\{\mc{N}^e\}_{e\in E}\right)\geq \bar{f}^{\text{\emph{worst}}}_{\max}(G',\{Q^{\leftrightarrow}(\mc{N}^e)\}_{e\in E}),\\
&\mathcal{Q}^{\text{\emph{total}}}\left(G,\{\mc{N}^e\}_{e\in E}\right)\geq \bar{f}^{\text{\emph{total}}}_{\max}(G',\{Q^{\leftrightarrow}(\mc{N}^e)\}_{e\in E}),
\end{align}
where $\bar{f}^{\text{\emph{worst}}}_{\max}$ and $\bar{f}^{\text{\emph{total}}}_{\max}$ are given by adding an optimization over usage frequencies to the LPs given by Supplementary Equation (\ref{LP4})-(\ref{LP4c}) and (\ref{LP3})-(\ref{LP3b}), respectively.
\end{corollary}
%%%%%%%%%%%%%%%%%%%%%%%%%%%%%%%%%%%%%

As for the scenario of weighted user pairs, let us note that by Corollary \ref{Corr2}, we can for every user pair $(s_i,t_i)$ alone, for $i=1,...,r$, achieve a rate 
\begin{equation}\label{AchievePolytope1}
R_i^{m,\epsilon}\geq_{\epsilon} \bar{f}^{s_i\to t_i}_{\max}\left(G',\{Q^\leftrightarrow(\mc{N}^e) \}_{e\in E}\right),
\end{equation}
where the r.h.s. given by the LP defined by Supplementary Equations (\ref{LP2})-(\ref{LP2b}). In the space $\mathbb{R}_+^r$ of single-pair rates, this corresponds to a set of asymptotically achievable points 
\begin{equation}
\mc{A}:=\left\{\left(\bar{f}^{s_1\to t_1}_{\max},0,0,\cdots,0\right),\left(0,\bar{f}^{s_2\to t_2}_{\max},0,\cdots,0\right),\cdots,\left(0,0,\cdots,0,\bar{f}^{s_r\to t_r}_{\max}\right)\right\}.
\end{equation}
By means of \emph{time-sharing} (see e.g. \cite{el2011network}) between the single-pair protocols, we can achieve every point in the convex hull $\text{Conv}(\mc{A})$ of $\mc{A}$, \cite{leung2010quantum}. This provides us with the following result:

\begin{proposition}\label{multiunilower3}
In a network described by a graph $G$ with associated undirected graph $G'$ and a scenario of $r$ user pairs $(s_1,t_1),...,(s_r,t_r)$ with weights $q_1,...,q_r$ it holds
\begin{equation}
\mc{P}^{q_1,...,q_r}\left(G,\{\mc{N}^e\}_{e\in E}\right)\geq\max_{(\hat{R}_1,...,\hat{R}_r)\in\text{\emph{Conv}}(\mc{A})}\sum_{i=1}^rq_i\hat{R}_i,
\end{equation}
which is a linear program.
\end{proposition}

Let us note again that instead of the weighted sum of rates $\sum_{i=1}^rq_i{R}_i$, we could also maximize a general concave target function $f({R}_1,...,{R}_r)$ over a polytope, which would still a convex optimization problem.

%%%%%%%%%%%%%%%%%%%%%%%%%%%%%%

\subsection{Multipartite Target States}\label{sec:multi}
In this section we present our results on the distribution of multipartite entanglement. Let us consider a set of disjoint users $S=\{s_1,...,s_l\}$, who wish to establish a multipartite target state $\theta^d_{X_{s_1}...X_{s_l}}$, such as a Greenberger-Horne-Zeilinger (GHZ) state \cite{greenberger1989going}
\begin{equation}\label{eq:GHZ}
\ket{\Phi^{\text{GHZ},d}}_{M_{s_1}...M_{s_l}}=\frac{1}{\sqrt{d}}\sum_{i=0}^{d-1}\ket{i}_{M_{s_1}}\otimes\cdots\otimes\ket{i}_{M_{s_l}}
\end{equation}
or a multipartite private state \cite{augusiak2009multipartite},
\begin{equation}
\gamma^d_{K_{s_1}S_{s_1}...K_{s_l}S_{s_l}}=U^{\text{twist}}|\Phi^{\text{GHZ},d}\>\<\Phi^{\text{GHZ},d}|_{K_{s_1}...K_{s_l}}\otimes\sigma_{S_{s_1}...S_{s_l}}U^{\text{twist}\dagger},
\end{equation}
where $\sigma_{S_{s_1}...S_{s_l}}$ is an arbitrary state and the unitary $U^{\text{twist}}$ is given by
\begin{equation}
U^{\text{twist}}=\sum_{i_1,\cdots,i_l}|i_1,\cdots,i_l\>\<i_1,\cdots,i_l|_{K_{s_1}...K_{s_l}}\otimes U^{(i_1,\cdots,i_l)}_{S_{s_1}...S_{s_l}}.
\end{equation}
The corresponding capacities are defined analogously to \ref{sec:bip} as
\begin{equation}\label{multicap1}
\mathcal{C}^{\theta}_{\{p_e\}_{e\in E}}\left(G,\{\mc{N}^e\}_{e\in E}\right)=\lim_{\epsilon\to0}\lim_{n\to\infty}\sup_{\Lambda(n,\epsilon, \{p_e\}_{e\in E})}\left\{\frac{\langle\log d^{(k)}\rangle_{k}}{n}:\left\|\rho_{X_{s_1}...X_{s_l}}^{(n,k)}-\theta^{d^{(k)}}_{X_{s_1}...X_{s_l}}\right\|_1\leq\epsilon\right\},
\end{equation}
where the supremum is over all adaptive $(n,\epsilon, \{p_e\}_{e\in E})$ protocols $\Lambda$ and $k=(k_1,  \ldots  ,k_{m+1})$ is a vector of outcomes of the $m+1$ LOCC rounds in $\Lambda$, the averaging is over all those outcomes and $\rho_{X_{s_1} \ldots X_{s_l}}^{(n,k)} $ is the final state of $\Lambda$ for given outcomes $k$. Again we can optimize over usage frequencies, resulting in
\begin{equation}\label{multicap2}
\mathcal{C}^{\theta}\left(G,\{\mc{N}^e\}_{e\in E}\right)=\max_{p_e\geq0,\sum_ep_e=1}\mathcal{C}_{\{p_e\}_{e\in E}}^{\theta}\left(G,\{\mc{N}^e\}_{e\in E}\right).
\end{equation}
If the target is an $l$ party GHZ or private state among users in $S$, we also use the notation $\mc{Q}^S:=\mc{C}^{\Phi^\text{GHZ}}$ and $\mc{P}^S:=\mc{C}^\gamma$.  As the class of multipartite private states includes GHZ states, the multipartite private capacity is an upper bound on the multipartite quantum capacity. 

\subsubsection{Upper bounding the capacity}
We are now ready to provide our linear-program upper bound. As a special case of Corollary 4 in \cite{bauml2017fundamental} that for an $(n,\epsilon, \{p_e\}_{e\in E})$ multipartite key generation protocol yielding, among $S=\{s_1,...,s_l\}$, a state $\rho_{X_{s_1}...X_{s_l}}^{(n,k)}$ such that  $\left\|\rho_{X_{s_1}...X_{s_l}}^{(n,k)}-\gamma^{d^{(k)}}_{K_{s_1}S_{s_1}...K_{s_l}S_{s_l}}\right\|_1\leq\epsilon$, it holds \begin{equation}\label{Corr4}
\langle\log d^{(k)}\rangle_{k}\leq\min_{V_S}\frac{1}{1-b\epsilon}\left(\sum_{e\in E:\{e\} \in\partial(V_S)} np_e E_{\sq}(\mc{N}^e)+g(\epsilon)\right),
\end{equation}
where $\min_{V_S}$ is a minimization over all $V_S\subset V$ such that there is at least one pair of vertices $s_i,s_j\in S$ with $s_i\in V_S$ and $s_j\in V\setminus V_S$. Further, $b>0$, $g(\epsilon)\to0$ as $\epsilon\to0$. In the Methods section we have shown that
\begin{equation}\label{minSCutMaxSflow}
\min_{V_S}\sum_{\{vw\}\in \partial(V_S)}c'(\{vw\})=f^{S}_{\max}\left(G',\{c'(\{vw \}) \}_{\{vw\}\in E'}\right),
\end{equation}
where the r.h.s. is given by the LP
\begin{align}
f^{S}_{\max}\left(G',\{c'(\{vw \}) \}_{\{vw\}\in E'}\right)=\textrm{max} &\ f\label{LP5}\\
\forall i,j> i:\ &\ f-\sum_{v:\{s_iv\}\in E'} \left(f^{(ij)}_{s_iv}-f^{(ij)}_{vs_i}\right)\leq0\label{LP5a}\\
\forall i,j> i,\{vw\}\in E':\ &\  f^{(ij)}_{vw}+f^{(ij)}_{wv}\le  c'(\{vw\})\label{LP5b}\\
\forall i,j> i,\forall \{vw\}\in E':\ &\  f^{(ij)}_{vw},f^{(ij)}_{wv}\geq 0\label{LP5c}\\
\forall i,j> i,\ \forall w\in V,w\neq s_i,s_j:\ &\ \sum_{v:\{vw\}\in E'} \left(f^{(ij)}_{vw} -f^{(ij)}_{wv}\right)=0.\label{LP5d}
\end{align}
In the LP given by Supplementary Equations (\ref{LP5})-(\ref{LP5d}) we compute the $S$-connectivity using by finding the smallest max-flow for any disjoint pairs $s_i,s_j\in S$. As the direction of the flow does not matter we can w.l.o.g. assume that $j>i$. We do this by maximizing variable $f$ such that $f$ is a lower bound for the flows between all disjoint pairs $s_i,s_j\in S$ with $j>i$, as expressed in the constraint given by Supplementary Equation (\ref{LP5a}). The constraints given by Supplementary Equations (\ref{LP5b}) and (\ref{LP5d}) are simply the edge capacity and flow conservation constraints, as in the LP given by Supplementary Equations (\ref{LP2})-(\ref{LP2b}), applied to all disjoint pairs $s_i,s_j\in S$ with $j>i$.

Taking the limit $n\to\infty$ and $\epsilon\to0$, we obtain the following:

\begin{proposition}\label{multiupper}
In a network described by a graph $G$ with associated undirected graph $G'$ it holds for a set $S=\{s_1,...,s_l\}$ of users
\begin{equation}
\mc{P}^{S}_{\{p_e\}_{e\in E}}\left(G,\{\mc{N}^e\}_{e\in E}\right)\leq f^{S}_{\max}\left(G',\{c'_{E_{\sq}}(\{vw\},p_{vw},p_{wv})\}_{\{vw\}\in E'}\right),
\end{equation}
where $f^{S}_{\max}$ is given by the LP defined by Supplementary Equations (\ref{LP5}).
\end{proposition}

We can, again, include an optimization over all usage frequencies $p_e$, providing

\begin{corollary}
In a network described by a graph $G$ with associated undirected graph $G'$ it holds in a scenario of a multipartite user group $S$,
\begin{equation}
\mathcal{P}^{S}\left(G,\{\mc{N}^e\}_{e\in E}\right)\leq\bar{f}^{S}_{\max}(G',\{E_{\sq}(\mc{N}^e)\}_{e\in E}),
\end{equation}
where the r.h.s. is given by the linear program defined by Supplementary Equation (\ref{LP5}) with added optimization over usage frequencies.
\end{corollary}

\subsubsection{Lower bounding the capacity}
As our last result, we can obtain the following lower bound on $\mc{Q}^{S}_{\{p_e\}_{e\in E}}$:
\begin{proposition}\label{MultiLower}
In a network described by a graph $G$ with associated undirected graph $G'$ it holds for a set $S=\{s_1,...,s_r\}$ of users
\begin{equation}
\mc{Q}^{S}_{\{p_e\}_{e\in E}}\left(G,\{\mc{N}^e\}_{e\in E}\right)\geq \frac{1}{2} f^{S}_{\max}\left(G',\{c'_{Q^{\leftrightarrow}}(\{vw\},p_{vw},p_{wv})\}_{\{vw\}\in E'}\right),
\end{equation}
with $f^{S}_{\max}$ given by the LP defined by Supplementary Equation (\ref{LP5}).
\end{proposition}

Before proving Proposition \ref{MultiLower}, we need the following:
\begin{lemma}\label{Lemma:MultiLower}
Let us assume we have a finite undirected graph $G'$ with capacities $c'(\{vw\})$ (that can in general depend on $m$ and $k$) and a set $S=\{s_1,...,s_r\}$ of users. Let $N:=\max_{i,j\in\{1,...,r\},i\neq j}N_{ij}$, where $N_{ij}$ are the numbers of directed paths from $s_i$ to $s_j$ that exist in $G'$. Let further $k,m\in\mathbb{N}$. Then it holds
\begin{equation}
\lambda_S(G''_{\lfloor 2mkNc'\rfloor})\geq 2mkN\left(f^{S}_{\max}\left(G',\{c'(\{vw \}) \}_{\{vw\}\in E'}\right)-\frac{1}{k} \right),
\end{equation}
where  $\lambda_S(G''_{\lfloor 2mkNc'\rfloor})$ is even, $G''_{\lfloor 2mkNc'\rfloor}$ is an undirected multigraph with unit capacities as introduced above and $f^{S}_{\max}\left(G',\{c'(\{vw \}) \}_{\{vw\}\in E'}\right)$ is the solution of the LP given by Supplementary Equation (\ref{LP5}). 
\end{lemma}

\begin{proof}
Let $m,k\in\mathbb{N}$ and let $\{f^{(ij)}_{vw}\}_{i,j\in\{1...r\},i\neq j,\{vw\}\in E'}$ be the set of edge flows maximizing the LP given by Supplementary Equation (\ref{LP5}) for $G'$  with capacities $c'(\{vw\})$, which in general can depend on $m$ and $k$. As $G'$ is finite, we can, for any $i,j\in\{1,...,r\}$ such that $i\neq j $ always find a finite number $N_{ij}$ of directed paths $P_{s_i \to s_j}^{(ijl)}$ from $s_i$ to $s_j$. Hence $N:=\max_{i,j\in\{1,...,r\},i\neq j}N_{ij}$ is finite. For each path $P_{s_i \to s_j}^{(ijl)}$ we can assign a path-flow $f^{(ijl)}\geq0$ such that for every disjoint $i,j\in\{1...r\}$ it holds
\begin{equation}
\sum_{l=1}^{N_{ij}}f^{(ijl)}=\sum_{v:\{s_iv\}\in E'}\left( f^{(ij)}_{s_iv}-f^{(ij)}_{vs_i}\right)\geq f^{S}_{\max}\left(G',\{c'(\{e\})\}_{\{e\}\in E'}\right)
\end{equation}
(See also \cite{ford1956maximal}). Let us define in analogy to Supplementary Equation (\ref{eq:fvwi})
\begin{equation}
f^{(ijl)}_{vw}=\begin{cases}f^{(ijl)}\text{ if }vw\in P_{s_i \to s_j}^{(ijl)}\\0\text{ else.}\end{cases}
\end{equation}
As in Supplementary Equation (\ref{eq:sumfvwi}), it holds for every edge $\{vw\} \in E'$ that
\begin{align}
&f^{(ij)}_{wv}\geq\sum_{l=1}^{N_{ij}}f^{(ijl)}_{wv}=\sum_{j=1}^{N_{ij}}f^{(ijl)}\delta(ijl,wv),\\
&f^{(ij)}_{vw}\geq\sum_{l=1}^{N_{ij}}f^{(ijl)}_{vw}=\sum_{j=1}^{N_{ij}}f^{(ijl)}\delta(ijl,vw),
\end{align}
where 
\begin{equation}
\delta(ijl,uv)=\begin{cases}
1 \text{ if }uv\in P^{(ijl)}_{s_i\to s_j}\\ 0\text{ else.}
\end{cases}
\end{equation}
Then for each $f^{(ijl)}$ there exists $\bar{n}^{(ijl)}\in\mathbb{N}_0$ such that
\begin{equation}
f^{(ijl)}-\frac{1}{kN}\leq \frac{\bar{n}^{(ijl)}}{kN}\leq  f^{(ijl)}.
\end{equation}
Let us also define even integers $F^{(ijl)}=2m\bar{n}^{(ijl)}$ and $F^{(ij)}_{vw}=\sum_{l=1}^{N_{ij}}F^{(ijl)}\delta(ijl,vw)$. As the $f^{(ij)}_{vw}$ are feasible solutions of the LP given by Supplementary Equation (\ref{LP5}), it holds for any disjoint pair $i,j\in\{1,...,r\}$ and any edge $\{vw\}\in E'$ that
\begin{equation}
\left(F^{(ij)}_{wv}+F^{(ij)}_{vw}\right)\leq \lfloor 2mkN( f^{(ij)}_{wv}+f^{(ij)}_{vw})\rfloor \leq \lfloor 2mkNc'(\{vw\})\rfloor.
\end{equation}
Further, for all disjoint $i,j\in\{1...r\}$, it holds for all $w\in V$ with $w\neq s_i,s_j$ that
\begin{equation}
\sum_{v:\{vw\}\in E'}\left( F^{(ij)}_{vw} - F^{(ij)}_{wv}\right)=\sum_{l=1}^{N_{ij}}F^{(ijl)}\sum_{v:\{vw\}\in E'}\left(\delta(ijl,vw)  -\delta(ijl,wv)\right)=0,
\end{equation}
where we use the same argument as explained after Supplementary Equation (\ref{FlowCons}). It also holds
\begin{align}
\sum_{v:\{s_iv\}\in E'}F^{(ij)}_{s_iv}&=\sum_{l=1}^{N_{ij}}F^{(ijl)}\sum_{v:\{s_iv\}\in E'}\delta(ijl,s_iv)\\
&=\sum_{l=1}^{N_{ij}}F^{(ijl)}\\
&\geq 2mkN\sum_{l=1}^{N_{ij}}\left(f^{(ijl)}-\frac{1}{kN}\right)\\
&\geq 2mkN\left(f^{S}_{\max}\left(G',\{c'(\{e\})\}_{\{e\}\in E'}\right)-\frac{1}{k}\right).
\end{align}
Hence $\{F^{(ij)}_{vw}\}_{i,j\in\{1...r\},i\neq j, \{vw\}\in E'}$ is a feasible solution of the LP given by Supplementary Equation (\ref{LP5}) with capacities $\lfloor 2mkNc'(\{vw\})\rfloor$, providing, for any disjoint $i,j\in\{1,...,r\}$ a flow of
\begin{equation}
F^{s_i\to s_j}=\sum_{l=1}^{N_{ij}}F^{(ijl)}\geq 2m kN\left(f^{S}_{\max}\left(G',\{c'(\{e\})\}_{\{e\}\in E'}\right)-\frac{1}{k}\right).
\end{equation}
As any integer flow of value $F^{(ijl)}$ corresponds to $F^{(ijl)}$ edge-disjoint paths in $G''_{\lfloor 2mkNc'\rfloor}$, we can conclude that there are $F^{s_i\to s_j}$ edge-disjoint paths for every disjoint pair $i,j\in\{1,...,r\}$. Application of Menger's Theorem \cite{menger1927allgemeinen}, which is  the integer version of the max-flow min-cut theorem, finishes the proof.
\end{proof}

\begin{proof} (of Proposition \ref{MultiLower})
Let $N:=\max_{i,j\in\{1,...,r\},i\neq j}N_{ij}$, where $N_{ij}$ are the numbers of directed paths from $s_i$ to $s_j$ that exist in $G'$. Let further $\epsilon>0$ and $k,m\in\mathbb{N}$. As shown in the proof of Proposition \ref{lowerbound}, we can obtain a state $\epsilon |E|$-close in trace distance to a network of Bell states  $\bigotimes_{e\in E}{\Phi_e^+}^{\otimes \lfloor 2(m-\tilde{m}) kNp_e R_{\epsilon}^\leftrightarrow(\mc{N}^e)\rfloor}$, by using each channel $\mc{N}^e$ $\lfloor2 mkNp_e\rfloor$ times, which corresponds to an undirected unit-capacity multigraph $G''_{\lfloor 2(m-\tilde{m}) kNc'_{R_{\epsilon}^\leftrightarrow}\rfloor}$. By Lemma \ref{Lemma:MultiLower} and equation (38) in the Methods section, with $g_1=\frac{1}{2}$ and $g_2=\frac{|V\setminus S|}{2}+1$ as in \cite{petingi2009packing}, $G''_{\lfloor 2(m-\tilde{m}) kNc'_{R_{\epsilon}^\leftrightarrow}\rfloor}$ contains 
\begin{equation}
t_S\left(G''_{\left\lfloor2 (m-\tilde{m}) kNc'_{R_{\epsilon}^\leftrightarrow}\right\rfloor}\right)\geq (m-\tilde{m})kN\left(f^{S}_{\max}\left(G',\{c'_{R_{\epsilon}^\leftrightarrow}(\{vw\},p_{vw},p_{wv})\}_{\{vw\}\in E'}\right)-\frac{1}{k} \right)-g_2
\end{equation}
edge-disjoint $S$-trees. The Bell states forming an $S$-tree can be transformed into a qubit GHZ state among all vertices in the set $S$ by means of the following protocol: In a first step all Bell states are merged into a GHZ state among all nodes in the $S$-tree. This can be done by means of  projective measurement and Pauli corrections \cite{wallnofer20162d}. All the unwanted parties in the GHZ state can be removed by means projective measurements \cite{wallnofer20162d}, leaving only a GHZ state among the nodes in $S$. See also a related work by \cite{yamasaki2017graph}. Hence we can obtain the following rate:
\begin{equation}\label{GHZrate}
\frac{t_S\left(G''_{\left\lfloor 2(m-\tilde{m}) kNc'_{R_{\epsilon}^\leftrightarrow}\right\rfloor}\right)}{2m kN}\geq   \frac{1}{2}\left(1-\frac{\tilde{m}}{m}\right)\left(f^{S}_{\max}\left(G',\{c'_{R_{\epsilon}^\leftrightarrow}(\{vw\},p_{vw},p_{wv})\}_{\{vw\}\in E'}\right)-\frac{1}{k} \right) -\left(1-\frac{\tilde{m}}{m}\right)\frac{g_2}{2mkN}.
\end{equation}
Going to the limit $m\to\infty$ and $\epsilon \to0$, the rates $R_{\epsilon}^\leftrightarrow(\mc{N}^e)$ reach the two-way assisted quantum capacities $Q^\leftrightarrow(\mc{N}^e)$. Hence, as $g_2$ is finite,
\begin{equation}
 \mathcal{Q}^S_{\{p_e\}_{e\in E}}\left(G,\{\mc{N}^e\}_{e\in E}\right)\geq\lim_{\epsilon\to0}\lim_{m\to\infty}\frac{t_S\left(G''_{\left\lfloor 2(m-\tilde{m}) kNc'_{R_{\epsilon}^\leftrightarrow}\right\rfloor}\right)}{2m kN}\geq \frac{1}{2}\left(f^{S}_{\max}\left(G',\{c'_{Q^\leftrightarrow}(\{vw\},p_{vw},p_{wv})\}_{\{vw\}\in E'}\right)-\frac{1}{k} \right).
\end{equation}
Taking the limit $k\to\infty$ finishes the proof.
\end{proof}

Optimization over usage frequencies provides us with the following:
\begin{corollary}
In a network described by a graph $G$ with associated undirected graph $G'$ it holds in a scenario of a multipartite user group $S$,
\begin{equation}
\mathcal{Q}^S\left(G,\{\mc{N}^e\}_{e\in E}\right)\geq \frac{1}{2}\bar{f}^{S}_{\max}(G',\{Q^{\leftrightarrow}(\mc{N}^e)\}_{e\in E}),
\end{equation}
where the r.h.s. is given by the linear program given by Supplementary Equation (\ref{LP5}) with added optimization over usage frequencies.
\end{corollary}

\subsection{Numerical Examples}\label{sec:num}

\begin{figure}

\centering
\includegraphics[width=\textwidth]{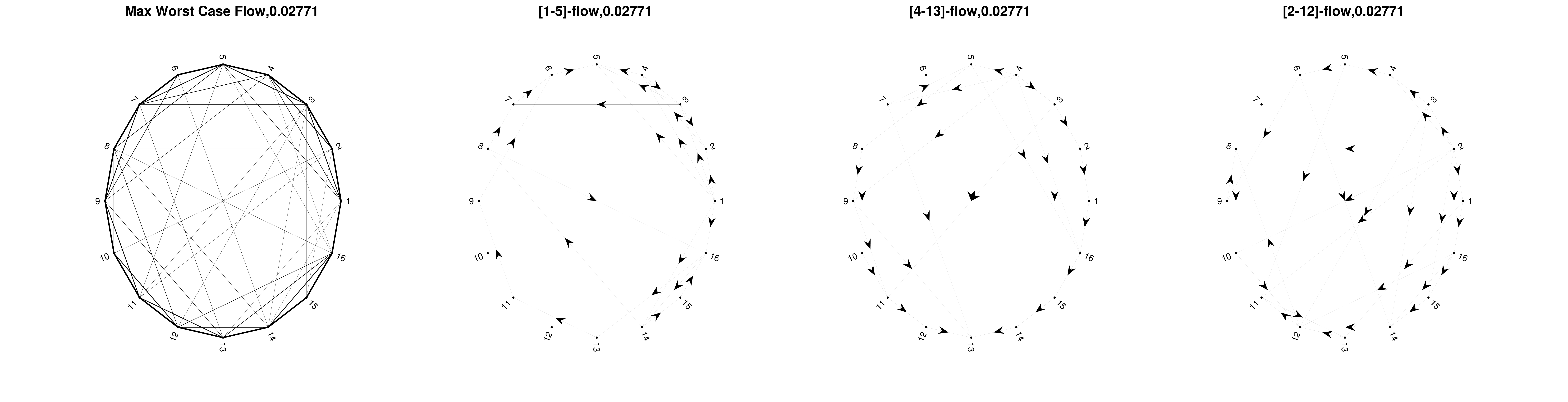}
\caption{Example of a multi-commodity flow instance, maximizing the worst-case flow, the LP given by Supplementary Equations (\ref{LP4})-(\ref{LP4c}), in a chord network with $l=4$ without optimization of the usage frequencies. On the left, one can see the network, with linewidths corresponding to its capacities. The other plots with legend above the plot `$[a-b]-$flow, $x$' show the respective edge flows between the pair of users $a$ and $b$ with flow value $x$, that were obtained in the optimization.}\label{fig:chord1}

\
\includegraphics[width=\textwidth]{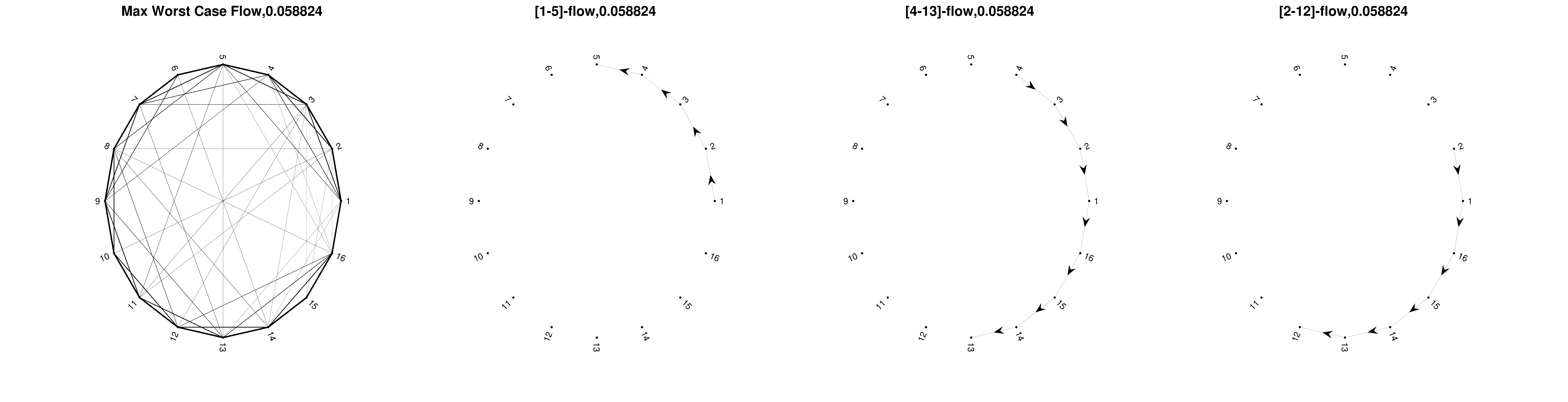}

\caption{Same network as in Supplementary Figure \ref{fig:chord1} with optimization over usage frequencies.}\label{fig:chord1a}

%\end{figure}

%\begin{figure}
%\centering
\

\includegraphics[width=\textwidth]{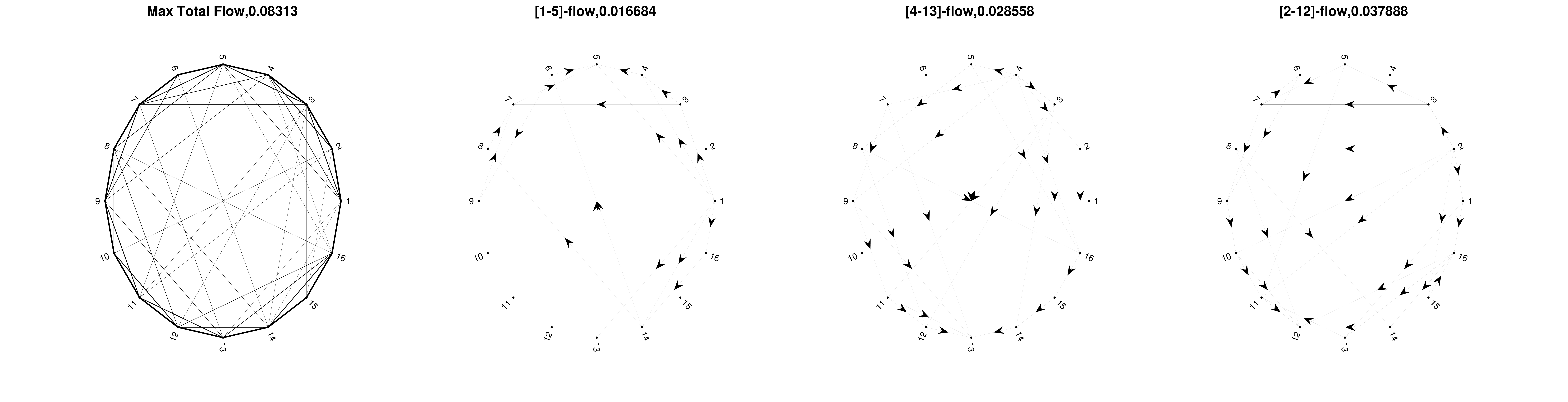}%, trim={1cm 7cm 1cm 5cm},clip=true
\caption{Example of a multi-commodity flow instance, maximizing the total multi-commodity flow, the LP given by Supplementary Equations (\ref{LP3})-(\ref{LP3b}) without optimization of the usage frequencies.}\label{fig:chord2}

\
\includegraphics[width=\textwidth]{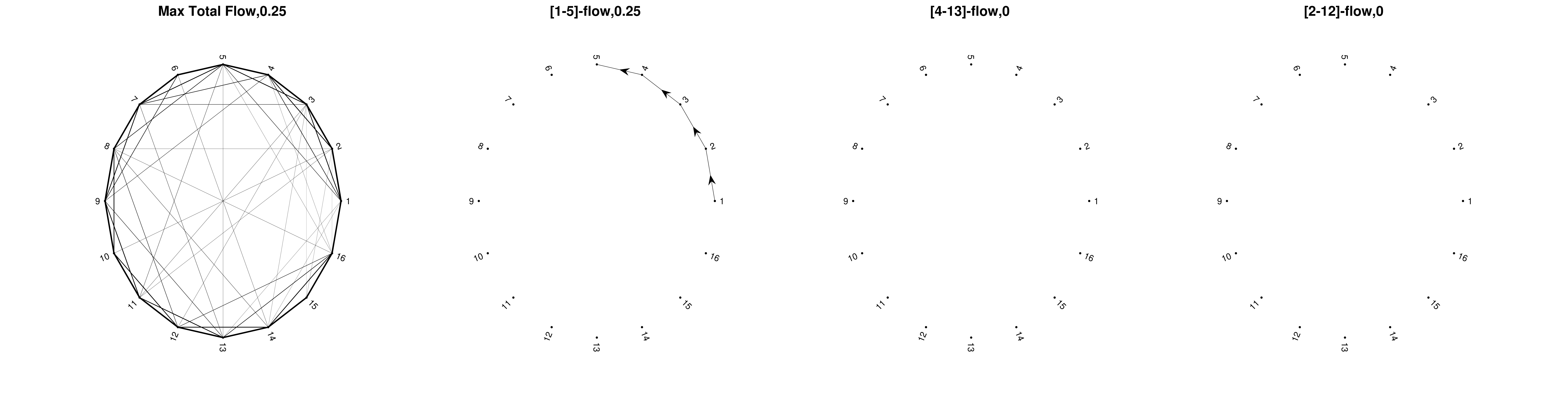}

\caption{Same network as in Supplementary Figure \ref{fig:chord2} with optimization over usage frequencies.}\label{fig:chord2a}
\end{figure}

As a proof of principle demonstration, we considered a chord network, which is a model for classical peer-to-peer networks \cite{stoica2003chord}. A chord network consists of $N=2^l$ nodes $v_0,...,v_{N-1}$ arranged in a circle and connected by edges $\{ e_{\rm circle}\}$ with some constant capacity $c'(\{e_{\text{circle}}\})=c_0$.  In addition, there are diagonal edges $\{e_{\rm diagonal}\}$ connecting randomly chosen nodes $v_n$ and $v_{n + m_i  \mod N}$.  The index $m_i$ is an integer randomly chosen out of the interval $\left[2^{i-1},2^{i}\right]$ for $1\leq i\leq l$. The capacity of these edges decrease with distance, which we model as $c'(\{e_{\text{diagonal}}\})=\frac{c_0}{|m_i|}$. 

In the case of the network consisting of lossy optical channels with transmissivity $\eta$, such that there is one channel $\mc{N}_{\eta}^e$ for each undirected edge $\{e_{\rm circle}\}$ in the circle, i.e. $|E|=|E'|$, and a flooding protocol with constant user frequencies $|E|^{-1}$ one can set $c_0=|E|^{-1}Q^\leftrightarrow(\mc{N}_{\eta}^e)=-|E|^{-1}\log(1-\eta)$ \cite{PLOB17}. As lossy optical channels are distillable, by Theorems \ref{multiuni:upper} and \ref{ConcLowerBound}, we can obtain upper and lower bounds on the worst-case and total quantum capacities by computing the LPs given by Supplementary Equations (\ref{LP4})-(\ref{LP4c}) and (\ref{LP3})-(\ref{LP3b}), respectively. 

Supplementary Figure \ref{fig:chord1} shows an example of a multi-commodity flow instance, maximizing the worst-case multi-commodity flow given by Supplementary Equations (\ref{LP4})-(\ref{LP4c}) in a  chord network with $l=4$. Supplementary Figure \ref{fig:chord2} shows the total multi-commodity flow maximizing  Supplementary Equations (\ref{LP3})-(\ref{LP3b}). 

Comparing Supplementary Figures \ref{fig:chord1} and \ref{fig:chord2} with Supplementary Figures \ref{fig:chord1a} and \ref{fig:chord2a}, respectively, one can see the effect of the the optimization over usage frequencies: For fixed usage frequencies, as in Supplementary Figures \ref{fig:chord1} and \ref{fig:chord2}, the capacities of the circle edges are exhausted and diagonal edges have to be used. When we optimize over usage frequencies, as in Supplementary Figures \ref{fig:chord1} and \ref{fig:chord2}, the capacities of the circle edges are increased and no use of diagonal edges is necessary.

\acknowledgements
We would like to thank Bill Munro, Simone Severini, Hayata Yamasaki, Kenneth Goodenough, Kaushik Chakraborty and Stephanie Wehner for insightful discussions. This work was supported by the Netherlands Organization for Scientific Research (NWO/OCW), as part of the Quantum Software Consortium program (project number 024.003.037 / 3368) and an NWO Vidi grant. K.A. thanks support from JST, PRESTO Grant Number JPMJPR1861. S.B. acknowledges support from the Spanish MINECO (QIBEQI FIS2016-80773-P, Severo Ochoa SEV-2015-0522), Fundacio Cellex, Generalitat de Catalunya (SGR 1381 and CERCA Programme) as well as from the European Union's Horizon 2020 research and innovation programme, grant agreement number 820466 (project CiViQ).

\newpage

\bibliography{Network}

\end{document}